\documentclass[runningheads]{llncs}

\usepackage[usenames]{color}
\usepackage{amsmath}
\usepackage{stmaryrd}
\usepackage{thm-restate}
\usepackage{colortbl}
\usepackage[dvipsnames]{xcolor}
\usepackage{amssymb}
\usepackage{paralist}
\usepackage{graphicx}
\usepackage{grffile}
\usepackage{epstopdf}
\usepackage{wrapfig}
\usepackage{stackrel}
\usepackage{adjustbox}
\usepackage{float}
\usepackage{booktabs}
\usepackage{subcaption}
\usepackage{tikz}
\usepackage{todonotes}
\usepackage{url}
\usepackage[capitalise]{cleveref}

\usepackage{newtxmath}
\usepackage{newtxtext}
\DeclareMathAlphabet{\mathcalligra}{T1}{calligra}{m}{n}
\usetikzlibrary{arrows,automata}
\usetikzlibrary{calc}
\usetikzlibrary{automata,positioning,trees,shapes,petri,arrows,snakes,backgrounds,calc,fit,arrows.meta,shadows,math,decorations.markings}
\usetikzlibrary{decorations.pathmorphing}
\usetikzlibrary{shapes.multipart,arrows,automata}
\usetikzlibrary{shapes.geometric}
\usetikzlibrary{positioning}
\usetikzlibrary{calc,patterns,patterns.meta}
\usetikzlibrary{fit}

\usepackage[linesnumbered,noend,ruled]{algorithm2e}

\InputIfFileExists{cutmargins.tex}{%
  \pdfpagesattr{/CropBox [125 82 490 688]} 
}{%
}

\SetKwInput{KwInput}{Input}
\SetKw{Continue}{continue}

\newcounter{claimcounter}
\renewenvironment{claim}{\refstepcounter{claimcounter}{\medskip\noindent \underline{Claim \theclaimcounter:}}\itshape}{\smallskip}
\crefname{claimcounter}{Claim}{Claims}

\Crefname{algocf}{Algorithm}{Algorithms}

\SetKwComment{Comment}{$\triangleright$\ }{}
\SetCommentSty{mycommfont}

\usepackage{trimclip}

\makeatletter
\DeclareRobustCommand{\shortto}{%
  \mathrel{\mathpalette\short@to\relax}%
}

\DeclareRobustCommand{\shortminus}{%
  \mathrel{\mathpalette\short@minus\relax}%
}

\newcommand{\short@to}[2]{%
  \mkern2mu
  \clipbox{{.5\width} 0 0 0}{$\m@th#1\vphantom{+}{\rightarrow}$}%
}

\newcommand{\short@minus}[2]{%
  \mkern2mu
  \clipbox{{.5\width} 0 0 0}{$\m@th#1\vphantom{+}{-}$}%
}
\makeatother

\newcommand{\labeledto}[1]{{{\shortminus}\hspace{-2pt}\raisebox{0.16ex}{$\scriptstyle\{ #1\hspace{-0.28pt}\}$}\hspace{-2.2pt}{\shortto}}}

\newcommand{\scriptlabeledto}[1]{{{\shortminus}\hspace{-1.0pt}\raisebox{0.12ex}{$\scriptscriptstyle\{ #1\hspace{-0.28pt}\}$}\hspace{-1.6pt}{\shortto}}}

\newcommand\move[3]{
\mathchoice
{#1\,\labeledto{#2}\,#3}
{#1\labeledto{#2}#3}
{#1\scriptlabeledto{#2}#3}
{#1\scriptlabeledto{#2}#3}
}

\newcommand{\vh}[1]{\textcolor{blue}{\ifmmode \text{[#1]}\else [VH: #1] \fi}}
\newcommand{\lh}[1]{\textcolor{magenta}{\ifmmode \text{[#1]}\else [LH: #1] \fi}}
\newcommand{\ol}[1]{\textcolor{green!60!black}{\ifmmode \text{[OL: #1]}\else [OL: #1] \fi}}
\newcommand{\yfc}[1]{\textcolor{blue!60!black}{\ifmmode \text{[YFC: #1]}\else [YFC: #1] \fi}}

\newcommand{\blind}[1]{\textcolor{gray}{[blinded for review]}}

\newcommand{\X}{\mathbb{X}}
\newcommand{\E}{\mathcal{E}}

\newcommand{\G}{G}

\newcommand{\St}{V}

\newcommand{\ILoveLukas}[2]{\!\tikz[baseline,anchor=base]{\node[rounded corners=1mm,inner sep=0.6mm,text height=2mm,text depth=0.5mm] {$\vphantom{A_1^2} \mathsf{#1} \! \othersubseteq \! \mathsf{#2}$};}\!}

\newcommand{\graph}{G}
\newcommand{\vertices}{V}
\newcommand{\edges}{E}
\newcommand{\vars}{\mathbb X}
\newcommand{\tightset}{\mathcal T}

\newcommand{\atoml}{\mathit {leftAtom}}
\newcommand{\atomr}{\mathit {rightAtom}}
\newcommand{\atom}{\mathit {atom}}
\newcommand{\atoms}{\mathcal A}
\newcommand{\pos}{\mathit {pos}}

\newcommand{\N}{\mathbb{N}}
\newcommand{\propagate}{\mathtt{propagate}}
\newcommand{\propagatemin}{\mathtt{propagate}_{\mathit{min}}}

\newcommand{\ass}{\nu}
\newcommand{\emptier}{\succ}
\newcommand{\lord}{<_{\mathrm{lex}}}
\newcommand{\sys}{\constr}
\newcommand{\Min}[1]{#1^{\min} }
\newcommand{\lang}{L}
\newcommand{\langof}[1]{\lang(#1)}
\newcommand{\minlangof}[1]{\Min{\lang}(#1)}

\newcommand{\smhole}{\boxtimes}

\newcommand{\aut}[0]{\mathcal{A}}
\newcommand{\autass}{\mathsf{Aut}}

\newcommand{\lnref}[1]{Line~\ref{#1}}
\newcommand{\condref}[1]{(IG\ref{#1})}
\newcommand{\igcondref}[1]{Condition~\condref{#1}}

\newcommand{\sat}{\mathtt{SAT}}
\newcommand{\unsat}{\mathtt{UNSAT}}
\newcommand{\refine}{\mathtt{refine}}
\newcommand{\scc}{\mathtt{SCC}}
\newcommand{\incl}{\mathtt{incl}}
\newcommand{\noodlify}{\mathtt{noodlify}}
\newcommand{\noodle}{\mathsf{N}}

\newcommand{\noodles}{\mathit{Noodles}}

\newcommand{\lass}{\mathsf{Lang}}
\newcommand{\constrlass}{\lass_\constr}
\newcommand{\constrautass}{\autass_\constr}

\newcommand{\nat}{\mathbb{N}}
\newcommand{\product}{\cap_\epsilon}

\newcommand{\reduce}{\mathit{minimise}}

\newcommand{\concat}{\cdot}

\newcommand{\concatxof}[1]{\circ_{#1}}

\newcommand{\concats}{\cdots}

\newcommand{\pow}[1]{\mathcal{P}(#1)}

\newcommand{\tterm}{t}
\newcommand{\sterm}{s}
\newcommand{\varsof}[1]{\mathit{Vars}(#1)}

\newcommand{\pop}{\mathtt{dequeue}}
\newcommand{\push}{\mathtt{enqueue}}
\newcommand{\worklist}[0]{\mathit{Branches}}
\newcommand{\frontier}[0]{W}

\makeatletter
\DeclareFontEncoding{LS1}{}{}
\makeatother
\DeclareFontSubstitution{LS1}{stix}{m}{n}
\DeclareSymbolFont{stixletters}{LS1}{stix}{m}{it}
\DeclareSymbolFont{stixoperators}{LS1}{stix}{m}{n}

\DeclareMathSymbol{\othersubseteq}{\mathrel}{stixletters}{"D3}
\DeclareMathSymbol{\otherin}{\mathrel}{stixoperators}{"CB}
\DeclareMathSymbol{\othereq}{\mathrel}{stixoperators}{`=}

\newcommand{\graphnode}[2]{ {\mathsf #1} \! \othersubseteq \! {\mathsf #2}}                                                                                                                                                                                                                                                                                                                                                                                                                                                                                                                                                                                                                                                                                                                                                                                                                                                                                                                                                                                                                                                                                                                                                                                                                                                                                                                                                                                                                                                                                                                                                                                                                                                                                                                                                                                                                                                                                                                                                                                                                                                                                                                                                                                                                                                                                                                                                                                                                                                                                                                                                                                                                                                                                                                                                                                                                                                                                                                                                                                                                                                                                                                                                                                                                                                                                                                                                                                                                                                                                                                                                                                                                                                                                                                                                                                                                                                                                                                                                                                                                                                                                                                                                                                                                                                                                                                                                                                                                                                                                                                                                                                                                                                                                                                                                                                                                                                                                                                                                                                                                                                                                                                                                         \renewcommand{\graphnode}[2]{\ILoveLukas{#1}{#2}}

\newcommand\constr{\Phi}
\newcommand\constreq{{\mathcal E}}

\newcommand{\lifo}[1]{\langle #1 \rangle}

\newcommand{\modulo}[0]{\mathbin{\mathrm{mod}}}

\newcommand{\claimqed}[0]{\hfill $\blacksquare$}
\newenvironment{claimproof}[1]{\par\noindent\underline{Proof:}\space#1}{\claimqed\medskip}
\newenvironment{claimproofnoqed}[1]{\par\noindent\underline{Proof:}\space#1}{}

\newcommand{\positions}[0]{P}
\newcommand{\splgraph}[0]{\mathit{SG}}
\newcommand{\varmap}[0]{\mathtt{var}}
\newcommand{\varmapof}[1]{\varmap(#1)}
\newcommand{\ubergraph}[0]{UG}
\newcommand{\sinks}[0]{\mathit{Sinks}}
\newcommand{\sources}[0]{\mathit{Sources}}

\newcommand{\Axyx}[0]{\aut_{xyx}}
\newcommand{\Azu}[0]{\aut_{zu}}
\renewcommand{\prod}{\mathit{Product}}
\newcommand{\prodex}{\Axyx\product\Azu}

\newcommand{\noodler}{\textsc{Noodler}\xspace}
\newcommand{\cvcv}{\textsc{CVC5}\xspace}
\newcommand{\cvciv}{\textsc{CVC4}\xspace}
\newcommand{\ziii}{\textsc{Z3}\xspace}
\newcommand{\ziiistriiire}{\textsc{Z3str3RE}\xspace}
\newcommand{\ziiistriv}{\textsc{Z3str4}\xspace}
\newcommand{\ziiitrau}{\textsc{Z3-Trau}\xspace}
\newcommand{\kepler}[0]{\texttt{Kepler}$_{\mathtt{22}}$\xspace}
\newcommand{\ostrich}[0]{\textsc{OSTRICH}\xspace}
\newcommand{\retro}[0]{\textsc{Retro}\xspace}
\newcommand{\sloth}[0]{\textsc{Sloth}\xspace}
\newcommand{\jsa}[0]{\textsc{JSA}\xspace}
\newcommand{\stranger}[0]{\textsc{Stranger}\xspace}
\newcommand{\trau}[0]{\textsc{Trau}\xspace}
\newcommand{\norn}[0]{\textsc{Norn}\xspace}

\newcommand{\pyex}{\textsc{PyEx}\xspace}
\newcommand{\pyexhard}{\textsc{PyEx-Hard}\xspace}
\newcommand{\kaluza}{\textsc{Kaluza}\xspace}
\newcommand{\kaluzahard}{\textsc{Kaluza-Hard}\xspace}
\newcommand{\strii}{\textsc{Str~2}\xspace}
\newcommand{\slog}{\textsc{Slog}\xspace}

\newcommand{\dualof}[1]{\mathsf{dual}(#1)}
\newcommand{\splitgraph}{\mathit{SG_\constreq}}

\newcommand{\beginexample}[0]{%
\medskip\noindent\stepcounter{example}\textit{Example~\theexample}. }

\title{
Word Equations in Synergy with Regular Constraints\\ (Technical Report)
}

\author{
  Franti\v{s}ek Blahoudek\inst{1},
  Yu-Fang Chen\inst{2},
  David Chocholat\'{y}\inst{1},\\
  Vojt\v{e}ch Havlena\inst{1},
  Luk\'{a}\v{s} Hol\'{i}k\inst{1},
  Ond\v{r}ej Leng\'{a}l\inst{1}, and
  Juraj S\'{i}\v{c}\inst{1}
}
\authorrunning{F. Blahoudek, Y. Chen, D. Chocholat\'{y}, V. Havlena, L. Hol\'{i}k, O. Leng\'{a}l, J. S\'{i}\v{c}}

\institute{
Faculty of Information Technology,
Brno University of Technology,
Czech Republic
\and
Academia Sinica, Taiwan
}

\begin{document}

\maketitle

\vspace{-6mm}

\begin{abstract}
When eating spaghetti, one should have the sauce and noodles mixed instead of
eating them separately.
We argue that also in string solving, word equations and regular
constraints are better mixed together than approached separately as in most current string solvers.
We propose a fast algorithm, complete for the fragment of chain-free
constraints, in which word equations and regular constraints are tightly
integrated and exchange information, efficiently pruning the cases generated by
each other and limiting possible combinatorial explosion.
The algorithm is based on a novel language-based characterisation of satisfiability of word equations with regular constraints.
We experimentally show that 
our prototype implementation is competitive with the best string solvers and even superior in that
it is the fastest on difficult examples
and has the least number of timeouts.  

\end{abstract}

\vspace{-8.0mm}
\section{Introduction}
\vspace{-2.0mm}

Solving of string constraints (string solving) has gained a significant traction in the last two decades,
drawing motivation from verification of programs that manipulate strings.
String manipulation is indeed ubiquitous, tricky, 
and error-prone.
It has been a source of security vulnerabilities, such as cross-site scripting or SQL injection, 
that have been occupying top spots in the lists of software security
issues~\cite{OWASP13,OWASP17,OWASP21};
moreover, widely used scripting languages like Python and PHP rely heavily on
strings.
Interesting new examples of an intensive use of critical string operations can
also be
found, e.g., in reasoning over configuration files of cloud services
\cite{hadarean_mosca} or smart contracts~\cite{AltBHS22}.
Emergent approaches and tools for string solving are already numerous, for instance 
[6-54].   

A practical solver must handle a wide range of string operations,
ranging from regular constraints and word equations across string length
constraints to complex functions such as ReplaceAll or integer-string
conversions. 
The solvers translate most kinds of constraints to a few types of basic string constraints. 
The base algorithm then determines the architecture of the string solver and is
the component with the largest impact on its efficiency.
The second ingredient of the efficiency are layers of opportunistic heuristics that are effective on established benchmarks. 
Outside the boundaries where the heuristics apply and 
the core algorithm must do a heavy lifting, the efficiency may deteriorate.

The most essential string constraints, word equations and regular constraints, are the primary source of difficulty.   
%
Their combination is PSPACE-complete \cite{plandowski99,jez2016}, decidable by
the algorithm of Makanin~\cite{makanin} and Je\.{z}'s recompression~\cite{jez2016}.
Since it is not known how these general algorithms may be implemented
efficiently, string solvers use incomplete algorithms or work only with
restricted fragments (e.g. straight-line of~\cite{LB16} or chain-free~\cite{LB16,ChainFree}, which cover most of existing
practical benchmarks), but even these are still PSPACE-complete (immediately due
to Boolean combinations of regular constraints) and practically hard.
Most of string solvers use base algorithms that resemble Makanin~\cite{makanin}
or Nielsen's~\cite{nielsen1917} algorithm in which word
equations and regular constraints each generate one level of disjunctive branching, and the two levels multiply. 
%
%
Regular constraints particularly are considered complex and expensive, 
and reasoning with them is sometimes postponed and done only as the last step.

In this work, we propose an algorithm 
in which regular constraints are not avoided but tightly integrated with equations, 
enabling an exchange of information between equations and regular constraints that leads to a mutual pruning of \mbox{generated disjunctive choices.}

For instance, in cases such as $zyx = xxz \land x\in a^* \land y\in a^+b^+ \land z\in b^*$, attempting to eliminate the equation results in an infinite case split (using, e.g., Nielsen's algorithm~\cite{nielsen1917} or the algorithm of~\cite{Norn}) and it indeed leads to failure for all solvers we have tried. 
The regular constraints enforce $\unsat$: since the $y$ on the left contains at least one $b$, 
the $z$ on the right must answer with at least one $b$ ($x$ has only $a$'s).
Then, since the first letter on the left is the $b$ of $z$, 
the first $x$ on the right must be $\epsilon$.
Since $x=\epsilon$, we are left with $zy=z$, but the $a$'s within the $y$ cannot
be matched by the $z$ on the \mbox{right as $z$ has only $b$'s.}


The ability to infer this kind of information from the regular constraints systematically is in
the core of our algorithm.
The algorithm gradually refines the regular constraints to fit the equation,
until an infeasible constraint is generated (with an empty language) or until a
solution is detected.
Detecting the existence of a~solution is based on our novel characterisation of
satisfiability of a string constraint: a~constraint $x_1\ldots x_m =
x_{m+1}\ldots x_n \land \bigwedge_{x\in \vars} x\in \lass(x)$,  where $\lass$ assigns regular languages to variables in~$\vars$, has a~solution if the constraint is \emph{stable}, that is, the languages of the two sides are equal,  $\lass(x_1) \concats \lass(x_m) = \lass(x_{m+1})\concats \lass(x_n)$.
A~refinement of the variable languages is derived from a~special product of the
automata for concatenations of the languages on the left-hand and right-hand
sides of the equation.
For the case with $zyx = xxz$ above, the algorithm  
terminates after 2-refinements 
(as discussed above, inferring that (1) $z\in b^+$ and $x = \epsilon$, 
(2) there is no $a$ on the right to match the $a$'s in $y$ on the left).
The wealth of information in the regular
constraints increases with refinements and prunes branches that would be explored otherwise
if the equation was considered alone.
The algorithm is hence effective \mbox{even for pure equations, as we show experimentally.}


Although our algorithm is complete for $\sat$ formulae, 
in $\unsat$ cases the refinement steps may go on forever.
We prove that it is, however, guaranteed to terminate and hence complete for the \emph{chain-free}
fragment~\cite{ChainFree} (and its subset the
\emph{straight-line} fragment~\cite{LB16,AnthonyReplaceAll2018}), the largest known decidable
fragment of string constraints that combines equations, regular and transducer
constraints, and length constraints.
For this fragment, the equality in the definition of
stability may be replaced by a single inclusion and only one refinement
step is sufficient (the case of single equation generalises to multiple
equations where the inclusions must be chosen according to certain criteria).

We have experimentally shown that on established benchmarks featuring hard
combinations of word equations and regular constraints, our prototype implementation is competitive with
a representative selection of string solvers (CVC5, \ziii, \ziiistriv, \ziiistriiire, \ziiitrau,
\ostrich, \sloth, \retro). 
Besides being generally quite fast, 
it seems to be superior especially on difficult instances and has the smallest number of timeouts.

\vspace{-3.0mm}
\section{Overview}\label{sec:overview}
\vspace{-2.0mm}

We will first give an informal overview of our algorithm on the following example
\vspace{-2mm}
\begin{equation}\label{eq:overview}
  xyx = zu \quad \land \quad ww = xa \quad \land \quad u \in (baba)^*a \quad
  \land \quad z \in a(ba)^* \\[-2mm]
\end{equation}
with variables~$u, w,x,y,z$ over the alphabet~$\Sigma = \{a,b\}$. 

Our algorithm works by iteratively refining/pruning the languages in the regular membership constraints from words that cannot be present in any solution.
We denote the regular constraint for a variable $x$ by $\lass(x)$. 
In the example, we have 
$\lass(u) = (baba)^*a$, $\lass(z) = a(ba)^*$ and, implicitly, $\lass(x)=\lass(y)=\lass(w)=\Sigma^*$.

The equation $xyx=zu$ enforces that any solution, an assignment $\ass$ of
strings to variables, satisfies that the string $s =
\ass(x)\concat\ass(y)\concat\ass(x) = \ass(z)\concat\ass(u)$ belongs to the
intersection of the concatenations of languages on the left and the right-hand
side of the equation, $\lass(x)\concat \lass(y) \concat \lass(x) \cap \lass(z)\concat \lass(u)$, as in \cref{eq:xyxzu} below:

\newcommand{\untervar}[2]{\overset{#1}{\raisebox{0pt}[14pt][0pt]{$\overbrace{\raisebox{0pt}[5pt][0pt]{$#2$}}$}}}
\newcommand{\unterthing}[2]{\overset{#1}{\raisebox{0pt}[14pt][0pt]{$#2$}}}

\begin{wrapfigure}[3]{r}{7cm}
\vspace{-10mm}
\hspace{-3mm}
\begin{minipage}{7.3cm}
\begin{equation}
  s \in \untervar{x}{\Sigma^*} ~ \untervar{y}{\Sigma^*} ~
  \untervar{x}{\Sigma^*}
  ~\unterthing{=}{\cap}~~
  \untervar {z}{a(ba)^*} ~
  \untervar{u}{(baba)^*a}.
\label{eq:xyxzu}
\end{equation}
\end{minipage}
\end{wrapfigure}


We may thus refine the languages of $x$ and $y$ by removing those words that cannot be a part of any string $s$ in the intersection.  
The refinement is implemented over finite automata representation of languages,  
assuming that every $\lass(x_i)$ is represented by the
automaton~$\autass(x_i)$. 
The main steps of the refinement are shown in \cref{fig:overview}.
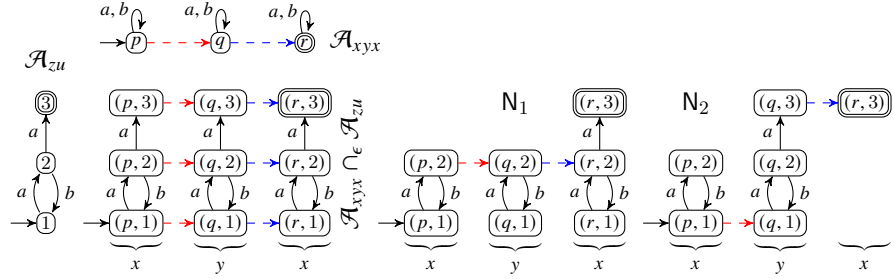
\begin{figure}[t]
\vspace{-2mm}
   \begin{tikzpicture}[->,>=stealth',shorten >=0pt,auto,node distance=10mm,transform shape,scale=0.8]
    \tikzset{prod/.style={inner sep=1pt}}
    \tikzset{prod/.style={draw,rectangle,rounded corners=1mm,inner sep=2pt}}
    \tikzset{hidden/.style={draw=none,fill=white,opacity=0}}

    \node[prod, initial, initial text={}] (p1) {$(p,1)$};
    \node[prod,above of=p1] (p2) {$(p,2)$};
    \node[prod,above of=p2] (p3) {$(p,3)$};

    \node[prod,right of=p1, node distance=14mm] (q1) {$(q,1)$};
    \node[prod,above of=q1] (q2) {$(q,2)$};
    \node[prod,above of=q2] (q3) {$(q,3)$};

    \node[prod,right of=q1, node distance=14mm] (r1) {$(r,1)$};
    \node[prod,above of=r1] (r2) {$(r,2)$};
    \node[prod,above of=r2,accepting] (r3) {$(r,3)$};

    \node[below of=p1,yshift=5mm] (p_below) {$\underset{\displaystyle x}{\underbrace{~~~~~~~~~~~}}$};
    \node[below of=q1,yshift=4.5mm] (q_below) {$\underset{\displaystyle y}{\underbrace{~~~~~~~~~~~}}$};
    \node[below of=r1,yshift=5mm] (r_below) {$\underset{\displaystyle x}{\underbrace{~~~~~~~~~~~}}$};

    \draw (p1) edge[bend left] node{$a$} (p2)
          (p2) edge[bend left] node{$b$} (p1)
          (p2) edge node{$a$} (p3)
          (q1) edge[bend left] node{$a$} (q2)
          (q2) edge[bend left] node{$b$} (q1)
          (q2) edge node{$a$} (q3)
          (r1) edge[bend left] node{$a$} (r2)
          (r2) edge[bend left] node{$b$} (r1)
          (r2) edge node{$a$} (r3);
    \draw[dashed,->,color=red] (p2) -- (q2);
    \draw[dashed,->,color=red]      (p1) -- (q1);
    \draw[dashed,->,color=red]      (p3) -- (q3);
    \draw[dashed,->,color=blue]      (q2) -- (r2);
    \draw[dashed,->,color=blue]      (q3) -- (r3);
    \draw[dashed,->,color=blue]      (q1) -- (r1);

    \node[right of=r2,yshift=-0mm,xshift=-2mm,rotate=90] (prod) {\large{$\prodex$}};


      \node[prod, initial, initial text={}, above of=p3] (p) {$p$};
      \node[prod,above of=q3] (q) {$q$};
      \node[prod,accepting,above of=r3] (r) {$r$};

      \draw (p) edge[loop above] node[left]{$a,b$} (p)
            (q) edge[loop above] node[left]{$a,b$} (q)
            (r) edge[loop above] node[left]{$a,b$} (r);
      \draw[dashed,->,color=red]      (p) -- (q);
      \draw[dashed,->,color=blue]      (q) -- (r);    

      \node[right of=r,yshift=0mm,xshift=-1.5mm] (Axyx) {\large{$\Axyx$}};

  
      \node[prod, initial, initial text={},left of=p1, node distance=15mm] (1) {$1$};
      \node[prod,above of=1] (2) {$2$};
      \node[prod,accepting,above of=2] (3) {$3$};
  
      \draw (1) edge[bend left] node{$a$} (2)
            (2) edge[bend left] node{$b$} (1)
            (2) edge node{$a$} (3); 

      \node[above of=3,yshift=-3mm,xshift=-0mm] (Azu) {\large{$\Azu$}};


    \node[prod, initial, initial text={},right of=r1, node distance=21mm] (n3p1) {$(p,1)$};
    \node[prod,above of=n3p1] (n3p2) {$(p,2)$};
    \node[hidden,above of=n3p2] (n3p3) {$(p,3)$};

    \node[prod,right of=n3p1, node distance=14mm] (n3q1) {$(q,1)$};
    \node[prod,above of=n3q1] (n3q2) {$(q,2)$};
    \node[hidden,above of=n3q2] (n3q3) {$(q,3)$};

    \node[prod,right of=n3q1, node distance=14mm] (n3r1) {$(r,1)$};
    \node[prod,above of=n3r1] (n3r2) {$(r,2)$};
    \node[prod,above of=n3r2,accepting] (n3r3) {$(r,3)$};

    \node[below of=n3p1,yshift=5mm] (n3p_below) {$\underset{\displaystyle x}{\underbrace{~~~~~~~~~~~}}$};
    \node[below of=n3q1,yshift=4.5mm] (n3q_below) {$\underset{\displaystyle y}{\underbrace{~~~~~~~~~~~}}$};
    \node[below of=n3r1,yshift=5mm] (n3r_below) {$\underset{\displaystyle x}{\underbrace{~~~~~~~~~~~}}$};

    \draw (n3p1) edge[bend left] node{$a$} (n3p2)
          (n3p2) edge[bend left] node{$b$} (n3p1)
          (n3q1) edge[bend left] node{$a$} (n3q2)
          (n3q2) edge[bend left] node{$b$} (n3q1)
          (n3r1) edge[bend left] node{$a$} (n3r2)
          (n3r2) edge[bend left] node{$b$} (n3r1)
          (n3r2) edge node{$a$} (n3r3);
    \draw[dashed,->,color=red] (n3p2) -- (n3q2);
    \draw[dashed,->,color=blue]      (n3q2) -- (n3r2);

    \node[above of=n3q2,yshift=0mm,xshift=0mm] (N1) {\large{$\noodle_1$}};


    \node[prod, initial, initial text={},right of=n3r1, node distance=16mm] (n4p1) {$(p,1)$};
    \node[prod,above of=n4p1] (n4p2) {$(p,2)$};
    \node[hidden,above of=n4p2] (n4p3) {$(p,3)$};

    \node[prod,right of=n4p1, node distance=14mm] (n4q1) {$(q,1)$};
    \node[prod,above of=n4q1] (n4q2) {$(q,2)$};
    \node[prod,above of=n4q2] (n4q3) {$(q,3)$};

    \node[hidden,right of=n4q1, node distance=14mm] (n4r1) {$(r,1)$};
    \node[hidden,above of=n4r1] (n4r2) {$(r,2)$};
    \node[prod,above of=n4r2,accepting] (n4r3) {$(r,3)$};

    \node[below of=n4p1,yshift=5mm] (n4p_below) {$\underset{\displaystyle x}{\underbrace{~~~~~~~~~~~}}$};
    \node[below of=n4q1,yshift=4.5mm] (n4q_below) {$\underset{\displaystyle y}{\underbrace{~~~~~~~~~~~}}$};
    \node[below of=n4r1,yshift=5mm] (n4r_below) {$\underset{\displaystyle x}{\underbrace{~~~~~~~~~~~}}$};

    \draw (n4p1) edge[bend left] node{$a$} (n4p2)
          (n4p2) edge[bend left] node{$b$} (n4p1)
          (n4q1) edge[bend left] node{$a$} (n4q2)
          (n4q2) edge[bend left] node{$b$} (n4q1)
          (n4q2) edge node{$a$} (n4q3);
    \draw[dashed,->,color=red]      (n4p1) -- (n4q1);
    \draw[dashed,->,color=blue]      (n4q3) -- (n4r3);

    \node[above of=n4p2,yshift=0mm,xshift=0mm] (N1) {\large{$\noodle_2$}};

  \end{tikzpicture}
\vspace{-2mm}
 \caption{Automata constructions within the refinement.
    Dashed lines represent $\epsilon$.}
\vspace{-5mm}
 \label{fig:overview}
\end{figure}
First, we construct automata for the two sides of the equation:
\begin{itemize}
\vspace{-2mm}
  \item
    $\Axyx$ is obtained by concatenating $\autass(x)$, $\autass(y)$, and $\autass(x)$ again.
    It has $\epsilon$-transitions that delimit the borders of occurrences of~$x$
        and~$y$.\

  \item
    $\Azu$ is obtained by concatenating $\autass(z)$ and $\autass(u)$.
\end{itemize}
\vspace{-2mm}
We then combine $\Axyx$ with $\Azu$ through a synchronous product construction that preserves $\epsilon$-transitions into an automaton $\prodex$. 
Seeing $\epsilon$ as a letter that delimits variable occurrences, $\prodex$
accepts strings $w_x^1\epsilon w_y \epsilon w_x^2$ such that $w_x^1 w_y w_x^2\in
\lass(z)\concat\lass(u)$, $w_x^1\in \lass(x)$, $w_y\in \lass(y)$, and $w_x^2\in \lass(x)$. 

Note that for refining the languages $x,y$ on the left, 
we do not need to see the borders between $z$ and $u$ on the right.
The $\epsilon$-transitions can hence be eliminated from $\Azu$ and it can be minimised. 
In our particular case, this gives much smaller automaton than the one
obtained by connecting~$\autass(z)$ and~$\autass(u)$ (representing $a(ba)^*$ and $(baba)^*a$, respectively).
This is a~significant advantage against algorithms that enumerate alignments of borders of the left and the right-hand side variables/solved forms \cite{solvedform}.

To extract from $\prodex$ the new languages for $x$ and $y$,
we decompose the automata to a disjunction of several automata, which we call \emph{noodles}. 
Each noodle represents a concatenation of languages $L_x^1 \epsilon L_y \epsilon L_x^2$,
and is obtained by choosing one $\epsilon$-transition separating the first
occurrence of $x$ from $y$ (the left column of red $\epsilon$-transitions in
\cref{fig:overview}), one $\epsilon$-transition separating~$y$ from the second
occurrence of $x$ (the right column of blue $\epsilon$-transitions), removing
the other $\epsilon$-transitions, and trimming the automaton.
We have to split the product into noodles because  
some values of $x$ can appear together only with some values of $y$,
and this relation must be preserved after extracting their languages from the product 
(for instance, in $\prodex$ in \cref{fig:overview}, both first occurences of~$x$ and~$y$ can have, among others, values $aa$ and $\epsilon$, but if $x=aa$ then $y$ must be $\epsilon$). 

\cref{fig:overview} shows two noodles, 
$\noodle_1$~and~$\noodle_2$, out of 9~possible noodles from $\prodex$.
We extract the automata for languages $L_x^1$, $L_y$, and $L_x^2$ (their initial
and final states are the states with incoming and outgoing $\epsilon$-transitions in the noodle). 
The refined language for $y$ is then $\lass(y) = L_y$. 
The refined language for $x$ is obtained by unifying  the languages of the first
and the second occurrence of $x$,  $\lass(x) = L_x^1 \cap L_x^2$ (by constructing
a~standard product of the two automata):

\begin{itemize}
    \vspace{-2mm}
  \item  For $\noodle_1$, the refinement is $y \in (ba)^*$ and
    $x \in a$ (computed as $a(ba)^* \cap (ba)^* a$).

  \item  For $\noodle_2$,
    the refinement is $y \in a(ba)^* a$ and
    $x \in \epsilon$ (computed as $(ab)^* \cap \epsilon$).

\end{itemize}
    \vspace{-2mm}
The 7~remaining noodles generated from  $\prodex$ yield $x \in \emptyset$ and are discarded.
Noodles $\noodle_1$ and $\noodle_2$ spawn two disjunctive branches of the computation.

\begin{wrapfigure}[3]{r}{4.7cm}
\vspace{-10mm}
\hspace{-3mm}
\begin{minipage}{5cm}
\begin{equation}
  s \in \untervar{w}{\Sigma^*} ~ \untervar{w}{\Sigma^*}
  ~\unterthing {=}{\cap}~~
  \untervar{x}{a} ~
  \unterthing{a}{a}.
\label{eq:wwxa}
\end{equation}
\end{minipage}
\end{wrapfigure}
For the branch of $\noodle_1$, 
we use the equation $ww = xa$ for the next refinement.
Using the newly derived constraint $x\in a$, we obtain \cref{eq:wwxa} on the right:
%
%
Similarly as in the previous step, the refinement deduces
that $w \in a$.
At this point, the languages on both
sides of all equations match, and so no more refinement is possible:
\vspace{-2mm}
\begin{equation}
  \untervar{x}{a} ~ \untervar{y}{(ba)^*} ~
  \untervar{x}{a}
  ~\unterthing{=}{=}~~
  \untervar{z}{a(ba)^*} ~
  \untervar{u}{(baba)^* a}
  \qquad\text{and}\qquad
  \untervar{w}{a} ~
  \untervar{w}{a}
  ~\unterthing{=}{=}~~
  \untervar{x}{a} ~
  \unterthing{a}{a}.
\vspace{-2mm}
\end{equation}
%

One of the main contributions of this paper, and a cornerstone of our algorithm, is a theorem stating that in this state, when language equality holds for all equations, 
a~solution is guaranteed to exist (see \cref{thm:main}).
We can thus conclude with~$\sat$.

\vspace{-3.0mm}
\section{Preliminaries}
\vspace{-2.0mm}

\paragraph{Sets and strings.}
We use $\nat$ to denote the set of natural numbers (including~0).
We fix a finite \emph{alphabet} $\Sigma$ of \emph{symbols/letters} (usually denoted $a,b,c,\ldots$) for the rest
of the paper. 
A~sequence of symbols $w = a_1 \cdots a_n$ from $\Sigma$ is a \emph{word} or a \emph{string} over $\Sigma$, 
with its \emph{length} $n$ denoted by $|w|$. The set of all words over $\Sigma$
is denoted as~$\Sigma^*$.
$\epsilon \notin \Sigma$ is the \emph{empty word} with $|\epsilon| = 0$.
The \emph{concatenation} of words $u$ and $v$ is denoted $u\concat v$,  $uv$ for
short ($\epsilon$ is a~neutral element of concatenation). A set of words over $\Sigma$ is a \emph{language}, the concatenation of
languages is $L_1\concat L_2 = \{u\concat v\mid u\in L_1 \land v \in L_2\}$,
$L_1 L_2$ for short.
\emph{Bounded iteration} $x^i$, $i\in \nat$, of a word or a language $x$ is defined by $x^0 = \epsilon$ for a~word, $x^0 = \{\epsilon\}$ for a~language, and $x^{i+1} = x^i\concat x$. 
Then $x^* = \bigcup_{i\in\nat}x^i$. 
We often denote regular languages using regular expressions with the standard
notation.

%

\vspace{-2mm}
\paragraph{Automata.}
\newcommand{\aore}{a}
A~\emph{(nondeterministic) finite automaton (NFA)} over $\Sigma$ is a tuple
$\aut= (Q,\Delta,I,F)$
 where $Q$ is a finite set of \emph{states},
$\Delta$ is a set of \emph{transitions} of the form $\move q \aore r$ with $q,r\in
Q$ and $\aore\in\Sigma\cup\{\epsilon\}$, $I\subseteq Q$ is the set of \emph{initial states}, and $F\subseteq Q$
is the set of \emph{final states}.
A~run of~$\aut$ over a~word~$w \in \Sigma^*$ is
a~sequence
 $\move{p_0}{\aore_1}{p_1}
  \move{}{\aore_2}{}
  \ldots
 \move{}{\aore_n}{p_n}$  
where for all $1\leq i \leq n$ it holds that $\aore_i \in \Sigma \cup
\{\epsilon\}$, $\move{p_{i-1}} {\aore_i} {p_i}\in\Delta$, and $w = \aore_1
\concat \aore_2 \concats \aore_n$.
The run is \emph{accepting} if $p_0 \in I$ and $p_n\in F$, and the language
$\langof{\aut}$ of $\aut$ is the set of all words for which $\aut$ has an accepting run.
A~language $L$ is called \emph{regular} if it is accepted by some NFA. 
Two NFAs with the same language are called \emph{equivalent}. 
An automaton without $\epsilon$-transitions is called \emph{$\epsilon$-free}. An automaton with each state belonging to some accepting run is \emph{trimmed}.
%
To concatenate languages of two 
NFAs $\aut= (Q,\Delta,I,F)$ and $\aut'= (Q',\Delta',I',F')$, we construct their \emph{$\epsilon$-concatenation} 
$\aut \concatxof{\epsilon} \aut' = (Q \uplus Q', \Delta \uplus \Delta' \uplus \{ \move{p}{\epsilon}{q} \mid p \in F, q \in I' \}, I, F')$.
To intersect their languages, we construct their \emph{$\epsilon$-preserving product} $\aut\product\aut' = (Q\times Q',\Delta^\times,I\times I',F\times F')$ where 
$\move {(q,q')}{\aore}{(r,r')}\in\Delta^{\times}$ iff either (1) $\aore\in\Sigma$ and 
$\move{q}{\aore}{r}\in\Delta, 
\move{q'}{\aore}{r'}\in\Delta'$, 
or 
(2)~$\aore=\epsilon$ and either 
$q'=r'$, $\move{q}{\epsilon}{r}\in\Delta$ 
or
$q=r$, $\move{q'}{\epsilon}{r'}\in\Delta'$. 

\vspace{-2mm}
\paragraph{String constraints.}
We focus on the most essential string constraints, Boolean combinations of atomic string constraints of two types: word equations and regular constraints.
Let $\vars$ be a set of \emph{string variables} (denoted $u,v,\ldots,z)$, fixed for the rest of the paper.
%
A~\emph{word equation} is an equation of the form $\sterm = \tterm$ where
$\sterm$ and $\tterm$ are (different) \emph{string terms}, i.e., words
from~$\vars^*$.%
\footnote{
Note that terms with letters from~$\Sigma$, sometimes used
in our examples,
can be encoded by replacing each occurrence $o$ of a letter $a$ by a \emph{fresh} variable $x_o$ and a \mbox{regular constraint $x_o\in \{ a \}$.}
}
We do not distinguish between $\sterm = \tterm$ and $\tterm = \sterm$.
A~\emph{regular constraint} is of the form $x\in L$, where $x\in \vars$ and $L$
is a~regular language.
A~\emph{string assignment} is a map $\ass\colon \vars\rightarrow \Sigma^*$.
The assignment is a~solution for a word equation $\sterm=\tterm$ if
$\ass(\sterm) = \ass(\tterm)$ where $\ass(\tterm')$ for a term $\tterm' =
x_1\ldots x_n$ is defined as $\ass(x_1)\concats\ass(x_n)$, and it is a~solution for a regular constraint $x\in L$ if $\ass(x)\in L$. A~solution for a~\emph{Boolean combination} of atomic constraint is then defined as usual.%
%



\vspace{-3.0mm}
\section{Stability of String Constraints}
\vspace{-2.0mm}
The core ingredient of our algorithm, which allows to tightly integrate equations
with regular constraints, 
is the notion of \emph{stability} of a~string constraint, that indicatees. 
%

\vspace{-3.0mm}
\subsection{Stability of Single-Equation Systems}
\vspace{-1.0mm}

We will first discuss stability of a \emph{single-equation system}
$\constr\colon \sterm=\tterm\land \bigwedge_{x\in \vars}x\in \constrlass(x)$ 
where $\constrlass\colon\vars\rightarrow\pow{\Sigma^*}$ is a~\emph{language
assignment}, an assignment of regular languages to variables.
%
We say that a~language assignment $\lass$ \emph{refines} $\constrlass$ if $\lass(x)\subseteq
\constrlass(x)$ for all $x\in\vars$.
If $\lass(x) = \emptyset$ for some~$x \in \vars$, it is \emph{infeasible},
otherwise it is \emph{feasible}.
%
%
For a~term $u = x_1\ldots x_n$,
we define $\lass(u) = \lass(x_1)\concats\lass(x_n)$.
We say that~$\lass$ is \emph{strongly stable for~$\constr$} if $\lass(\sterm) = \lass(\tterm)$. 
The core result of this work is that the existence of a~stable language
assignment for~$\constr$
implies the existence of a~solution, which is formalised below.
%

\begin{restatable}{theorem}{thmMain}\label{thm:main}
$\constr$ has a feasible strongly stable language assignment that refines $\constrlass$, 
iff it has a~solution.
\end{restatable}

\begin{proof}[Sketch of $\Rightarrow$, the other direction is trivial]
Let 
$\sterm= y_1 \ldots y_m$ and $\tterm =  y_{m+1} \ldots y_n$.
Note that a solution cannot be found easily by just taking any words
  $w_i\in\lass(x_i)$, for $1\leq i \leq n$, such that $w_1\concats w_m = w_{m+1} \concats w_n$. The reason is that multiple occurrences of the same variable must have the same value.
To construct a solution, 
we first notice that it is enough to use the shortest words in the languages of the variables. 
We can then assume the lengths of the strings valuating each variable fixed (the
  smallest lengths in the languages), which in turn fixes the positions of
  variables' occurrences within the sides of the equation.
We then construct the strings in the solution by selecting letters and propagating them through equalities of opposite positions of the equation sides and also between different occurrences of the same position in the same variable. 
Showing that the process of selecting and propagating letters terminates requires to show that the sequence of constructed partial solutions is decreasing w.r.t. a complex well-founded ordering of partial solutions.
The full proof may be found in \cref{section:proof}.
\qed
\end{proof}

\vspace{-1mm}

Additionally, 
in the special case of \emph{weak} equations---i.e., equations $\sterm = \tterm$
where one of the sides, say $\tterm$, satisfies the condition that all variables
occurring in~$\tterm$ occur in $\sterm = \tterm$ exactly
once---the stability condition in \cref{thm:main} can be weakened to one-sided
language inclusion only:
In case $\tterm$ is the term satisfying the condition,
we say that~$\lass$ is \emph{weakly stable for~$\constr$} if
$\lass(\sterm)\subseteq\lass(\tterm)$.  



\vspace{-1mm}
\begin{theorem}\label{thm:weaklystable}
$\constr$ with a weak equation has a~feasible weakly stable language assignment that refines $\constrlass$
iff~$\constr$ has a~solution.
\end{theorem}
\vspace{-1mm}
Note that weak stability allows multiple occurrences of a variable on the
left-hand side of $\sterm = \tterm$.
Intuitively, the multiple occurrences must have the same value, and having them
on the left-hand side of the inclusion forces their synchronisation.
For instance, for $\constr\colon xx=y \land x\in\{a,b\} \land y\in\{ab\}$, the
inclusion $\lass(xx)\subseteq \lass(y)$ is satisfied by no feasible refinement
$\lass$ of $\lass_\constr$, revealing that $\constr$ has no solution, while
$\lass(xx)\supseteq \lass(y)$ is satisfied already by $\lass_\constr$ itself.
%
%

\vspace{-3.0mm}
\subsection{Stability of Multi-Equation Systems}
\vspace{-1.0mm}
Next, we extend the definition of stability to \emph{multi-equation systems}, conjunctions of the form 
$\constr\colon \constreq \land \bigwedge_{x\in \vars}x\in \constrlass(x)$ where
$\constreq\colon \bigwedge_{i=1}^m \sterm_i=\tterm_i$ for $m\in\nat$.
We assume that every two equations are different, i.e., $\{\sterm_i,\tterm_i\}\neq\{\sterm_j,\tterm_j\}$ if $i\neq j$. 



We generalise stability in a way that utilises both strong and weak stability of
single equation systems and both \cref{thm:main} and \cref{thm:weaklystable}.
Again, we interpret every equation $\sterm_i = \tterm_i$ as a pair of
inclusions and show that it suffices to satisfy a~certain subset of these
inclusions in order to obtain a~solution.
A~sufficient subset of inclusions is defined through the notion of an
\emph{inclusion graph} of~$\constreq$.
It is a~directed graph $\graph =(V,E)$ where vertices~$V$ are inclusion
constraints of the form $\graphnode{\sterm_i}{\tterm_i}$ or
$\graphnode{\tterm_i}{\sterm_i}$, for $1\leq i \leq m$, and $E\subseteq V\times
V$. 
An~inclusion graph must satisfy the following conditions:
%
\enlargethispage{-2mm}
\begin{enumerate}[({IG}1)]
\vspace{-1.5mm}
\item
For each $\sterm_i = \tterm_i$ in~$\constreq$,
at least one of the nodes $\graphnode{\sterm_i}{\tterm_i}$,
$\graphnode{\tterm_i}{\sterm_i}$ is in $\vertices.$
\label{cond:atleastone}
\item 
If $\graphnode{\sterm_i}{\tterm_i}\in \vertices$ and $\tterm_i$ has a variable with multiple occurrences in right-hand sides of vertices of $V$, 
then also $\graphnode{\tterm_i}{\sterm_i}\in \vertices$.
\label{cond:multioccur}
\item 
$(\graphnode{\sterm_i}{\tterm_i}, \graphnode{\sterm_j}{\tterm_j})
\in\edges$ iff $\graphnode{\sterm_i}{\tterm_i}, \graphnode{\sterm_j}{\tterm_j}\in\vertices $ and $\sterm_i$ and $\tterm_j$ share a~variable.  
\label{cond:edges}
\item 
If $\graphnode{\sterm_i}{\tterm_i} \in \vertices$ lies on a cycle, then also $\graphnode{\tterm_i}{\sterm_i} \in \vertices.$
\label{cond:cycle}
\vspace{-1.5mm}

\end{enumerate}
Note that by \condref{cond:edges}, $\edges$ is uniquely determined by~$\vertices$.
We define that a~language assignment $\lass$ is \emph{stable for an inclusion graph} $\G=(V,E)$ of
$\constreq$ if it \mbox{satisfies every inclusion in~$\vertices$.}

\vspace{-2mm}
\begin{restatable}{theorem}{thmGraphWStable}\label{thm:graphwstable}
Let $\graph$ be an inclusion graph of~$\constreq$.
Then there is a~feasible language assignment that refines $\constrlass$ and is stable
for~$\graph$ iff~$\constr$ has a~solution.
\end{restatable}
\vspace{-1mm}

The proof of \cref{thm:graphwstable} is in \cref{sec:proofThmGraphWStable}.
Intuitively, the set of inclusions needed to guarantee a solution is specified by the vertices of an inclusion graph. 
All equations must contribute with at least one inclusion, by \igcondref{cond:atleastone}.
Including only one inclusion corresponds to using weak stability.
Including both inclusions corresponds to using strong stability. 
We will use inclusion graphs in our algorithm to direct propagation of
refinements of language assignments. We will wish to avoid using strong
stability when possible since cycles that it creates in the graph may cause the algorithm to diverge.

Conditions \condref{cond:multioccur}--\condref{cond:cycle} specify where weak
stability is not enough.
Namely, \igcondref{cond:multioccur} enforces that 
to use weak stability, multiple occurrences of a~variable can only occur on the
left-hand side of an inclusion (as in the definition of weak stability),
otherwise strong stability must be used. 
%
The edges defined by \igcondref{cond:edges} are used in \igcondref{cond:cycle}.
An edge means that a refinement of the language assignment made to satisfy
the inclusion in the source node may invalidate the inclusion in the target node. 
\igcondref{cond:cycle} covers the case of a cyclic dependency of a~variable on
itself. A~self-loop indicates that a~variable occurs on both sides of an equation
(breaking the definition of weak stability). A~longer cycle indicates such
a~cyclic dependency caused by transitively propagating the inclusion relation.

\newcommand{\figInclGraph}[0]{
\begin{wrapfigure}[7]{r}{2.8cm}
\hspace*{-4mm}
\begin{minipage}{2.8cm}
\scalebox{0.80}{
\begin{tikzpicture}[>=stealth,shorten >=0pt,auto,node distance=16mm,transform shape]

  \node[inner sep=0]                (uvx<x) {$\graphnode{uvx}{x}$};
  \node[left  of=uvx<x,inner sep=0] (x<uvx) {$\graphnode{x}{uvx}$};
  \node[above of=uvx<x,inner sep=0,yshift=-4mm] (u<v)   {$\graphnode{u}{v}$};
  \node[above of=x<uvx,inner sep=0,yshift=-4mm] (v<u)   {$\graphnode{v}{u}$};
  \node[above of=u<v,inner sep=0,yshift=-4mm]   (z<u)   {$\graphnode{z}{u}$};
  \node[above of=v<u,inner sep=0,yshift=-4mm]   (u<z)   {$\graphnode{u}{z}$};

  \draw 
        (uvx<x) edge[->,loop below] coordinate (l1) (uvx<x)
        (uvx<x) edge[->,bend right] (z<u)
        (uvx<x) edge [->](v<u)
        (uvx<x) edge [->](u<v)
        (x<uvx) edge[->,loop below] coordinate(l2) (x<uvx)
        (u<z)   edge[->] (v<u)
        (u<z)   edge[->,bend right] (x<uvx)
        (u<v)   edge[->] (z<u)
        (u<v)   edge[->] (x<uvx)
        (v<u)   edge[->] (x<uvx)

        (uvx<x) edge[<->,dashed,bend left] (x<uvx)
        (u<v) edge[<->,dashed,bend right] coordinate(l3) (v<u)
  ;

  \begin{pgfonlayer}{background}
  \node[-,dotted,rectangle,fill=Red!20,draw=black!70,rounded corners=5pt,inner sep=4pt, fit=(uvx<x) (l1) ] (scc1) {};
  \node[-,dotted,rectangle,fill=Red!20,draw=black!70,rounded corners=5pt,inner sep=4pt, fit=(x<uvx) (l2) ] (scc2) {};
  \node[-,dotted,rectangle,fill=Blue!40,fill opacity=0.2,draw=black!70,rounded corners=5pt,inner sep=4pt, fit=(u<v) (v<u) (scc1) (scc2) (l2) ] (strong) {};
  \path[-,dotted,fill=Green!60,fill opacity=0.1,draw=black!70,rounded
    corners=5pt,inner sep=4pt]
    ([xshift=-6mm,yshift=3mm]u<z.north west) |-
    ([xshift=4.5mm,yshift=-8mm]uvx<x.south east) |-
    ([xshift=6mm,yshift=3mm]v<u.north east) --
    ([xshift=6mm,yshift=3mm]u<z.north east) -- cycle;
  \end{pgfonlayer}

\end{tikzpicture}
}
\end{minipage}
\end{wrapfigure}
}

\vspace{-3mm}
\subsection{Constructing Inclusion Graphs and Chain-Freeness}
\enlargethispage{2mm}
\vspace{-1mm}

We now discuss a construction of a~suitable inclusion graph. 
%
Our algorithm for solving string constraints will use the graph nodes to gradually refine the language assignment, propagating information along the graph edges. It is guaranteed to terminate when the graph is acyclic. 
Below, we give an algorithm that generates an inclusion graph that contains as
few inclusions as possible and is acyclic whenever possible.

The graph is obtained from a~simplified version $\splitgraph$
of the splitting graph of \cite{ChainFree}, 
which is the basis of the definition of the chain-free fragment, for which our
algorithm is complete.
The nodes of~$\splitgraph$
are all inclusions $\graphnode{\sterm_i}{\tterm_i},\graphnode{\tterm_i}{\sterm_i}$, for $1\leq i \leq m$, and it has an edge from $\graphnode\sterm\tterm$ to $\graphnode{\sterm'}{\tterm'}$ 
if $\sterm$ and $\tterm'$ each have a~different occurrence of the same variable
(the ``\emph{different}'' here meaning not the same position in the same term in
the same equation, e.g., for inclusions induced by the equation $u =
v$, for $u,v \in \vars$, there will be no edge between $\graphnode u v$ and $\graphnode
v u$).

\begin{wrapfigure}[9]{r}{6.3cm}
\vspace{-12mm}
\clipbox{{1mm} 0 {7mm} 0}{%
\begin{minipage}{7cm}
\small
\begin{algorithm}[H]
  \caption{$\incl(\constreq)$}
  \label{alg:inclgraph}
  \KwIn{Conjunction of string equations $\constreq$.}
  \KwOut{An inclusion graph of $\constreq$.}

  $\G := \splitgraph$;
  $\St' := \emptyset$\;
  \While{$G$ has a trivial source SCC $(\{v\},\emptyset)$}{
    $\G := \G \setminus \{v, \dualof{v}\}$\;
    $\St' := \St' \cup \{v\}$\;\label{line:addtscc}
  }
  $V:=$ $\St' \cup $ the remaining nodes of $\G$\;\label{line:addrest}
  \Return the inclusion graph with nodes $\St$\;
\end{algorithm}
\end{minipage}
}
\end{wrapfigure}

The algorithm for constructing an inclusion graph from $\splitgraph$ starts by
iteratively removing nodes that are trivial source strongly connected components
(SCCs) from
$\splitgraph$ (trivial means a graph $(\{v\},\emptyset)$ with no edges, source means with no edges coming
from outside into the component). 
With every removed node $v$, 
the algorithm removes from $\splitgraph$ also the dual node  $\dualof v$ (the other inclusion),
and it adds $v$ to the inclusion graph. 
When no trivial source SCCs are left, that is, the remaining nodes are all reachable from non-trivial SCCs, the algorithm adds to the inclusion graph all the remaining 
nodes. 
The pseudocode of the algorithm is shown in \cref{alg:inclgraph}.
It~uses $\scc(\G)$ to denote the set of SCCs of $\G$ and $\G\setminus V$ to denote the graph obtained from $\G$ by removing the vertices in $V$ together with the adjacent edges.


\begin{wrapfigure}[7]{r}{2.8cm}
\hspace*{-4mm}
\begin{minipage}{2.8cm}
\scalebox{0.80}{
\begin{tikzpicture}[>=stealth,shorten >=0pt,auto,node distance=16mm,transform shape]

  \node[inner sep=0]                (uvx<x) {$\graphnode{uvx}{x}$};
  \node[left  of=uvx<x,inner sep=0] (x<uvx) {$\graphnode{x}{uvx}$};
  \node[above of=uvx<x,inner sep=0,yshift=-4mm] (u<v)   {$\graphnode{u}{v}$};
  \node[above of=x<uvx,inner sep=0,yshift=-4mm] (v<u)   {$\graphnode{v}{u}$};
  \node[above of=u<v,inner sep=0,yshift=-4mm]   (z<u)   {$\graphnode{z}{u}$};
  \node[above of=v<u,inner sep=0,yshift=-4mm]   (u<z)   {$\graphnode{u}{z}$};

  \draw 
        (uvx<x) edge[->,loop below] coordinate (l1) (uvx<x)
        (uvx<x) edge[->,bend right] (z<u)
        (uvx<x) edge [->](v<u)
        (uvx<x) edge [->](u<v)
        (x<uvx) edge[->,loop below] coordinate(l2) (x<uvx)
        (u<z)   edge[->] (v<u)
        (u<z)   edge[->,bend right] (x<uvx)
        (u<v)   edge[->] (z<u)
        (u<v)   edge[->] (x<uvx)
        (v<u)   edge[->] (x<uvx)

        (uvx<x) edge[<->,dashed,bend left] (x<uvx)
        (u<v) edge[<->,dashed,bend right] coordinate(l3) (v<u)
  ;

  \begin{pgfonlayer}{background}
  \node[-,dotted,rectangle,fill=Red!20,draw=black!70,rounded corners=5pt,inner sep=4pt, fit=(uvx<x) (l1) ] (scc1) {};
  \node[-,dotted,rectangle,fill=Red!20,draw=black!70,rounded corners=5pt,inner sep=4pt, fit=(x<uvx) (l2) ] (scc2) {};
  \node[-,dotted,rectangle,fill=Blue!40,fill opacity=0.2,draw=black!70,rounded corners=5pt,inner sep=4pt, fit=(u<v) (v<u) (scc1) (scc2) (l2) ] (strong) {};
  \path[-,dotted,fill=Green!60,fill opacity=0.1,draw=black!70,rounded
    corners=5pt,inner sep=4pt]
    ([xshift=-6mm,yshift=3mm]u<z.north west) |-
    ([xshift=4.5mm,yshift=-8mm]uvx<x.south east) |-
    ([xshift=6mm,yshift=3mm]v<u.north east) --
    ([xshift=6mm,yshift=3mm]u<z.north east) -- cycle;
  \end{pgfonlayer}

\end{tikzpicture}
}
\end{minipage}
\end{wrapfigure}
    
\beginexample
In the picture in the right,
we show an example of the construction of the inclusion graph~$\graph$ from $\splitgraph$ for
$\constreq\colon z=u \land u =v \land uvx = x$.
Edges of $\splitgraph$ are solid lines,
the inclusion graph has both solid and dashed
edges. The inner red boxes are 
the non-trivial SCCs of $\splitgraph$. They
are enclosed in the box of nodes that are added on
\lnref{line:addrest} of \cref{alg:inclgraph}. The outer-most box encloses
the inclusion graph, including one node added on \lnref{line:addtscc}.
\qed

%
\begin{restatable}{theorem}{thmInclCorr}
  \label{thm:inclGraphCorr}
  For a conjunction of equations $\constreq$, $\incl(\constreq)$ is an inclusion
  graph for~$\constreq$ with the smallest number of vertices.
  Moreover, if there exists an acyclic inclusion graph for~$\constreq$, then
  $\incl(\constreq)$ is acyclic.
\end{restatable}
In \cref{sec:satcheck}, we will show a satisfiability checking algorithm that guarantees termination when given an acyclic inclusion graph. 
Here we prove that the existence of an acyclic inclusion graph coincides with the
chain-free fragment of string constraints \cite{ChainFree}, which is the largest known decidable fragment of string constraints with equations, regular and transducer constraints, and length constraints (up to the incomparable fragment of quadratic equations). 
\emph{Chain-free} constraints are defined as those where the simplified
splitting graph $\splitgraph$ has no cycle.
The following theorem is proven in \cref{sec:proof_acyclic_splitting}.

\begin{restatable}{theorem}{thmAcyclicIGChainFree}
\label{thm:acyclicIGChainFree}
A~multi-equation system~$\constr$ is chain-free iff there exists an acyclic
inclusion graph for~$\constr$.
\end{restatable}

\newcommand{\algRefinement}{
\begin{algorithm}[H]
  \caption{$\refine(v,\autass)$}
  \label{alg:refinement}
  \KwIn{A vertex $v = \graphnode \sterm \tterm$ with $\sterm = x_1\concats x_n$ and $\tterm=y_1\concats y_m$, an automata assignment $\autass$}
  \KwOut{A tight refinement of $\autass$ w.r.t.\ $v$}
  
  \vspace{0.5mm}
  $\prod := \autass(\sterm) \product \reduce(\autass(\tterm))$\;\label{ln:product}
  $\noodles := \noodlify(\prod)$\;\label{ln:noodlify}
  $\tightset := \emptyset$\;
  \For{$\noodle\in\noodles$}
  {\label{ln:for_noodles}
    $\autass' := \autass$\;
     \For{$1\leq i\leq n$}{\label{ln:for_intersect}
      $\autass'(x_i) := \bigcap\{\noodle(j) \mid 1\leq j\leq n,x_i = x_j\}$\;
    }
    \lIf{$\langof{\autass'(\sterm)} = \emptyset$}{
      \Continue
    } 
    $\tightset := \tightset \cup \{\autass'\}$\;
  }
  \Return{$\tightset$}\;
  \vspace{-0.7mm}
\end{algorithm}
}

\newcommand{\algPropagate}[0]{
\small
\begin{algorithm}[H]
  \caption{$\propagate(\graph_\constreq, \constrautass)$}
  \label{alg:propagate}

  \KwIn{Inclusion graph $\graph_\constreq = (\vertices,\edges)$, initial automata assignment $\constrautass$.}
  \KwOut{$\sat$ if $\constr$ is satisfiable, ~~~~~~~~~~~ $\unsat$ if $\constr$ is unsatisfiable}

  \vspace{1mm}
  $\worklist :=  \lifo{(\constrautass,\vertices)}$\; \label{ln:prop:init}
  \While{$\worklist \neq \emptyset$}{
    $(\autass, \frontier) := \worklist.\pop()$\;

    \lIf{$\frontier = \emptyset$}{\label{ln:prop:sat}
      \Return{$\sat$}
    }

    $v = \graphnode \sterm \tterm := \frontier.\pop()$\;

    \If{$\langof{\autass(\sterm)} \subseteq \langof{\autass(\tterm)}$ \label{line:inclusiontest}}{
      $\worklist.\push((\autass, \frontier))$\;
      \Continue\;
    }

    $\tightset := \refine(v,\autass)$\;

    $\frontier' := \frontier$\;


    \ForEach{$(v,u)\in\edges$ s.t. $u\not\in \frontier$}{
      $\frontier'.\push(u)$\;
    }
    \ForEach{$\autass'\in\tightset$}{
      $\worklist.\push(\autass', \frontier')$\;
    }

  }
  \Return{$\unsat$}\;
\end{algorithm}
}

\vspace{-3mm}
\section{Algorithm for Satisfiability Checking}
\label{sec:satcheck}
\vspace{-1mm}

Our algorithm for testing satisfiability of a multi-equation system $\constr$ is based on \cref{thm:graphwstable}. 
The algorithm first constructs a~suitable inclusion graph of $\constreq$ using
\cref{alg:inclgraph} and then
it gradually refines the original language assignment $\constrlass$ according to
the dependencies in the inclusion graph until it either finds a~stable feasible language assignment or
concludes that no such language assignment exists. 

A language assignment $\lass$ is in the algorithm represented by an
\emph{automata assignment}~$\autass$, which assigns to 
every variable~$x$ an $\epsilon$-free NFA $\autass(x)$ with $\langof{\autass(x)} = \lass(x)$. 
We use $\autass(\tterm)$ for a~term $\tterm = x_1 \dots x_n$ to denote the NFA $\autass(x_1) 
\concatxof{\epsilon} \cdots \concatxof{\epsilon} \autass(x_n)$.
In the following text, we identify a language assignment with the corresponding automata assignment and vice versa. 


\vspace{-2mm}
\subsection{Refining Language Assignments by Noodlification}
\vspace{-0mm}

The task of a~refinement step is to create a new language assignment that
refines the old one, $\lass$, and satisfies one of the inclusions previously not
satisfied, say $\graphnode{\sterm}{\tterm}$. 
In order for the algorithm to be sound when returning $\unsat$, a~refinement
step must preserve all existing solutions.
%
It will therefore return a~set~$\tightset$ of refinements of $\lass$ that is
\emph{tight w.r.t.~$\graphnode\sterm\tterm$}, that is, every solution of $\sterm = \tterm$ under $\lass$ is also a solution of $\sterm=\tterm$ under some of its refinements in~$\tightset$. 

\cref{alg:refinement} computes such a tight set. 
\lnref{ln:product} computes the automaton $\prod$, which accepts $\lass(\sterm)\cap\lass(\tterm)$. 
In order to be able to extract new languages for the variables of $\sterm$ from it, $\prod$ marks borders between the variables of $\sterm$ with $\epsilon$-transitions. That is, when $\epsilon$ is understood as a special letter, $\prod$ accepts the \emph{delimited language} $\lang^\epsilon(\prod)$ of words $w_1\epsilon\concats\epsilon w_n$ with $w_i\in\lass(x_i)$ for $1\leq i \leq n$ and $w_1\concats w_n\in\lass(\tterm)$.
Notice that 
$\autass(\tterm)$ 
is on \lnref{ln:product} minimised. This means removal 
of $\epsilon$-transitions marking the borders  of variables' occurrences,  and then minimisation
\begin{wrapfigure}[14]{r}{7.7cm}
\vspace{-8.5mm}
\small
\algRefinement    
\end{wrapfigure}
 by any automata size reduction method (we use simulation quotient~\cite{Aziz93,HHK95}). Since the product is then representing only the borders of the
 variables on the left (because $\autass(\sterm)$ keeps the
 $\epsilon$-transitions generated from the concatenation with
 $\concatxof{\epsilon}$), but not the borders of variables in $t$, it does not
 actually generate an explicit representation of possible alignments of borders
 of variables' occurrences. 

We then extract from $\prod$ a~language for each \emph{occurrence} of a variable in $\sterm$.
\lnref{ln:noodlify} 
divides $\prod$ into a~set of the so-called \emph{noodles}, which are sequences of automata $\noodle = \noodle(1),\ldots,\noodle(n)$ 
that preserve the delimited language in the sense that $\bigcup_{\noodle\in\noodles} \lang^\epsilon(\noodle(1)\concatxof{\epsilon}\cdots\concatxof{\epsilon}\noodle(n)) = \lang^\epsilon(\prod)$. 

Technically,
assuming w.l.o.g that $\prod$ has a single initial state $r_0$ and a~single final state $q_n$,
$\noodlify(\prod)$ generates one noodle $\noodle$ for each $(n-1)$-tuple $\move{q_1}\epsilon{r_1},\ldots,\move{q_{n-1}}\epsilon{r_{n-1}}$ of transitions that appear, in that order, in an accepting run of $\prod$ (note that every accepting run has $n-1$ $\epsilon$-transitions by construction of $\prod$, since $\autass(s)$ also had $n-1$ $\epsilon$-transitions in each accepting run and $\reduce(\autass(t))$ is $\epsilon$-free):  
for each $1\leq i \leq n$, $\noodle(i)$ arises by trimming $\prod$ after its initial states were replaced by $\{r_{i-1}\}$ and final states by $\{q_i\}$.

The \textbf{for} loop on \lnref{ln:for_noodles} then turns each noodle $\noodle$
into a refined automata assignment $\autass'$ in~$\tightset$ by unifying/intersecting languages of different occurrences of the same variable:
for each $x\in \vars$, $\autass'(x)$ is the automata intersection of all automata $\noodle(i)$ with $x_i = x$.
The fact that $\tightset$ is a~tight set of refinements (i.e., that it preserves
all solutions of~$\autass$) follows from 
that every path of $\prod$ can be found in $\noodles$ and that the use of
$\epsilon$-transitions allows to reconstruct the NFAs
\mbox{corresponding to the variables.}
\vspace{-1mm}

\newcommand{\figRefinementEx}[0]{
\begin{wrapfigure}[4]{r}{1.0cm}
\vspace*{-8mm}
\hspace*{-2mm}
\begin{minipage}{1.0cm}
\begin{tikzpicture}[->,>=stealth',shorten >=0pt,auto,node
  distance=15mm,transform shape,scale=0.8]
  \tikzset{every node/.style={rounded corners=1mm}}

  \node[draw] (xyx<zu) {$\graphnode{xyx}{zu}$};
  \node[draw,below of=xyx<zu,yshift=2mm] (ww<xa) {$\graphnode{ww}{xa}$};

  \draw (xyx<zu) edge (ww<xa);

\end{tikzpicture}
\end{minipage}
\end{wrapfigure}
}

\begin{wrapfigure}[4]{r}{1.0cm}
\vspace*{-8mm}
\hspace*{-2mm}
\begin{minipage}{1.0cm}
\begin{tikzpicture}[->,>=stealth',shorten >=0pt,auto,node
  distance=15mm,transform shape,scale=0.8]
  \tikzset{every node/.style={rounded corners=1mm}}

  \node[draw] (xyx<zu) {$\graphnode{xyx}{zu}$};
  \node[draw,below of=xyx<zu,yshift=2mm] (ww<xa) {$\graphnode{ww}{xa}$};

  \draw (xyx<zu) edge (ww<xa);

\end{tikzpicture}
\end{minipage}
\end{wrapfigure}
\beginexample\label{ex:refinement}
  Consider the multi-equation system $\constr$ from \cref{sec:overview} and the vertex 
  $\graphnode{xyx}{zu}$ of its inclusion graph given in
  the right.
  The construction of the 
  product automaton $\prod$ from \cref{alg:refinement} is shown in \cref{fig:overview}.
  The set of noodles $\noodlify(\prod) = \{ \noodle_1, \dots, \noodle_7 \}$ is
  given in \cref{sec:noodles_example} ($\noodle_1$~and~$\noodle_2$ are in
  \cref{fig:overview}).
  On \lnref{ln:for_intersect}, we need to compute intersections of $\noodle_i(1)
  \cap \noodle_i(3)$ for each noodle~$\noodle_i$. These 
  parts of the noodle correspond to the two occurrences of the same variable~$x$.
  The only noodles yielding nonempty languages for~$x$
  are~$\noodle_1$ and~$\noodle_2$.
  The noodle~$\noodle_1$ leads to a~refinement~$\autass_1$ of~$\autass$ where 
  $\langof{\autass_1(x)} = a$ (computed as the intersection of languages
  $\noodle_1(1) = a(ba)^*$ and $\noodle_1(3) = (ba)^*a$) and
  $\langof{\autass_1(y)} = (ba)^*$.
  The~noodle~$\noodle_2$ leads to a refinement~$\autass_2$ of~$\autass$ where
  $\langof{\autass_2(x)} = \epsilon$ (computed as the intersection of
  languages $\noodle_2(1) = (ab)^*$ and $\noodle_2(3) = \epsilon$) and $\langof{\autass_2(y)} = a(ba)^*a$.
  \qed

\vspace{-1mm}
\begin{example}
An example with a non-terminating sequence of refinement steps is $xa=x \land
  x\in a^+$, explained in detail in \cref{sec:add-examples}. Every $i$-th step refines $\lass(x)$ to $x\in a^{i+1}a^*$. Note that many similar examples could be handled by simple heuristics that take into account lengths of strings, already used in other solvers. 
\qed
\end{example}

\vspace{-5mm}
\subsection{Satisfiability Checking by Refinement Propagation}
\vspace{-1.5mm}

\begin{wrapfigure}[20]{r}{6.3cm}
\vspace{-8.5mm}
\algPropagate     
\end{wrapfigure}
The pseudocode of the satisfiability check of $\constr$ is given in \cref{alg:propagate}.
It starts with the automaton assignment $\constrautass$ corresponding to $\constrlass$, 
and it uses graph nodes $\graphnode \sterm \tterm$ not satisfied in the current $\autass$ to refine it, 
that is, to replace $\autass$ by some automaton assignment returned by $\refine(\graphnode \sterm \tterm,\autass)$.

The algorithm maintains the current value of $\autass$ and a worklist $\frontier$ of nodes for which the weak-stability 
condition might be invalidated, either initially or since they were affected by some previous refinement.
Nodes are picked from the worklist, and if the inclusion at a~node is found not satisfied in the current automata assignment $\autass$, the node is used to refine it.
Stability is detected when $\frontier$ is empty---there in no potentially unsatisfied inclusion. 

Since 
$\refine(\graphnode \sterm \tterm,\autass)$ does not return a single language
assignment but a~set of language assignments that refine $\autass$,
the computation spawns an independent branch for each of them. 
\cref{alg:propagate} adds the branches for processing in the queue
$\worklist$.
The branching is disjunctive, meaning stability is returned when a single branch detects stability.
If all branches terminate with an infeasible assignment, 
then the algorithm concludes that the constraint is unsatisfiable.

The worklist and the queue of branches are first-in first-out (this is important for showing termination in \cref{thm:sattermination}). To minimise the number of refinement steps, the nodes are initially inserted in $\frontier$ in 
an order compatible with a \mbox{topological order of the SCCs.}

\vspace{-1mm}
\begin{example}
  Consider again the multi-equation system $\constr$ from \cref{sec:overview} and the inclusion graph
  in Example~2.
  The initial automata 
  assignment $\constrautass$ is then given as $\langof{\constrautass(a)} = \{ a \}$, 
  $\langof{\constrautass(z)} = a(ba)^*$, $\langof{\constrautass(u)} = (baba)^*a$, and $\langof{\constrautass(x)} = \langof{\constrautass(y)} 
  =\langof{\constrautass(w)} = \Sigma^*$. The queue $\worklist$ on
  \lnref{ln:prop:init} of \cref{alg:propagate} is hence 
  initialised as $\worklist = \lifo{(\constrautass, \lifo{\graphnode{xyx}{zu}, \graphnode{ww}{xa}}) }$.
  The computation of the main loop of \cref{alg:propagate} then proceeds as
  follows.

  \begin{description}
    \vspace{-1mm}
    \item[1st iteration.]
      The dequeued element is $(\constrautass, \lifo{\graphnode{xyx}{zu},
      \graphnode{ww}{xa}})$ and $v$ (dequeued from~$\frontier$) is
      $\graphnode{xyx}{zu}$.
      The condition on \lnref{line:inclusiontest} is not
      satisfied, hence the algorithm calls
      $\refine(\graphnode{xyx}{zu}, \constrautass)$.
      The refinement 
    yields two new automata assignments, $\autass_1, \autass_2$, which are
      defined in Example~2.   
    The queue $\worklist$ is hence extended to $\lifo{(\autass_1,
      \lifo{\graphnode{ww}{xa}}), (\autass_2, \lifo{\graphnode{ww}{xa}})}$.

    \item[2nd iteration.]
      The dequeued element is $(\autass_2, \lifo{\graphnode{ww}{xa}})$.
      The condition on \lnref{line:inclusiontest} is not satisfied since
      $\langof{\autass_2(x)} = \{ \epsilon \}$ 
    and $\langof{\autass_2(w)} = \Sigma^*$. In this case,
      $\refine(\graphnode{ww}{xa}, \autass_2) = \emptyset$, 
    hence nothing is added to $\worklist$, i.e., $\worklist = \lifo{(\autass_1, \lifo{\graphnode{ww}{xa}})}$.

    \item[3rd iteration.]
      The dequeued element is $(\autass_1, \lifo{\graphnode{ww}{xa}})$.
      The condition on \lnref{line:inclusiontest} is not satisfied
      ($\Sigma^*\concat\Sigma^* \not\subseteq a\concat a$) and
      $\refine(\graphnode{ww}{xa}, \autass_1) = \{\autass_3\}$
      where $\autass_3$ is as $\autass_1$ except that $\autass_3(w)$ accepts only $a$. 
    $\worklist$ is then updated to $\lifo{(\autass_3, \emptyset)}$.

    \item[4th iteration.] The condition on \lnref{ln:prop:sat} is satisfied and the algorithm returns $\sat$.
      \qed
  \end{description}
\end{example}
\vspace{-1mm}

As stated by \cref{thm:soundness} below, 
the algorithm is sound in the general case (an answer is always correct).
Moreover, \cref{thm:chain-free-compl,thm:acyclicIGChainFree} imply that when \cref{alg:inclgraph} is used to construct the inclusion graph, we have a complete algorithm for chain-free constraints.  

\vspace{-1mm}
\begin{restatable}[Soundness]{theorem}{thmSoundness}\label{thm:soundness}
  If $\propagate(\graph_\constreq, \constrautass)$ 
  returns $\sat$, then $\constr$ is satisfiable, and
  if $\propagate(\graph_\constreq, \constrautass)$ returns $\unsat$,
  $\constr$ is unsatisfiable.
\end{restatable}

\vspace{-5mm}
\begin{restatable}{theorem}{thmChainFreeComplete}\label{thm:chain-free-compl}
  If $\graph_\constreq$ is acyclic, then $\propagate(\graph_\constreq, \constrautass)$
  terminates.
\end{restatable}

\subsection{Working with Shortest Words}

\cref{alg:propagate} can be improved with a weaker termination condition that takes
into account only shortest words in the languages assigned to variables. Importantly, this
gives us completeness in the $\sat$ case for general constraints, i.e., the algorithm is always guaranteed to return $\sat$ if a~solution exists. 

Let $\Min{\lass}$ be the language assignment obtained from~$\lass$ by assigning to every $x\in \vars$
the set of shortest words from $\lass(x)$
i.e., $\Min{\lass}(x) = \{w \in \lass(x) \mid \forall u \in \lass(x)\colon |w| \leq |u|\}$.
Then, for $\constr\colon \sterm=\tterm\land \bigwedge_{x\in \vars}x\in \constrlass(x)$, we say 
that~$\lass$ is \emph{strongly min-stable for $\constr$} if $\Min \lass$ is stable for~$\constr$.
Similarly, $\lass$ is \emph{weakly min-stable for $\constr$} if $s=t$ is weak and 
$\Min{\lass}(\sterm) \subseteq \lass(\tterm)$. 
Note that for weak min-stability, it is enough to have the min-language only on the left, which gives a weaker condition.
%
%
\cref{thm:main,thm:weaklystable} hold for min-stability and weak min-stability,
respectively, as well (the proof of the min-versions are in fact a part of the proof of the Theorems~\ref{thm:main}~and~\ref{thm:weaklystable}).
%
%
%

%
The min-stability of a multi-equation system is then defined 
in the same way as before, different only in that it uses the min-stability at the nodes instead of  stability.
Namely, in \cref{alg:propagate}, the test $\langof{\autass(\sterm)} \subseteq \langof{\autass(\tterm)}$ on \lnref{line:inclusiontest} is replaced by 
$\minlangof{\autass(\sterm)} \subseteq \langof{\autass(\tterm)}$.  
We call this variant of the algorithm $\propagatemin$. 
Not only that this new algorithm is still partially correct, may terminate after less refinements, and uses a~cheaper test to detect termination,
but, mainly, it is complete in the SAT case: if there is a~solution, then it is guaranteed to terminate for 
any system, no matter whether chain-free or not.
Intuitively, the algorithm in a sense explores the words in the languages of the variables systematically, taking the words ordered by length, the shortest ones first.
Hence, besides that variants
of \cref{thm:soundness,thm:chain-free-compl} with $\propagatemin$ still hold, we
also have \cref{thm:sattermination}:
%

\vspace{-2mm}
\begin{restatable}{theorem}{thmSatTermination}
If $\constr$ is satisfiable, then $\propagatemin(\incl(\constreq), \constrautass)$ terminates.
\label{thm:sattermination}
\end{restatable}

\newcommand{\figpyexothers}[0]{
\begin{figure}[t]
  \begin{subfigure}[t]{0.47\linewidth}
  \begin{center}
  \includegraphics[width=\linewidth,keepaspectratio]{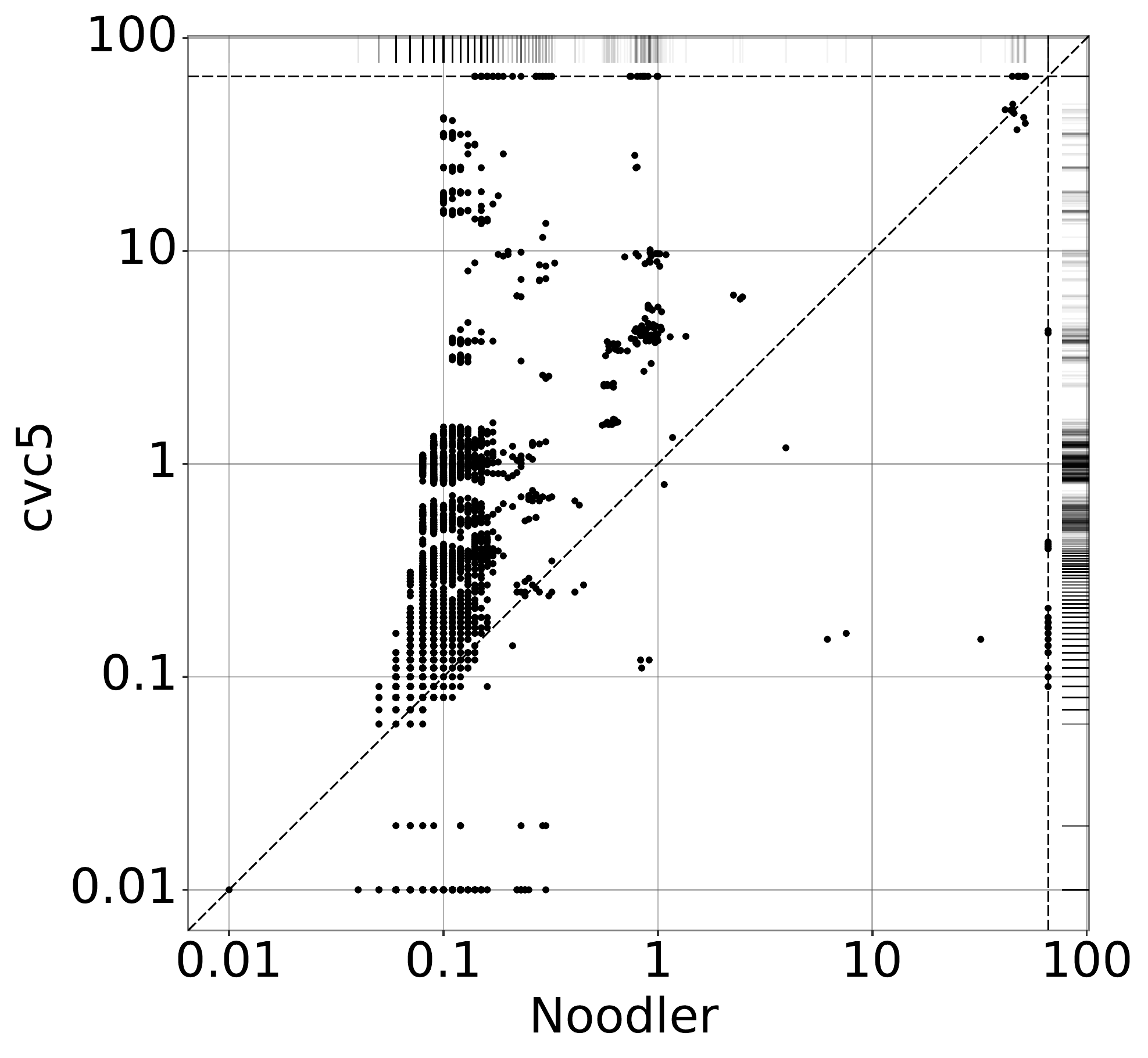}
  \end{center}
  \vspace{-5mm}
  \caption{\noodler vs.\ $\cvcv$.}
  \label{fig:pyex-cvc5}
  \end{subfigure}
  \hfill
  %
  \begin{subfigure}[t]{0.47\linewidth}
  \begin{center}
  \includegraphics[width=\linewidth,keepaspectratio]{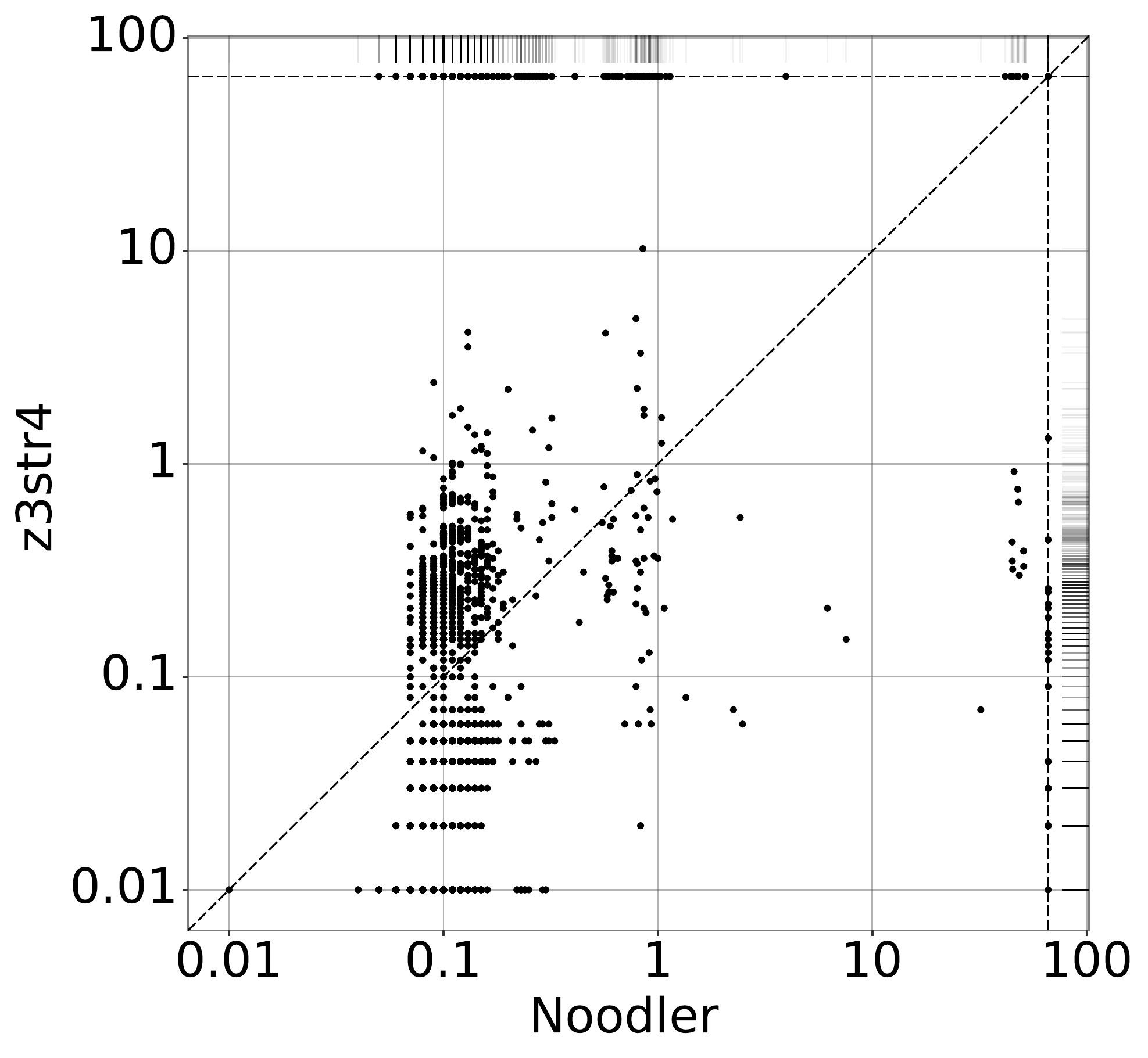}
  \end{center}
  \vspace{-5mm}
  \caption{\noodler vs.\ $\ziiistriv$.}
  \label{fig:pyex-cvc5}
  \end{subfigure}
  \vspace{-2mm}
  \caption{The performance of \noodler and other tools on $\pyexhard$.
    Times are given in seconds, axes are logarithmic.
    Dashed lines represent timeouts (60\,s).
    }
  \vspace*{-3mm}
\label{fig:scatter}
\end{figure}
}

\newcommand{\figCactus}[0]{
\begin{wrapfigure}[15]{r}{5.8cm}
\vspace*{-11.8mm}
\hspace*{-2mm}
\begin{minipage}{6cm}
\includegraphics[width=\linewidth,keepaspectratio]{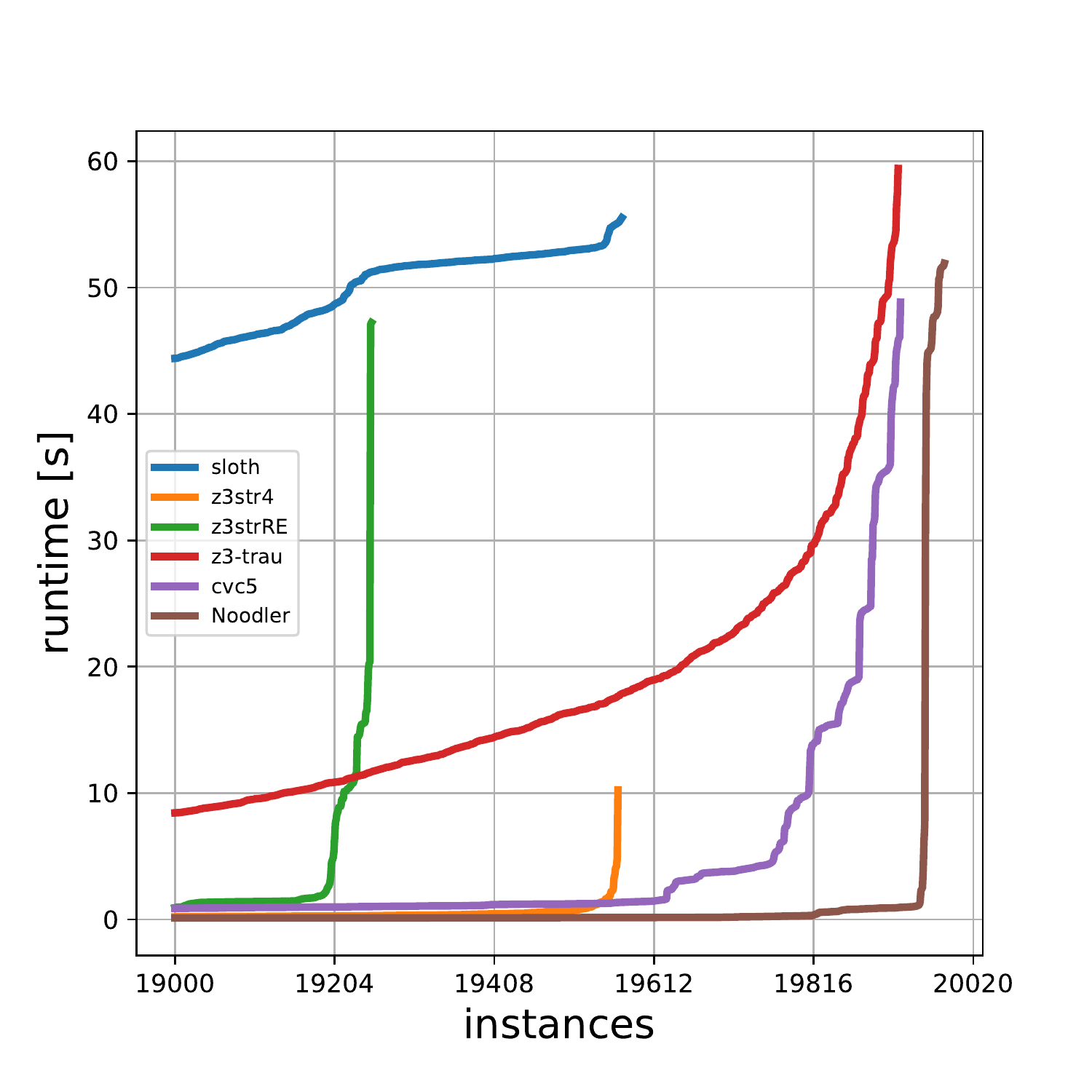}
\end{minipage}
\vspace*{-6mm}
\caption{Times for solving the hardest 1,023 formulae for the tools on \pyexhard}
\label{fig:cactus}
\end{wrapfigure}
}

\newcommand{\tabResults}[0]{
\begin{table}[t]
\caption{
  Results of experiments.
  For each benchmark and tool, we give the number of timeouts (``T/Os''), the
  total run time (in seconds), and the run time without timeouts
  (``time$-$T/O'').
  Best values are in \textbf{bold}.
}
\label{tab:results}
\hspace*{-4mm}
\resizebox{1.04\textwidth}{!}{
\newcolumntype{g}{>{\columncolor{Gray!30}}r}
\newcolumntype{f}{>{\columncolor{Gray!30}}l}
\newcolumntype{h}{>{\columncolor{Gray!30}}c}
\newcommand{\ping}[0]{\bf}
\begin{tabular}{lgggrrrgggrrr}
\toprule
                   & \multicolumn{3}{h}{\pyexhard (20,023)} & \multicolumn{3}{c}{\kaluzahard (897)} & \multicolumn{3}{h}{\strii (293)} & \multicolumn{3}{c}{\slog (1,896)}\\
                   & \multicolumn{1}{h}{T/Os} & \multicolumn{1}{h}{time} & \multicolumn{1}{h}{time$-$T/O} &
                     \multicolumn{1}{c}{T/Os} & \multicolumn{1}{c}{time} & \multicolumn{1}{c}{time$-$T/O} &
                     \multicolumn{1}{h}{T/Os} & \multicolumn{1}{h}{~time} & \multicolumn{1}{h}{time$-$T/O} &
                     \multicolumn{1}{c}{T/Os} & \multicolumn{1}{c}{time} & \multicolumn{1}{c}{time$-$T/O} \\
\midrule
\rowcolor{GreenYellow}
\noodler      & \ping 39 & \ping 5,266 & 2,926     & \ping 0 & \ping 46 & 46       & 3                       & \ping 198 & 18      & \ping 0                             & 165      & 165     \\
\ziii         & 2,802    & 178,078     & 9,958     & 207     & 15,360   & 2,940    & 149                     & 8,955     & 15      & 2                                   & 332      & 212     \\
\cvcv         & 112      & 12,523      & 5,803     & \ping 0 & 55       & 55       & 92                      & 5,525     & \ping 5 & \ping 0                             & \ping 14 & \ping 14\\
\ziiistriiire & 814      & 49,744      & 904       & 10      & 622      & 22       & 149                     & 8,972     & 32      & 55                                  & 4,247    & 947     \\
\ziiistriv    & 461      & 28,114      & \ping 454 & 17      & 1,039    & \ping 19 & 154                     & 9,267     & 27      & 208                                 & 16,508   & 4,028   \\
Z3-\trau      & 108      & 33,551      & 27,071    & \ping 0 & 201      & 201      & 10                      & 724       & 124     & 5                                   & 970      & 670     \\
\ostrich      & 2,979    & 214,846     & 36,106    & 111     & 14,912   & ~8,252   & 238                     & ~14,497    & 217     & 2                                   & 13,601   & 13,481  \\
\sloth        & 463      & ~371,373    & ~343,593  & \ping 0 & 3,195    & 3,195    & \multicolumn{3}{h}{N/A}                       & 202                                 & ~24,940   & ~12,820 \\
\retro        & 3,004    & 199,107     & 18,867    & 148     & 16,404   & 7,524    & \ping 1                 & 299       & 239     & \multicolumn{3}{c}{N/A}         \\
\bottomrule
\end{tabular}

}
\vspace{-4mm}
\end{table}
}


\vspace{-6.0mm}
\section{Experimental Evaluation}
\vspace{-1.0mm}


We implemented our algorithm in a~prototype string solver called
\noodler~\cite{Noodler} using Python and a~homemade C++ automata library for manipulating NFAs.
We compared the performance of \noodler with a~comprehensive selection of other
tools, namely,
\cvcv~\cite{cvc5} (version 1.0.1),
\ziii~\cite{z3} (version 4.8.14),
\ziiistriiire~\cite{Z3str3RE},
\ziiistriv~\cite{MoraBKNG21},
\ziiitrau~\cite{holik_popl_20},
\ostrich~\cite{AnthonyComplex2019},
\sloth~\cite{holik_string_2018}, and
\retro~\cite{ondravojtastrings20}.
In order to have a~meaningful comparison with compiled tools (\cvcv,
\ziii, \ziiistriiire, \ziiistriv, \ziiitrau), the reported time for \noodler
does not contain the startup time of the Python interpreter and the time taken
by loading libraries (this is a~constant of around~1.5\,s).
To be fair, one should take this into account when considering the time of
other interpreted tools, such as \ostrich, \sloth (both Java), and \retro
(Python). 
As can be seen from the results, it would, however, not significantly impact
the overall outcome.
The experiments were executed on a~workstation with an Intel Core i5
661 CPU at 3.33\,GHz
with 16\,GiB of RAM
running Debian GNU/Linux.
\mbox{The timeout was set to 60\,s.}

\vspace{-2mm}
\paragraph{Benchmarks.}
We consider the following benchmarks, having removed unsupported formulae
(i.e., formulae with length constraints or transducer operations).
\begin{itemize}
  \item
    \pyexhard (\cite{ondravojtastrings20}, 20,023 formulae): it comes from the
    \pyex benchmark~\cite{cvc417}, in particular, it is obtained from 967
    difficult instances that neither \cvciv nor \ziii could solve in 10\,s.
    \pyexhard then contains 20,023 conjunctions of word equations that \ziii's
    DPLL(T) algorithm sent to its string theory solver when trying to solve
    them.

  \item
    \kaluzahard (897 formulae): it is obtained from the \kaluza
    benchmark~\cite{saxena_symbolic} by taking hard formulae from its solution
    in a similar way for \pyexhard.

  \item
    \strii (\cite{LeH18}, 293 formulae) the original benchmark
    from~\cite{LeH18} contains 600 hand-crafted formulae including word
    equations and length constraints; the 307 formulae containing length
    constraints are removed.

  \item 
    \slog (\cite{fang-yu-circuits}, 1,896 formulae) contains 1,976 formulae
    obtained from real web applications using static analysis tools
    \jsa~\cite{ChristensenMS03} and \stranger~\cite{Stranger}.
    80~of these formulae contain transducer operations (e.g., ReplaceAll).
\end{itemize}
%
From the benchmarks, only \slog initially contains regular constraints.
Note that an interplay between equations and regular constraints happens in our algorithm even with pure equations on the input. 
Refinement of regular constraints is indeed the only means in which our algorithm accumulates information.
Complex regular constraints are generated by refinement steps from an initial
assignment of $\Sigma^*$ for every variable.
We also include useful constraints
in preprocessing steps, for instance, the equation $z = x a y$
where~$x$ and~$y$ do not occur elsewhere is substituted by $z \in \Sigma^* a \Sigma^*$.

\begin{table}[t]
\caption{
  Results of experiments.
  For each benchmark and tool, we give the number of timeouts (``T/Os''), the
  total run time (in seconds), and the run time without timeouts
  (``time$-$T/O'').
  Best values are in \textbf{bold}.
}
\label{tab:results}
\hspace*{-4mm}
\resizebox{1.04\textwidth}{!}{
\newcolumntype{g}{>{\columncolor{Gray!30}}r}
\newcolumntype{f}{>{\columncolor{Gray!30}}l}
\newcolumntype{h}{>{\columncolor{Gray!30}}c}
\newcommand{\ping}[0]{\bf}
\begin{tabular}{lgggrrrgggrrr}
\toprule
                   & \multicolumn{3}{h}{\pyexhard (20,023)} & \multicolumn{3}{c}{\kaluzahard (897)} & \multicolumn{3}{h}{\strii (293)} & \multicolumn{3}{c}{\slog (1,896)}\\
                   & \multicolumn{1}{h}{T/Os} & \multicolumn{1}{h}{time} & \multicolumn{1}{h}{time$-$T/O} &
                     \multicolumn{1}{c}{T/Os} & \multicolumn{1}{c}{time} & \multicolumn{1}{c}{time$-$T/O} &
                     \multicolumn{1}{h}{T/Os} & \multicolumn{1}{h}{~time} & \multicolumn{1}{h}{time$-$T/O} &
                     \multicolumn{1}{c}{T/Os} & \multicolumn{1}{c}{time} & \multicolumn{1}{c}{time$-$T/O} \\
\midrule
\rowcolor{GreenYellow}
\noodler      & \ping 39 & \ping 5,266 & 2,926     & \ping 0 & \ping 46 & 46       & 3                       & \ping 198 & 18      & \ping 0                             & 165      & 165     \\
\ziii         & 2,802    & 178,078     & 9,958     & 207     & 15,360   & 2,940    & 149                     & 8,955     & 15      & 2                                   & 332      & 212     \\
\cvcv         & 112      & 12,523      & 5,803     & \ping 0 & 55       & 55       & 92                      & 5,525     & \ping 5 & \ping 0                             & \ping 14 & \ping 14\\
\ziiistriiire & 814      & 49,744      & 904       & 10      & 622      & 22       & 149                     & 8,972     & 32      & 55                                  & 4,247    & 947     \\
\ziiistriv    & 461      & 28,114      & \ping 454 & 17      & 1,039    & \ping 19 & 154                     & 9,267     & 27      & 208                                 & 16,508   & 4,028   \\
Z3-\trau      & 108      & 33,551      & 27,071    & \ping 0 & 201      & 201      & 10                      & 724       & 124     & 5                                   & 970      & 670     \\
\ostrich      & 2,979    & 214,846     & 36,106    & 111     & 14,912   & ~8,252   & 238                     & ~14,497    & 217     & 2                                   & 13,601   & 13,481  \\
\sloth        & 463      & ~371,373    & ~343,593  & \ping 0 & 3,195    & 3,195    & \multicolumn{3}{h}{N/A}                       & 202                                 & ~24,940   & ~12,820 \\
\retro        & 3,004    & 199,107     & 18,867    & 148     & 16,404   & 7,524    & \ping 1                 & 299       & 239     & \multicolumn{3}{c}{N/A}         \\
\bottomrule
\end{tabular}

}
\vspace{-4mm}
\end{table}

\paragraph{Results.}
The results of experiments are given in \cref{tab:results}.
For each benchmark, we list the number of timeouts (i.e., unsolved formulae), the total run time (including timeouts), and also the run time on the successfully decided formulae.
The results show that from all tools, \noodler has the lowest number of timeouts
on the aggregation of all benchmarks (42 timeouts in total) and also on each
individual benchmark
(except \strii where it is the second lowest, 3~against~1).
Furthermore, in all benchmarks except \slog, \noodler is faster than other tools (and for \slog it is the second).
The results for \sloth on \strii are omitted because \sloth was incorrect on
this benchmark (the benchmark is not straight-line) and the results for \retro
on $\slog$ are omitted because \retro does not support regular constraints.

\figCactus   
In \cref{fig:scatter}, we provide scatter plots comparing the run times of
\noodler with the best competitors, \cvcv and \ziiistriv, on the \pyexhard
benchmark (scatter plots for the other benchmarks are less interesting and can
be found in \cref{sec:det-results}).
We can see that there is indeed a~large number of benchmarks where \noodler is
faster than both competitors (and that the performance of \noodler is more
stable, which may be caused by the heuristics in the other tools not
always working well).
Notice that \noodler and \cvcv are on this benchmark complementary: they have
both some timeouts, but each formula is solved by at least one of the tools.

Moreover, in \cref{fig:cactus}, we provide a~graph showing times needed to solve 1,023 most difficult formulae for the tools on the \pyexhard benchmark.


\figpyexothers   


\vspace{-2mm}
\paragraph{Discussion.}
The results of the experiments show that our algorithm (even in its prototype
implementation in Python) can beat well established solvers such as \cvcv,
\ziii, and \ziiistriv.
In particular, it can solve more benchmarks, and also the average time for
(successfully) solving a~benchmark is low (as witnessed by the ``time$-$T/O''
column in~\cref{tab:results}).
The scatter plots also show that it is often complementary to other solvers.

\vspace*{-2mm}
\section{Related Work}
\vspace*{-1mm}
\enlargethispage{2mm}

Our algorithm is an improvement of the automata-based algorithm first proposed in~\cite{AutomataSplitting},
which is, at least in part, used as the basis of several string solvers, namely,
\norn~\cite{%
AutomataSplitting,%
Norn,%
ChainFree%
},
\trau~\cite{%
holik_popl_20,%
Trau,%
Flatten,%
Notsubstring%
},
\ostrich~\cite{%
LB16,%
AnthonyReplaceAll2018,%
AnthonyComplex2019%
},
%
and \ziiistriiire~\cite{%
Z3str3RE%
}.
The original algorithm first transforms equations to the disjunction of their solved forms \cite{solvedform} through generating alignments of variable boundaries on the equation sides (essentially an incomplete version of Makanin's algorithm).  
%
Second, it eliminates concatenation from regular constraints
by \emph{automata splitting}.
The algorithm replaces $x\concat y\in L$ by a disjunction of cases 
$x\in L_x \land y\in L_y$, one case for each state of $L$'s automaton.
Each disjunct later entails testing emptiness of $L_x\cap \lass(x)$ and
$L_y\cap \lass(y)$ by the automata product construction. 
\trau uses this algorithm within an unsatisfiability check. 
\trau's main solution finding algorithm also performs a~step similar to our
refinement, though with languages underapproximated as arithmetic formulae
(representing their Parikh images). 
\sloth~\cite{holik_popl_20} implements a compact version of automata splitting through alternating automata. 
\ostrich has a way of avoiding the variable boundary alignment for the straight-line formulae, although still uses it outside of it.
\ziiistriiire optimises the algorithm of~\cite{AutomataSplitting} heavily by the use of length-aware heuristics. 

The two levels of disjunctive branching (transformation into solved form and
automata splitting) are costly. 
For instance, for $xyx=zu \land z\in a(ba)^* \land u \in (baba)^*a$
(a~subformula of the example in \cref{sec:overview}),  
there would be 14 alignments/solved forms, e.g.~those characterised using
lengths as follows:
(1)~$|zu|=0$; 
(2)~$|y|=|zu|$; 
(3)~$|x|<|z|,|y|=0$; 
(4)~$|xy|<z,|y|>0$;
(5)~$|x|<|z|,|xy|>z$;
\ldots
In the case~(5) alone---corresponding to the solved form $z=z_1z_2,u=u_1z_1,
x=z_1,y=z_2u_1$---automata splitting would generate 15 cases from $z_1z_2 \in
\lass(z)$ and $u_1u_2\in \lass(u)$, each entailing one intersection emptiness
check (the NFAs for~$z$ and~$u$ have~3 and~5 states respectively).
There would be about a hundred of such cases overall.
On the contrary, our algorithm generates only 9 of equivalent cases, 7 if optimised (see \cref{sec:overview}).

Our algorithm has an advantage also over pure automata splitting, irrespective of aligning equations. 
For instance, consider the constraint $xyx\in L \land x\in \lass(x) \land y\in \lass(y)$. 
Automata splitting generates a~disjunction of $n^2$ constraints $x \in L_x \land y\in L_y$, with~$n$ being the number of states of the automaton for~$L$, 
each constraint with emptiness checks for $\lass(x)\cap L_x$ and $\lass(y)\cap L_y$. 
Our algorithm avoids generating much of these cases by intersecting with the languages of $\lass(x)$ and $\lass(y)$ early---the construction of $\lass(x)\concat\lass(y)\concat\lass(x)$ prunes much of $L$'s automaton immediately.  
For instance, if $L=(ab)^*a^+ (abcd)^*$ (its NFA has 7~states) and $\lass(x) =
(a+b)^*$, automata splitting explores $7^2 = 49$ cases while our algorithm
explores 9 (7~when optimised) of these cases---it would compute the same
product and noodles as in \cref{sec:overview}, essentially ignoring the
disjunct $(abcd)^*$ of $L$.

Approaches and tools for string solving are numerous and diverse, 
with various representations of constraints, algorithms, or sorts of inputs.
Many approaches use automata, e.g.,
\stranger~\cite{Stranger,fmsd14,yu2011}, 
\norn~\cite{AutomataSplitting,Norn},
\ostrich~\cite{LB16,AnthonyReplaceAll2018,AnthonyComplex2019,AnthonyRegex2022,AnthonyInteger2020},
\trau~\cite{ChainFree,Trau,Flatten,Notsubstring},
\sloth~\cite{holik_popl_20},
\slog~\cite{fang-yu-circuits},
Slent~\cite{slent},
\ziiistriiire~\cite{Z3str3RE},
\retro~\cite{ondravojtastrings20},
ABC~\cite{ABCpaper,ABCtool},
Qzy~\cite{arlen}, or
BEK~\cite{BEK}.
Around word equations are centered tools such as CVC4/5~\cite{cvc4_string14,tinelli-fmsd16,tinelli-hotsos16,tinelli-frocos16,cvc417,cvc422,cvc420},
\ziii~\cite{BTV09,z3},
S3~\cite{S3},
\kepler~\cite{LeH18},
StrSolve~\cite{HW12},
Woorpje~\cite{DayEKMNP19};
bit vectors are (among other things) used in
Z3Str/2/3/4~\cite{Z3-str,Z3Str3,Z3str4,Z3-str15},
HAMPI~\cite{HAMPI};
PASS~uses arrays \cite{PASS};
G-strings~\cite{gstrings} and GECODE+S~\cite{gecode+s} use SAT-solving.
Most of these tools and methods handle much wider range of string constraints than equations and regular constraints.
%
Our algorithm is not a complete alternative but a promising basis that could improve some of the existing solvers and become a core of a new one.
With regard to equations and regular constraints, the fragment of chain-free constraints~\cite{ChainFree} that we handle, 
handled also by \trau, 
is the largest for which any string solvers offers formal completeness guarantees, 
with the exception of quadratic equations, handled, e.g., by~\cite{ondravojtastrings20,LeH18}, which are incomparable but of a smaller practical relevance (although some tools actually implement Nielsen's algorithm~\cite{nielsen1917} to handle simple quadratic cases).
The other solvers guarantee completeness on smaller fragments, notably that of \ostrich (straight-line), \norn, and \ziiistriiire; 
or use incomplete heuristics that work in practice (giving up guarantees of termination, over or under-approximating by various means).
Most string solvers tend to avoid handling regular expressions, by means of postponing them as much as possible or 
abstracting them into arithmetic/length and other constraints (e.g. \trau, \ziiistriiire, \ziiistriv, CVC4/5, S3). 
A major point of our work is that taking the opposite approach may work even better when automata are approached from the right angle and implemented carefully, though, heuristics that utilise length information or Parikh images would most probably speed up our algorithm as well.
The main selling point of our approach is its efficiency compared to the others, demonstrated on benchmark sets used in other works.
\vspace{-3.0mm}
\section{Conclusion and Future Work}\label{sec:label}
\vspace{-1.0mm}
\enlargethispage{1mm}
We have presented a new algorithm for solving a fragment of word equations with
regular constraints, complete in $\sat$ cases and for the chain-free fragment.  
It is based on a tight interconnection of equations with regular constraints and built around a novel characterisation of satisfiability of a string constraint through the notion of stability.
We have experimentally shown that the algorithm is very competitive with existing solutions, better especially on difficult examples.

We plan to continue from here towards a complete string solver. 
This involves including other types of constraints and coming up with a mature and optimised implementation. 
The core algorithm might also be optimised by using a more compact automata representation of noodles that would eliminate redundancies.  

%
%
%

\bibliographystyle{splncs}
\bibliography{literature}

\begin{thebibliography}{10}

\bibitem{OWASP13}
OWASP:
\newblock Top 10.
\newblock \url{https://www.owasp.org/images/f/f8/OWASP_Top_10_-_2013.pdf}
  (2013)

\bibitem{OWASP17}
OWASP:
\newblock Top 10.
\newblock \url{https://owasp.org/www-project-top-ten/2017/} (2017)

\bibitem{OWASP21}
OWASP:
\newblock Top 10.
\newblock \url{https://owasp.org/Top10/} (2021)

\bibitem{hadarean_mosca}
{Liana Hadarean}:
\newblock String solving at {Amazon}.
\newblock \url{https://mosca19.github.io/program/index.html} (2019) Presented
  at MOSCA'19.

\bibitem{AltBHS22}
Alt, L., Blicha, M., Hyv{\"{a}}rinen, A.E.J., Sharygina, N.:
\newblock {SolCMC}: Solidity compiler's model checker.
\newblock In Shoham, S., Vizel, Y., eds.: Computer Aided Verification - 34th
  International Conference, {CAV} 2022, Haifa, Israel, August 7-10, 2022,
  Proceedings, Part {I}. Volume 13371 of Lecture Notes in Computer Science.,
  Springer (2022)  325--338

\bibitem{plandowski99}
Plandowski, W.:
\newblock Satisfiability of word equations with constants is in {NEXPTIME}.
\newblock In: Proceedings of the Thirty-First Annual ACM Symposium on Theory of
  Computing. STOC '99, New York, NY, USA, Association for Computing Machinery
  (1999)  721–725

\bibitem{jez2016}
Je\.{z}, A.:
\newblock Recompression: A simple and powerful technique for word equations.
\newblock J. ACM \textbf{63}(1) (feb 2016)

\bibitem{makanin}
Makanin, G.S.:
\newblock The problem of solvability of equations in a free semigroup.
\newblock Matematicheskii Sbornik \textbf{32}(2) (1977)  147--236 {(in
  Russian).}

\bibitem{LB16}
Lin, A.W., Barcel{\'{o}}, P.:
\newblock String solving with word equations and transducers: {T}owards a logic
  for analysing mutation {XSS}.
\newblock In: POPL'16, {{ACM} Trans. Comput. Log.} (2016)  123--136

\bibitem{ChainFree}
Abdulla, P.A., Atig, M.F., Diep, B.P., Hol{\'{\i}}k, L., Janku, P.:
\newblock Chain-free string constraints.
\newblock In Chen, Y., Cheng, C., Esparza, J., eds.: Automated Technology for
  Verification and Analysis - 17th International Symposium, {ATVA} 2019,
  Taipei, Taiwan, October 28-31, 2019, Proceedings. Volume 11781 of Lecture
  Notes in Computer Science., Springer (2019)  277--293

\bibitem{nielsen1917}
Nielsen, J.:
\newblock Die isomorphismen der allgemeinen, unendlichen gruppe mit zwei
  erzeugenden.
\newblock Mathematische Annalen \textbf{78}(1) (1917)  385--397

\bibitem{Norn}
Abdulla, P.A., Atig, M.F., Chen, Y., Hol{\'{\i}}k, L., Rezine, A.,
  R{\"{u}}mmer, P., Stenman, J.:
\newblock Norn: An {SMT} solver for string constraints.
\newblock In Kroening, D., Pasareanu, C.S., eds.: Computer Aided Verification -
  27th International Conference, {CAV} 2015, San Francisco, CA, USA, July
  18-24, 2015, Proceedings, Part {I}. Volume 9206 of Lecture Notes in Computer
  Science., Springer (2015)  462--469

\bibitem{AnthonyReplaceAll2018}
Chen, T., Chen, Y., Hague, M., Lin, A.W., Wu, Z.:
\newblock What is decidable about string constraints with the replaceall
  function.
\newblock Proc. {ACM} Program. Lang. \textbf{2}({POPL}) (2018)  3:1--3:29

\bibitem{solvedform}
Ganesh, V., Minnes, M., Solar{-}Lezama, A., Rinard, M.C.:
\newblock Word equations with length constraints: What's decidable?
\newblock In Biere, A., Nahir, A., Vos, T.E.J., eds.: Hardware and Software:
  Verification and Testing - 8th International Haifa Verification Conference,
  {HVC} 2012, Haifa, Israel, November 6-8, 2012. Revised Selected Papers.
  Volume 7857 of Lecture Notes in Computer Science., Springer (2012)  209--226

\bibitem{Aziz93}
Aziz, A., Singhal, V., Swamy, G., Brayton, R.K.:
\newblock Minimizing interacting finite state machines.
\newblock Technical Report UCB/ERL M93/68, EECS Department, University of
  California, Berkeley (Sep 1993)

\bibitem{HHK95}
Henzinger, M., Henzinger, T., Kopke, P.:
\newblock Computing simulations on finite and infinite graphs.
\newblock In: Proceedings of IEEE 36th Annual Foundations of Computer Science.
  (1995)  453--462

\bibitem{Noodler}
Blahoudek, F., Chen, Y., Chocholat\'{y}, D., Havlena, V., Hol\'{i}k, L.,
  Leng\'{a}l, O., S\'{i}\v{c}, J.:
\newblock Noodler.
\newblock \url{https://github.com/vhavlena/Noodler} (2022)

\bibitem{cvc5}
Barbosa, H., Barrett, C., Brain, M., Kremer, G., Lachnitt, H., Mann, M.,
  Mohamed, A., Mohamed, M., Niemetz, A., N{\"o}tzli, A., Ozdemir, A., Preiner,
  M., Reynolds, A., Sheng, Y., Tinelli, C., Zohar, Y.:
\newblock cvc5: A versatile and industrial-strength smt solver.
\newblock In Fisman, D., Rosu, G., eds.: Tools and Algorithms for the
  Construction and Analysis of Systems, Cham, Springer International Publishing
  (2022)  415--442

\bibitem{z3}
de~Moura, L., Bj{\o}rner, N.:
\newblock Z3: An efficient smt solver.
\newblock In Ramakrishnan, C.R., Rehof, J., eds.: Tools and Algorithms for the
  Construction and Analysis of Systems, Berlin, Heidelberg, Springer Berlin
  Heidelberg (2008)  337--340

\bibitem{Z3str3RE}
Berzish, M., Kulczynski, M., Mora, F., Manea, F., Day, J.D., Nowotka, D.,
  Ganesh, V.:
\newblock An {SMT} solver for regular expressions and linear arithmetic over
  string length.
\newblock In Silva, A., Leino, K.R.M., eds.: Computer Aided Verification - 33rd
  International Conference, {CAV} 2021, Virtual Event, July 20-23, 2021,
  Proceedings, Part {II}. Volume 12760 of Lecture Notes in Computer Science.,
  Springer (2021)  289--312

\bibitem{MoraBKNG21}
Mora, F., Berzish, M., Kulczynski, M., Nowotka, D., Ganesh, V.:
\newblock Z3str4: {A} multi-armed string solver.
\newblock In Huisman, M., Pasareanu, C.S., Zhan, N., eds.: Formal Methods -
  24th International Symposium, {FM} 2021, Virtual Event, November 20-26, 2021,
  Proceedings. Volume 13047 of Lecture Notes in Computer Science., Springer
  (2021)  389--406

\bibitem{holik_popl_20}
Abdulla, P.A., Atig, M.F., Chen, Y., Diep, B.P., Dolby, J., Janku, P., Lin, H.,
  Hol{\'{\i}}k, L., Wu, W.:
\newblock Efficient handling of string-number conversion.
\newblock In: Proc. of {PLDI}'20, {ACM} (2020)  943--957

\bibitem{AnthonyComplex2019}
Chen, T., Hague, M., Lin, A.W., R{\"{u}}mmer, P., Wu, Z.:
\newblock Decision procedures for path feasibility of string-manipulating
  programs with complex operations.
\newblock Proc. {ACM} Program. Lang. \textbf{3}({POPL}) (2019)  49:1--49:30

\bibitem{holik_string_2018}
Hol{\'{\i}}k, L., Janku, P., Lin, A.W., R{\"{u}}mmer, P., Vojnar, T.:
\newblock String constraints with concatenation and transducers solved
  efficiently.
\newblock Proc. {ACM} Program. Lang. \textbf{2}({POPL}) (2018)  4:1--4:32

\bibitem{ondravojtastrings20}
Chen, Y., Havlena, V., Leng{\'{a}}l, O., Turrini, A.:
\newblock A symbolic algorithm for the case-split rule in string constraint
  solving.
\newblock In d.~S.~Oliveira, B.C., ed.: Programming Languages and Systems -
  18th Asian Symposium, {APLAS} 2020, Fukuoka, Japan, November 30 - December 2,
  2020, Proceedings. Volume 12470 of Lecture Notes in Computer Science.,
  Springer (2020)  343--363

\bibitem{cvc417}
Reynolds, A., Woo, M., Barrett, C., Brumley, D., Liang, T., Tinelli, C.:
\newblock Scaling up {DPLL(T)} string solvers using context-dependent
  simplification.
\newblock In Majumdar, R., Kun{\v{c}}ak, V., eds.: Computer Aided Verification,
  Cham, Springer International Publishing (2017)  453--474

\bibitem{saxena_symbolic}
Saxena, P., Akhawe, D., Hanna, S., Mao, F., McCamant, S., Song, D.:
\newblock A symbolic execution framework for {J}ava{S}cript.
\newblock In: {SP}'10, IEEE Computer Society (2010)  513--528

\bibitem{LeH18}
Le, Q.L., He, M.:
\newblock A decision procedure for string logic with quadratic equations,
  regular expressions and length constraints.
\newblock In Ryu, S., ed.: Programming Languages and Systems, Cham, Springer
  International Publishing (2018)  350--372

\bibitem{fang-yu-circuits}
Wang, H., Tsai, T., Lin, C., Yu, F., Jiang, J.R.:
\newblock String analysis via automata manipulation with logic circuit
  representation.
\newblock In: {CAV}'16. Volume 9779 of LNCS., Springer (2016)  241--260

\bibitem{ChristensenMS03}
Christensen, A.S., M{\o}ller, A., Schwartzbach, M.I.:
\newblock Precise analysis of string expressions.
\newblock In Cousot, R., ed.: Static Analysis, 10th International Symposium,
  {SAS} 2003, San Diego, CA, USA, June 11-13, 2003, Proceedings. Volume 2694 of
  Lecture Notes in Computer Science., Springer (2003)  1--18

\bibitem{Stranger}
Yu, F., Alkhalaf, M., Bultan, T.:
\newblock Stranger: {A}n automata-based string analysis tool for {PHP}.
\newblock In: {TACAS}'10. Volume 6015 of LNCS., Springer (2010)  154--157

\bibitem{AutomataSplitting}
Abdulla, P.A., Atig, M.F., Chen, Y., Hol{\'{\i}}k, L., Rezine, A.,
  R{\"{u}}mmer, P., Stenman, J.:
\newblock String constraints for verification.
\newblock In Biere, A., Bloem, R., eds.: Computer Aided Verification - 26th
  International Conference, {CAV} 2014, Held as Part of the Vienna Summer of
  Logic, {VSL} 2014, Vienna, Austria, July 18-22, 2014. Proceedings. Volume
  8559 of Lecture Notes in Computer Science., Springer (2014)  150--166

\bibitem{Trau}
Abdulla, P.A., Atig, M.F., Chen, Y., Diep, B.P., Hol{\'{\i}}k, L., Rezine, A.,
  R{\"{u}}mmer, P.:
\newblock Trau: {SMT} solver for string constraints.
\newblock In Bj{\o}rner, N.S., Gurfinkel, A., eds.: 2018 Formal Methods in
  Computer Aided Design, {FMCAD} 2018, Austin, TX, USA, October 30 - November
  2, 2018, {IEEE} (2018)  1--5

\bibitem{Flatten}
Abdulla, P.A., Atig, M.F., Chen, Y., Diep, B.P., Hol{\'{\i}}k, L., Rezine, A.,
  R{\"{u}}mmer, P.:
\newblock Flatten and conquer: a framework for efficient analysis of string
  constraints.
\newblock In Cohen, A., Vechev, M.T., eds.: Proceedings of the 38th {ACM}
  {SIGPLAN} Conference on Programming Language Design and Implementation,
  {PLDI} 2017, Barcelona, Spain, June 18-23, 2017, {ACM} (2017)  602--617

\bibitem{Notsubstring}
Abdulla, P.A., Atig, M.F., Chen, Y., Diep, B.P., Hol{\'{\i}}k, L., Hu, D.,
  Tsai, W., Wu, Z., Yen, D.:
\newblock Solving not-substring constraint with flat abstraction.
\newblock In Oh, H., ed.: Programming Languages and Systems - 19th Asian
  Symposium, {APLAS} 2021, Chicago, IL, USA, October 17-18, 2021, Proceedings.
  Volume 13008 of Lecture Notes in Computer Science., Springer (2021)  305--320

\bibitem{fmsd14}
Yu, F., Alkhalaf, M., Bultan, T., Ibarra, O.H.:
\newblock Automata-based symbolic string analysis for vulnerability detection.
\newblock Formal Methods in System Design \textbf{44}(1) (2014)  44--70

\bibitem{yu2011}
Yu, F., Bultan, T., Ibarra, O.H.:
\newblock Relational string verification using multi-track automata.
\newblock Int. J. Found. Comput. Sci. \textbf{22}(8) (2011)  1909--1924

\bibitem{AnthonyRegex2022}
Chen, T., Flores{-}Lamas, A., Hague, M., Han, Z., Hu, D., Kan, S., Lin, A.W.,
  R{\"{u}}mmer, P., Wu, Z.:
\newblock Solving string constraints with regex-dependent functions through
  transducers with priorities and variables.
\newblock Proc. {ACM} Program. Lang. \textbf{6}({POPL}) (2022)  1--31

\bibitem{AnthonyInteger2020}
Chen, T., Hague, M., He, J., Hu, D., Lin, A.W., R{\"{u}}mmer, P., Wu, Z.:
\newblock A decision procedure for path feasibility of string manipulating
  programs with integer data type.
\newblock In Hung, D.V., Sokolsky, O., eds.: Automated Technology for
  Verification and Analysis - 18th International Symposium, {ATVA} 2020, Hanoi,
  Vietnam, October 19-23, 2020, Proceedings. Volume 12302 of Lecture Notes in
  Computer Science., Springer (2020)  325--342

\bibitem{slent}
Wang, H.E., Chen, S.Y., Yu, F., Jiang, J.H.R.:
\newblock A symbolic model checking approach to the analysis of string and
  length constraints.
\newblock In: Proceedings of the 33rd ACM/IEEE International Conference on
  Automated Software Engineering. ASE 2018, New York, NY, USA, Association for
  Computing Machinery (2018)  623–633

\bibitem{ABCpaper}
Aydin, A., Bang, L., Bultan, T.:
\newblock Automata-based model counting for string constraints.
\newblock In Kroening, D., P{\u{a}}s{\u{a}}reanu, C.S., eds.: Computer Aided
  Verification, Cham, Springer International Publishing (2015)  255--272

\bibitem{ABCtool}
Bultan, T., contributors:
\newblock {ABC} string solver

\bibitem{arlen}
Cox, A., Leasure, J.:
\newblock Model checking regular language constraints.
\newblock CoRR \textbf{abs/1708.09073} (2017)

\bibitem{BEK}
Hooimeijer, P., Livshits, B., Molnar, D., Saxena, P., Veanes, M.:
\newblock Fast and precise sanitizer analysis with {BEK}.
\newblock In: {USENIX} Security Symposium 2011, {USENIX} Association (2011)

\bibitem{cvc4_string14}
Liang, T., Reynolds, A., Tinelli, C., Barrett, C., Deters, M.:
\newblock A {DPLL(T)} theory solver for a theory of strings and regular
  expressions.
\newblock In Biere, A., Bloem, R., eds.: Computer Aided Verification, Cham,
  Springer International Publishing (2014)  646--662

\bibitem{tinelli-fmsd16}
Liang, T., Reynolds, A., Tsiskaridze, N., Tinelli, C., Barrett, C., Deters, M.:
\newblock An efficient {SMT} solver for string constraints.
\newblock Formal Methods in System Design \textbf{48}(3) (2016)  206--234

\bibitem{tinelli-hotsos16}
Barrett, C.W., Tinelli, C., Deters, M., Liang, T., Reynolds, A., Tsiskaridze,
  N.:
\newblock Efficient solving of string constraints for security analysis.
\newblock In: {HotSoS}'16, {{ACM} Trans. Comput. Log.} (2016)  4--6

\bibitem{tinelli-frocos16}
Liang, T., Tsiskaridze, N., Reynolds, A., Tinelli, C., Barrett, C.:
\newblock A decision procedure for regular membership and length constraints
  over unbounded strings.
\newblock In: FroCoS'15. Volume 9322 of LNCS., Springer (2015)  135--150

\bibitem{cvc422}
N{\"o}tzli, A., Reynolds, A., Barbosa, H., Barrett, C., Tinelli, C.:
\newblock Even faster conflicts and lazier reductions for string solvers.
\newblock In Shoham, S., Vizel, Y., eds.: Computer Aided Verification, Cham,
  Springer International Publishing (2022)  205--226

\bibitem{cvc420}
Reynolds, A., Notzlit, A., Barrett, C., Tinelli, C.:
\newblock Reductions for strings and regular expressions revisited.
\newblock In: 2020 Formal Methods in Computer Aided Design (FMCAD). (2020)
  225--235

\bibitem{BTV09}
Bj{\o}rner, N., Tillmann, N., Voronkov, A.:
\newblock Path feasibility analysis for string-manipulating programs.
\newblock In: {TACAS}'09. Volume 5505 of LNCS., Springer (2009)  307--321

\bibitem{S3}
Trinh, M., Chu, D., Jaffar, J.:
\newblock {S3:} {A} symbolic string solver for vulnerability detection in web
  applications.
\newblock In: CCS, {{ACM} Trans. Comput. Log.} (2014)  1232--1243

\bibitem{HW12}
Hooimeijer, P., Weimer, W.:
\newblock {StrSolve}: {S}olving string constraints lazily.
\newblock Autom. Softw. Eng. \textbf{19}(4) (2012)  531--559

\bibitem{DayEKMNP19}
Day, J.D., Ehlers, T., Kulczynski, M., Manea, F., Nowotka, D., Poulsen, D.B.:
\newblock On solving word equations using {SAT}.
\newblock In Filiot, E., Jungers, R.M., Potapov, I., eds.: Reachability
  Problems - 13th International Conference, {RP} 2019, Brussels, Belgium,
  September 11-13, 2019, Proceedings. Volume 11674 of Lecture Notes in Computer
  Science., Springer (2019)  93--106

\bibitem{Z3-str}
Zheng, Y., Zhang, X., Ganesh, V.:
\newblock Z3-str: {A} {Z}3-based string solver for web application analysis.
\newblock In: ESEC/FSE'13, {{ACM} Trans. Comput. Log.} (2013)  114--124

\bibitem{Z3Str3}
Berzish, M., Ganesh, V., Zheng, Y.:
\newblock Z3str3: A string solver with theory-aware heuristics.
\newblock In: 2017 Formal Methods in Computer Aided Design (FMCAD). (2017)
  55--59

\bibitem{Z3str4}
{Berzish, Murphy}:
\newblock Z3str4: A Solver for Theories over Strings.
\newblock PhD thesis (2021)

\bibitem{Z3-str15}
Zheng, Y., Ganesh, V., Subramanian, S., Tripp, O., Dolby, J., Zhang, X.:
\newblock Effective search-space pruning for solvers of string equations,
  regular expressions and length constraints.
\newblock In Kroening, D., P{\u{a}}s{\u{a}}reanu, C.S., eds.: Computer Aided
  Verification, Cham, Springer International Publishing (2015)  235--254

\bibitem{HAMPI}
Kiezun, A., Ganesh, V., Artzi, S., Guo, P.J., Hooimeijer, P., Ernst, M.D.:
\newblock {HAMPI:} {A} solver for word equations over strings, regular
  expressions, and context-free grammars.
\newblock {{ACM} Trans. Comput. Log.} \textbf{21}(4) (2012)  25:1--25:28

\bibitem{PASS}
Li, G., Ghosh, I.:
\newblock Pass: String solving with parameterized array and interval automaton.
\newblock In Bertacco, V., Legay, A., eds.: Hardware and Software: Verification
  and Testing, Cham, Springer International Publishing (2013)  15--31

\bibitem{gstrings}
Amadini, R., Gange, G., Stuckey, P.J., Tack, G.:
\newblock A novel approach to string constraint solving.
\newblock In Beck, J.C., ed.: Principles and Practice of Constraint
  Programming, Cham, Springer International Publishing (2017)  3--20

\bibitem{gecode+s}
Scott, J.D., Flener, P., Pearson, J., Schulte, C.:
\newblock Design and implementation of bounded-length sequence variables.
\newblock In Salvagnin, D., Lombardi, M., eds.: Integration of AI and OR
  Techniques in Constraint Programming, Cham, Springer International Publishing
  (2017)  51--67

\end{thebibliography}

\newpage
\appendix

\section{Proof of \cref{thm:main}}
\label{section:proof}


Assume a single-equation system $\constr\colon y_1\dots y_n=z_1\dots z_m\land \bigwedge_{x\in \vars}x\in \constrlass(x)$ 
where $\vars = \{ y_1, \dots, y_n, z_1, \dots, z_m \}$. In this section, we sometimes use $(e,\lass)$ to denote the single-equation system. 
%
%
Recall that $\Min{L}$ denotes the words that have the minimum length in the
language $L$. 
%
%



For a word $v = a_1\cdots a_\ell$ and two indices $0<i\leq j\leq \ell$, we write $v[i:j]$ to denote
the \emph{infix} $a_i\cdots a_j$ of $v$ between $i$ and $j$.
We also write $v[i]$ to denote the $i$th character $a_i$.

Most of the proof is concerned with proving a simpler version of \cref{thm:main},
in which each language $\constrlass(x),x\in \X$  consists of words of the same minimal length $\ell_x$. 
The simplified version of \cref{thm:main} is formulated as the following lemma.
\begin{lemma}\label{lemma:fixed}
If $\sys$ is strongly min-stable, then it has a solution.
\end{lemma}
The next lemma explains why \cref{lemma:fixed} implies \cref{thm:main}.
Let $\Min{L}$ denote the words that have the minimum length in the language $L$.
\begin{lemma}\label{lemma:general->fixed}
If $\constrlass(y_1)\cdots \constrlass(y_n) = \constrlass(z_1)\cdots \constrlass(z_m)$, then \\ $\Min{\constrlass}(y_1)\cdots \Min{\constrlass}(y_n) =  \Min{\constrlass}(z_1)\cdots \Min{\constrlass}(z_m)$.
\end{lemma}
\begin{proof}(Idea)
From the equality of the languages, every concatenation $u_1\cdots u_n$ of minimum length words on the left must have an equivalent counterpart $v_1\cdots v_m$ on the right, and vice versa.
The words on the right must be minimal too since otherwise one could compose a shorter word on the right, and its counterpart on the left would be shorter then $u_1\cdots u_n$, which contradicts that $u_1\cdots u_n$ are minimal.
\qed
\end{proof}

\thmMain*
\begin{proof}
  Corollary of \cref{lemma:fixed,lemma:general->fixed}.
  \qed
\end{proof}

To prove \cref{lemma:fixed}, we need to introduce some additional notation and prove some auxiliary Lemmas.

\subsection{Proof of \cref{lemma:fixed}}
 Let us fix a min-length system $\sys = (e,\constrlass)$ where $\constrlass(x)$ contains the shortest strings with the length $\ell_x$. Unless stated otherwise, $\constr$ will implicitly be the system of concern in this section.
 Let $\ell_\sys$ be the \emph{length of $\sys$}, the sum
 $\ell_\constr = \sum_{i=1}^n\ell_{y_i} = \sum_{i=1}^m\ell_{z_i}$. We call a number $p$ between
 $1$ and $\ell_\sys$ a \emph{position}, and we call the elements of the set $\atoms =
 \{(x,i)\mid x\in \X,1\leq i\leq\ell_x\}$
 \emph{atoms}.
We say that
$(y_k,j)$ with $k\leq n,j\leq \ell_{y_k}$ is the \emph{left atom at the position} $p = \sum_{i=1}^{k-1}\ell_{y_i}+j$ and write $\atoml(p) = (y_k,j)$.
Conversely, $(z_k,j)$ with $k\leq m,j\leq \ell_{z_k}$ is the \emph{right atom at the position}
$p = \sum_{i=1}^{k-1}\ell_{z_i}+j$,
written $\atomr(p) = (z_k,j)$, and $\atom(p) = \{\atoml(p),\atomr(p)\}$.
We lift $\atoml$, $\atomr$, and $\atom$ to sets of positions $P$ in the standard manner. We define also the set of positions of a set of atoms $A$ as $\pos(A) = \{p\mid \atom(p)\cap A \neq \emptyset\}$. We will write simply $\pos(\alpha)$ to denote $\pos(\{\alpha\})$.

Let us define the relation ${\sim} \subseteq \atoms \times \atoms$ as the smallest equivalence which relates any two atoms that share a position, i.e., they both belong to $\atom(p)$ for some position $p$.
An \emph{atom class} is then an equivalence class of $\sim$. We will denote by
$[\alpha]_\sim$ the atom class containing  the atom $\alpha$.

We say that atoms in a set $A\subseteq \atoms$ \emph{are in agreement}, or that $A$ \emph{agrees},
if there is an assignment $\ass$ and a symbol $a\in \Sigma$ such that for every atom $(x,i)\in
A$, $\ass(x)[i] = a$.  We may also say that $A$ \emph{agrees on the letter $a$ in the assignment $\ass$}.

It is not hard to see that
an  assignment is a solution if and only if every atom class agrees in it.
The core of our argument is a proof that any single atom class $[\alpha]_\sim$ agrees (note that this does not yet imply that all the classes agree in the same assignment, that will need an additional argument).
\begin{lemma}\label{lemma:oneagrees}
If $S$ is stable, then any single atom class agrees in it.
\end{lemma}

To prove \cref{lemma:oneagrees},  will for a given atom class $[\alpha]_\sim$ construct a sequence $A_0,\ldots,A_k = [\alpha]_\sim$ of sets of atoms that agree and eventually converge to $[\alpha]_\sim$. The proof we will need two components.
First, to ensure that the sequence terminates, we will argue that each $A_{i+1}$ is smaller then $A_i$ in certain well ordering $\emptier$ on sets of atoms. The ordering is defined as follows.
For each $x\in \X$, we define $seq_A(x)$ as the sequence of its indices $\{i\mid(x,i)\in A\}$ ordered from the smallest on the left.
Let $\lord$ be the usual lexicographic ordering on such sequences (e.g., $ 12379\lord 12468$).
Given another set of atoms $B$, we define $A\succeq_x B$ as $seq_A(x)\lord seq_B(x)$,
and we let $A \emptier B$ if and only if a) for all $x\in \X$, $A\succeq_x B$, and b) for at least one $x\in \X$, $B\not \succeq_x A$ (i.e., $A \succeq_x B$ holds for all variables and is sharp for at least one).
It is not difficult to show that $\emptier$ is a well ordering (It is transitive and antisymmetric. It has a finite domain, hence it has no infinite decreasing sequences.).

The second component needed in the proof of \cref{lemma:oneagrees} is a way of constructing $A_{i+1}$ from $A_i$.
It is the core of the entire proof and it is summarised in the following lemma.

\begin{lemma}\label{lemma:core}
For $\alpha\in\atoms$ and $A\subset[\alpha]_\sim$ that agrees, there is \mbox{$A'\subseteq[\alpha]_\sim$ with $A\emptier A'$}.
\end{lemma}

\begin{proof}
We call a position $p$ a \emph{hole} of $A$ if $|\atom(p)\cap A| = 1$, that is, exactly one of the two atoms at $p$ is present in $A$.
The \emph{missing atom} of the hole $p$ of $A$ is the atom $\alpha\in\atom(p)\setminus A$ (note that a hole has exactly one missing atom).
Since $A$ is strictly smaller then $[\alpha]_\sim$, it must have a hole (roughly, if there is no hole, then $\sim$ restricted to pairs incident with $A$ does not contain any pair incident with atoms outside $A$, which means that $A$ is an equivalence class of $\sim$, and this contradicts that $A$ is strictly smaller than the equivalence class $[\alpha]_\sim$).
Let $\smhole$ be the smallest hole in $A$ and let $(x,k)$ be its missing atom.
We will present the proof for the case that the missing atom is on the right, i.e., $(x,k) = \atomr(\smhole)$. The case $(x,k) = \atoml(\smhole)$ can be proved analogously.

The new atom class will be constructed in the form
$$A' = A \cup \{(x,k)\} \setminus \{(x,l)\mid k < l \leq \ell_x\}.$$
It is easy to prove that $A \emptier A'$ (for all variables $y\in \X$ except $x$, $A \succeq_y A'$ holds since $A'$ has the same $y$-atoms as $A$, and $A \succ_x A'$ follows from how $A'$ is constructed from $A$ and the def. of $\succeq_x$).
In the rest of the proof, we will concentrate on the most difficult part, which is showing that $A'$ agrees.

Assume that $A$ agrees on $a$ in an assignment $\nu$.
Let $w = \ass(y_1)\cdots\ass(y_n)$ be the word obtained using the $\ass$-values on the left side of the equation.
Then, from $\constrlass(y_1)\cdots \constrlass(y_n) = \constrlass(z_1)\cdots \constrlass(z_m)$, there must be words $v_1, \ldots, v_m$  from the right languages
	$\constrlass(z_1),\ldots,\constrlass(z_m)$, respectively, such that $w = v_1 \cdots v_m$
(note that this sequence of words does not correspond to an assignment since different occurrences of a variable may have different values).
Let $\smhole$ appear withing the $j$th word $v_j$ in the sequence $v_1 \cdots v_m$, that is,
$\sum_{i=1}^{j-1} \ell_{z_i} < \smhole \leq \sum_{i=1}^j \ell_{z_i}$.
Note that the position $\smhole - k + 1$ corresponds to the first character of $j$th word in the concatenation,
and that positions within this occurrence of $v_j$ can be translated to letters of $v_j$ by
\begin{itemize}
\item[($*$)]
$w[\smhole - k + i] = v_j[i]$ for all $i:1\leq i \leq \ell_x$
\end{itemize}
and that they also have the corresponding atoms of $x$ on the right, i.e.
\begin{itemize}
\item[($**$)]
$\atomr(\smhole - k + i) = (x,i)$ for all $i:1\leq i \leq \ell_x$.
\end{itemize}

We will show that when $\ass'$ is constructed from $\ass$ as $\ass' = \ass \setminus \{x\mapsto \ass(x)\} \cup \{x\mapsto v_j\}$, then $A'$ agrees on $a$ in $\ass'$.
For that, we will first argue that $v_j[i] = a$ for all $i$ s.t. $(x,i)\in A$ and $1\leq i < k$, and then we will show that $v_j[k] = a$ as well.

To show that $v_j[i] = a$ for all $i$ s.t. $(x,i)\in A, 1\leq i < k$, we show that
$w[p] = a$ for all positions in $P = \{\smhole - k + i\mid 1\leq i < k, (x,i)\in A\}$ (the implication follows from $(*)$).
Since $A$ agrees on $a$ in $\ass$,
the words $\ass(y_1),\ldots,\ass(y_n)$ have $a$ on all indices corresponding to the atoms in $A$, that is, $\ass(y_j)[i] = a$ for all $1\leq j \leq n$ and $i$ s.t. $(y_j,i)\in A$.
Therefore, for all positions $p$, we have $\atoml(p)\in A \implies w[p] = a$.
The definition of $\smhole$ says that for all positions $p$ smaller than $\smhole$,
$\atoml(p) \in A \Longleftrightarrow \atomr(p) \in A$.
This equivalence applies to all positions of $P$ because they are smaller then $\smhole$ by definition.
In summary, we have that for all $p \in P$,
we have $p =  \smhole-k+i$ for some $i:1\leq i < k \land (x,i)\in A$, that
$(x,i) = \atomr(p)$ by $(**)$, and that $\atomr(p)\in A \implies \atoml(p) \in A \implies w[p] = a$, which gives $w[p] = a$ for all $p\in P$.

We show that $v_j[k] = a$ by a similar argument.
We know that $\smhole$ has an atom from $A$ on the left (by definition of a hole and our initial assumption that the right atom is missing),
hence $w[\smhole] = a$, and because $w[\smhole] = v_j[k]$ by ($*$), we have $v_j[k] = a$.

We have shown that $\ass' = \ass \setminus \{x\mapsto \ass(x)\} \cup \{x\mapsto v_j\}$ can be used as an assignment in which $A'$ agrees on $a$.

\qed
\end{proof}

With \cref{lemma:core} at hand, we can prove \cref{lemma:oneagrees}.

\begin{proof}[\cref{lemma:oneagrees}]
Given arbitrary $\alpha\in\atoms$,
we will inductively construct the sequence $A_0,\ldots,A_k = [\alpha]_\sim$ of sets of atoms, that will for each $0\leq i < 0$ satisfy two conditions:
\begin{enumerate}
\item
$A_i\subseteq [\alpha]_\sim$,
\item
$A_{i}\emptier A_{i+1}$.
\end{enumerate}

For $i=0$, $A_0$ can be chosen as any singleton $\{\beta\} \subseteq [\alpha]_\sim$.
It satisfies both conditions:
$A_0 \subseteq [\alpha]_\sim$ is trivial, and since $A_0$ is a
singleton, it trivially agrees in any assignment.
The induction step then follows in case $A_i$ is not yet equal to $[\alpha]_\sim$.
In that case, we can use \cref{lemma:core} to derive that there is the needed $A_{i+1}$.
The construction terminates since $\emptier$ is a well ordering. The only way how it can terminate is by $A_k = [\alpha]_\sim$ (this is the only possible situation in which \cref{lemma:core} does induce existence of the next set in the sequence).
\qed
\end{proof}

With \cref{lemma:oneagrees} at hand,  we can now finish the proof of \cref{lemma:fixed}.

\begin{proof}[\cref{lemma:fixed}]
Consider all equivalence classes $\E = \{[\alpha_1]_\sim, \ldots,[\alpha_k]_\sim\}$ of
$\sim$. We will construct a sequence of language assignments $\constrlass = \lass^0,\ldots,\lass^k$,
and a sequence of letters $a_1,\ldots,a_k$ such that for each $i:0\leq i \leq k$,
the following holds:
\begin{enumerate}
\item
$\lass^i(x)$ is a nonempty subset of $\constrlass(x)$ for all $x\in\X$,
\item
$(e,\lass^i)$ is a system, that is, $\lass^i(y_1)\cdots \lass^i(y_n) = \lass^i(z_1)\cdots \lass^i(z_m)$,
\item
in the system $(e,\lass^i)$, all classes $[\alpha_j]_\sim, 1\leq j \leq i$ agree on $a_1,\ldots,a_i$, respectively, in every assignment  (notice that this condition is vacuous for $i=0$).
\end{enumerate}

For each $i:0\leq i < k$, we construct $\lass^{i+1}$ from $\lass^i$ so that
for each $x\in \X$, $\lass^{i+1}(x)$ contains the words of $\lass^i(x)$ due to which $[\alpha_i]_\sim$ could agree on $a_i$---those with $a_i$ on positions marked by the atoms in $[\alpha_i]_\sim$, formally,
$$\lass^{i+1}(x) = \{w\in \lass^i(x) \mid \forall j: (x,j)\in [\alpha_{i+1}]_\sim\rightarrow w[j]=a_{i+1}\}.$$

\medskip
Let us now argue by induction on $i$ that the points 1 to 3 above hold.
For $i=0$, $\lass^i = \constrlass$. Points 1 and 2 are obviously satisfied, 3 as well since it is is vacuous for $i=0$.

	Assume that 1-3 hold for $i$ and let us prove that they hold for $i+1$.
We first prove that 1 holds.
        First, $\lass^{i+1}(x)\subseteq \constrlass(x)$ follows from that 1 holds for $i$, that is, $\lass^{i}(x)\subseteq \constrlass(x)$, and from that $\lass^{i+1}(x)\subseteq \lass^{i}(x)$ by construction.
        Second, $\lass^{i+1}(x) \neq \emptyset$ follows from that
	$\lass^i(x) \neq \emptyset$ by IH and that $[\alpha_i]_\sim$ agrees on $a_i$ in $(e,\lass^i)$ (the definition of agreement implies that $\lass^{i}(x)$ must have at least one word with $w[j]=a_i$ for all $(x,j)\in[\alpha_i]_\sim$, and that word will appear in $\lass^{i+1}(x)$).

	Next, we prove point 2.
	We will show the proof only for $\lass^{i+1}(y_1)\cdots \lass^{i+1}(y_n) \subseteq \lass^{i+1}(z_1)\cdots
	\lass^{i+1}(z_m)$. The opposite inclusion may be proved analogously. Consider
	some $w \in \lass^{i+1}(y_1)\cdots \lass^{i+1}(y_n)$. From IH we have $w\in
	\lass^{i}(y_1)\cdots \lass^{i}(y_n)=  \lass^{i}(z_1)\cdots\lass^{i}(z_m)$.
	We can write $w$ as $w = w_{y_1}\cdots w_{y_n} = w_{z_1}\cdots w_{z_m}$ where $w_x \in
	\lass^i(x)$ for each $x\in\X$.
        Assume for the sake of contradiction that $w\not \in\lass^{i+1}(z_1)\cdots\lass^{i+1}(z_m)$.
        That means that for some $j$, $w_{z_j} \not\in \lass^{i+1}(z_j)$.
        By definition of $\lass^{i+1}(z_j)$, there must be an atom $(z_j,l)\in[\alpha_{i+1}]_\sim$ such that $w_{z_j}[l]\neq a_{i+1}$.
        The other, left, atom (on the position $p=\sum_{o=1}^{o\leq j}\ell_o + l$) also belongs in the equivalence class $[\alpha_{i+1}]_\sim$ by definition of $\sim$. Let it be some $(y_p,q)$. We have that $w_{y_p}[q]$ is not $a_{i+1}$, but this contradicts the definition of $\lass^{i+1}(y_p)$.

Last, we prove point 3. In the system $(e,\lass^i)$, $[\alpha_{i+1}]_\sim$ agrees on some $a_{i+1}$ by \cref{lemma:oneagrees}.
By the IH, all $[\alpha_j]_\sim, 1\leq j < i$ agree on $a_1,\ldots,a_i$, respectively, in all assignments. $\lass^{i+1}$ is then constructed by removing words from $\lass^i$ that do not have the $a_{i+1}$ on the positions of the atoms in $[\alpha_{i+1}]_\sim$.
Therefore, after the removal, only assignments in which $[\alpha_{i+1}]_\sim$ agrees on $a_{i+1}$ are are left.
We have thus proved that the properties 1-3 are satisfied for each $i:0\leq i \leq k$.

The solution of the original system $(e,\constrlass)$ is now obtained as follows.
In $\lass^k$, every language $\lass^k(x),x\in \X$ equals the singleton $\{w_x\}$ such that for each $i:1\leq i\leq \ell_x$,
$w_x[j] = a_i$ where $(x,i)\in [\alpha_i]_\sim$ (since 3 holds for $\lang^k$). We also know that $w_x\in \constrlass(x)$ (since 2 holds for all $0\leq i \leq k$). The map $\ass = \{x\mapsto w_x\mid x\in \X\}$ is then an assignment in $(e,\constrlass)$, and it is also a solution because all atom classes agree in it (since 3 holds).
\qed
\end{proof}

\vspace{-0.0mm}
\section{Proof of \cref{thm:graphwstable}}\label{sec:proofThmGraphWStable}
\vspace{-0.0mm}

We start by defining some notation.
Given an inclusion graph~$\graph$, a~\emph{strongly connected component} (SCC)
of~$\graph$ is a~maximal subgraph of~$\graph$ where each two vertices are
reachable.
An SCC is \emph{trivial} if it has exactly one vertex (which does not contain
a~self-loop) and \emph{non-trivial} otherwise.
An SCC~$C$ is \emph{terminal} if the set of vertices reachable from~$C$
equals~$C$.
Given an SCC~$C$ of~$\graph$, then $\varsof C$ is the set of all variables that
occur in the vertices of~$C$.

We first start with proving the following lemma that gives us some properties of
inclusion graphs.

\begin{lemma}\label{lem:graph_structure}
Let~$\graph = (\vertices, \edges)$ be an inclusion graph of~$\constr$.
Then the following holds:
\begin{enumerate}
  \item  Every non-trivial SCC of~$\graph$ is terminal.\label{lem:graph_structure:terminal}
  \item  Let~$C_1$ and~$C_2$ be two different non-trivial SCCs
    of~$\graph$, then $\varsof{C_1} \cap \varsof{C_2} =
    \emptyset$.\label{lem:graph_structure:vars}
  \item  If~$x$ is a~variable on the right-hand side of some vertex $\graphnode
    \sterm \tterm$ that is a~trivial SCC, then~$x$ does not occur on the
    right-hand side of any other vertex of~$\graph$.\label{lem:graph_structure:rhs}
\end{enumerate}
\end{lemma}

\begin{proof}
  \begin{enumerate}
    \item  For the sake of contradiction, assume that there is an SCC~$C$
      of~$\graph$, a~vertex $\graphnode{\sterm}{\tterm} \in \vertices$ that
      is not in~$C$, and a~vertex~$\graphnode{\sterm_C}{\tterm_C} \in C$ such
      that $(\graphnode{\sterm_C}{\tterm_C}, \graphnode{\sterm}{\tterm}) \in
      \edges$.
      From \igcondref{cond:edges} of inclusion graphs, it follows that
      $\sterm_C$ and $\tterm$ share a~variable, for example~$x$.
      Also, from Conditions~\condref{cond:edges} and~\condref{cond:cycle} and
      the fact that $\graphnode{\sterm_C}{\tterm_C}$ is in an SCC, it follows
      that~$C$ also contains the vertex $\graphnode{\tterm_C}{\sterm_C}$ and the
      transition $(\graphnode{\tterm_C}{\sterm_C}, \graphnode{\sterm_C}{\tterm_C})$.
      But then, it follows that the variable~$x$ occurs on the right-hand side
      of at least two vertices, so by \igcondref{cond:multioccur}, there also
      needs to be the vertex $\graphnode{\tterm}{\sterm} \in \vertices$, and 
      by \igcondref{cond:edges}, there also needs to be the edge
      $(\graphnode{\sterm}{\tterm}, \graphnode{\tterm}{\sterm})$.
      Finally, applying \igcondref{cond:edges} one more time, we obtain that
      $(\graphnode{\tterm}{\sterm}, \graphnode{\tterm_C}{\sterm_C}) \in \edges$
      (since $\sterm_C$ and $\tterm$ share~$x$).
      But then~$C$ is reachable from~$\graphnode{\sterm}{\tterm}$, which is
      a~contradiction.

    \item  For the sake of contradiction, assume two different SCCs~$C_1$
      and~$C_2$ of~$\graph$ and a~variable $x \in \varsof{C_1} \cap
      \varsof{C_2}$.
      Then there will be vertices
      $\graphnode{\sterm_1}{\tterm_1},\graphnode{\tterm_1}{\sterm_1} \in C_1$ and
      $\graphnode{\sterm_2}{\tterm_2},\graphnode{\tterm_2}{\sterm_2} \in C_2$
      with occurrences of $x$ in~$\sterm_1$ and~$\sterm_2$.
      From \igcondref{cond:edges}, we have that there will then also be edges
      $(\graphnode{\sterm_1}{\tterm_1}, \graphnode{\tterm_2}{\sterm_2}) \in
      \edges$ and
      $(\graphnode{\sterm_2}{\tterm_2}, \graphnode{\tterm_1}{\sterm_1}) \in
      \edges$, which is a~contradiction with the assumption that~$C_1$ and~$C_2$
      are different.

    \item  For the sake of contradiction, assume that~$x$ occurs on the
      right-hand side of $\graphnode \sterm \tterm \in \vertices$ and also on the right-hand
      side of $\graphnode{\sterm'}{\tterm'} \in \vertices$.
      Then, from \igcondref{cond:multioccur}, there will also be
      a~vertex~$\graphnode \tterm \sterm \in \vertices$ and by
      \igcondref{cond:edges}, there will be edges $(\graphnode \sterm \tterm,
      \graphnode \tterm \sterm) \in \edges$ and $(\graphnode \tterm \sterm,
      \graphnode \sterm \tterm) \in \edges$, therefore $\graphnode \sterm
      \tterm$ is not a~trivial SCC, which is a~contradiction.
      \qed
  \end{enumerate}
\end{proof}

Let us proceed to the proof of \cref{thm:graphwstable}.

\thmGraphWStable*

\begin{proof}
  ($\Leftarrow$): Trivial by construction of singleton language assignments from
  the solution.

  ($\Rightarrow$): We will show how to construct a~solution of~$\constr$ given
  a~live language assignment~$\lass$ that refines~$\constrlass$ and is stable
  for~$\graph$.
  Intuitively, the construction proceeds by computing \emph{partial} string
  assignments bottom-up on the structure of~$\graph$, starting from non-trivial
  SCCs (using \cref{thm:main}) and then proceeding upward on the (SCC-free)
  structure of~$\graph$, in each step taking the concatenation of the terminal
  vertices, removing those vertices, and using \cref{thm:weaklystable} to infer
  assignments to other variables.

  Formally, let us consider a~non-trivial SCC~$C$ of~$\graph$ (we note that due
  to \cref{lem:graph_structure}(\ref{lem:graph_structure:terminal}), $C$~is
  terminal).
  From \igcondref{cond:cycle}, it follows that for every vertex $\graphnode
  \sterm \tterm$ of~$C$, there is also the vertex $\graphnode \tterm \sterm$
  in~$C$.
  Therefore, since~$\lass$ is live and stable for~$\graph$ (and so also
  for~$C$), from \cref{thm:main} it follows that there exists a~string
  assignment~$\ass_C$ that assigns strings to the variables in $\varsof C$ that
  is a~solution of~$\constr_{|C}$ (where $\constr_{|C}$ is a~formula obtained
  from~$\constr$ by removing all subformulae with variables not occurring in the
  vertices of~$C$).
  We can now build the string assignment $\ass_0$ by uniting the assignments
  $\ass_C$ for all non-trivial SCCs~$C$ of~$\graph$ (this is possible since from
  \cref{lem:graph_structure}(\ref{lem:graph_structure:vars}), the sets of
  variables of two different non-trivial SCCs in~$\graph$ are disjoint).
  Let us also create the graph~$\graph_0$ by removing all non-trivial SCCs
  from~$\graph$ (they are not needed any more since we have their solutions
  in~$\ass_0$).
  We note that~$\graph_0$ is acyclic.

  We now proceed by constructing a~sequence of partial solutions $\ass_0
  \subseteq \ass_1 \subseteq \ldots \subseteq \ass_n$, where $\ass_n$ is
  a~solution of~$\constr$, as follows.  We start by assigning $i \gets 0$ and
  keep repeating the following steps until~$\graph_i$ is empty.
  \begin{enumerate}
    \item  First, we create the language assignment~$\lass_i$ such that
      \begin{equation}
        \lass_i(x) = \begin{cases}
          \{w\} & \text{if } (x \mapsto w) \in \ass_i, \\
          \lass(x) & \text{otherwise.}
        \end{cases}
      \end{equation}

    \item
      Now, we pick from~$\graph_i$ the vertices to process in this iteration.
      Let $T_i = \{\graphnode{\sterm_1}{\tterm_1}, \ldots,
      \graphnode{\sterm_k}{\tterm_k}\}$ be the largest set of terminal vertices
      of~$\graph_i$ such that for all $1 \leq j \leq k$, all variables
      in~$\sterm_j$ either
      \begin{inparaenum}[(i)]
        \item  are assigned a~string in~$\ass_i$ or
        \item  do not occur on the
          right-hand side of any vertex in~$\graph_i$.
      \end{inparaenum}
      We claim that~$T_i$ is unique and non-empty (we give the claim formally
      below as \cref{claim:T_i-non-empty}).

    \item  Using~$T_i$, we construct the single-equation system $\varphi_i\colon
      \sterm_1 \sharp \ldots \sharp \sterm_k = \tterm_1 \sharp \ldots \sharp
      \tterm_k \land \bigwedge_{x \in \vars} x \in \constrlass(x)$.
      We claim that~$\lass_i$ is weakly stable for~$\varphi_i$ (we formally
      give the claim and its proof below as \cref{claim:weakly_stable_i}).
      Using the claim, we can use \cref{thm:weaklystable} (especially
      the algorithm in its proof) to obtain a~solution~$\ass'$ to~$\varphi_i$.
      We then set $\ass_{i+1} \gets \ass_i \cup \ass'$.

    \item  Next, we set $\graph_{i+1}$ to be the same as~$\graph_i$ without the
      vertices (and the incident edges) from~$T_i$.
      \label{step:remove_nodes}

    \item  Finally, $i \gets i + 1$.
  \end{enumerate}

  \noindent
  Let us use~$n$ for the last value of~$i$.
  Before we continue, we give proofs of the two claim used above.
  First, we show that each~$T_i$ is unique and non-empty (which is also the
  reason why the algorithm above always terminates).

  \begin{claim}\label{claim:T_i-non-empty}
    For all $0 \leq i \leq n$ it holds that~$T_i$ is unique and non-empty.
  \end{claim}
  \begin{claimproofnoqed}
    \begin{itemize}
      \item  (\emph{non-emptiness})
    First, we show that~$T_i$ is non-empty by contradiction: assume that~$\graph_i$
    is non-empty and that for every $G_i$-node $\graphnode \sterm \tterm$,
    there is an occurrence of a~variable~$x$ in~$\sterm$ such that
    \begin{inparaenum}[(i)]
      \item  $x$~is not assigned a~string by~$\ass_i$ and
      \item  $x$~occurs on the right-hand side of some vertex
        $\graphnode{\sterm'}{\tterm'}$
    \end{inparaenum}
    (and so~$T_i$ is empty).
    Then, according to \igcondref{cond:edges} of inclusion graphs, it follows
    that there is an edge from $\graphnode \sterm \tterm$ to
    $\graphnode{\sterm'}{\tterm'}$.
    Because the previous condition holds for all nodes in~$\graph_i$, it follows
    that there is a~cycle in~$\graph_i$, which is a~contradiction with the fact
    that~$\graph_i$ is acyclic.

      \item (\emph{uniqueness})
        For the sake of contradiction, assume that there are two different
        non-empty sets~$T_i^1$ and~$T_i^2$ that are largest and satisfy the
        requirements.
        But then their union $T_i^1 \cup T_i^2$ would also satisfy the
        requirements, which is a~contradiction with~$T_i^1$ and~$T_i^2$ being
        largest.
        \claimqed
    \end{itemize}
  \end{claimproofnoqed}

  \noindent
  Next, we prove the claim that allowed application of~\cref{thm:weaklystable}
  above.

  \begin{claim}\label{claim:weakly_stable_i}
  For all $0 \leq i \leq n$ it holds that
  $\lass_i$ is weakly stable for~$\varphi_i$.
  \end{claim}
  \begin{claimproofnoqed}
    In order to show that~$\lass_i$ is weakly stable for~$\varphi_i$, we
    need to show the following two properties:
    \begin{enumerate}[(i)]
      \item  $\lass_i(\sterm_1 \sharp \ldots \sharp \sterm_k) \subseteq
        \lass_i(\tterm_1 \sharp \ldots \sharp \tterm_k)$:
        We prove the property by showing that for each vertex
        $\graphnode{\sterm}{\tterm}$ in $\graph$, we have
        $\lass_i(\sterm) \subseteq \lass_i(\tterm)$.
        We proceed by induction on~$i$.
        \begin{enumerate}
          \item  $i=0$:
            Let~$Y$ be the set of variables occurring in non-trivial SCCs
            of~$\graph$.
            Then, from the construction, $\lass_0$~assigns each variable~$y
            \in Y$ a~value obtained from the solution of non-trivial
            SCCs~$\ass_0$.
            As a~consequence, for every vertex $\graphnode \sterm \tterm$ of
            a~non-trivial SCC in~$\graph$, it holds that $\lass_0(\sterm) =
            \lass_0(\tterm)$, which is even stronger than $\lass_0(\sterm)
            \subseteq \lass_0(\tterm)$.

            Furthermore, for every vertex $\graphnode{\sterm'}{\tterm'}$
            that is in the acyclic part of~$\graph$ and that contains some 
            variables from~$Y$ on its left-hand side\footnote{%
            Note that due to
            \cref{lem:graph_structure}(\ref{lem:graph_structure:rhs}), it
            cannot happen that any variable from~$Y$ would be on the
            right-hand side of a~node in the acyclic part of~$\graph$.},
            we know that for any variable~$y \in Y$, it holds that
            $\lass_0(y) \subseteq \lass(y)$ ($\lass_0(y)$ is, in fact,
            a~singleton) and, therefore, since it originally
            held (from the assumption) that $\lass(\sterm') \subseteq
            \lass(\tterm')$, it will also hold that $\lass_0(\sterm')
            \subseteq \lass(\tterm')$ (note that due to
            \cref{lem:graph_structure}(\ref{lem:graph_structure:rhs}), it
            follows that $\lass_0(\tterm') = \lass(\tterm')$ since $\tterm'$
            contains no variable from~$Y$, so we can conclude that
            $\lass_0(\sterm') = \lass_0(\tterm')$).

          \item  $i > 0$:
            Our induction hypothesis is that for every vertex $\graphnode
            \sterm \tterm$ in~$\graph$, it holds that $\lass_{i-1}(\sterm)
            \subseteq \lass_{i-1}(\tterm)$.
            Let us now look at the inclusions represented by vertices
            from~$T_i$.
            By construction, $\lass_{i}(x) \subsetneq \lass_{i-1}(x)$ can only happen
            for variables~$x$ that do not appear on the right-hand side of any
            vertex in~$\graph_i$ (in the construction of~$T_i$, we could not have
            picked a~vertex containing~$x$).
            Therefore, we can conclude that $\lass_i(\sterm) \subseteq
            \lass_i(\tterm)$.

        \end{enumerate}

      \item  \emph{Each variable $x \in \vars$ has no more than one
        occurrence in~$\tterm_1 \sharp \ldots \sharp \tterm_k$}.
        We prove this by contradiction: assume that variable~$x$ has more
        than one occurrence on the right-hand sides of~$T_i$, and let one of
        the occurrences in the vertex
        $\graphnode{\sterm_j}{\tterm_j}$.
        Then, by \igcondref{cond:multioccur}, $\graph$~also contains the
        vertex~$\graphnode{\tterm_j}{\sterm_j}$ and, by
        \igcondref{cond:edges}, the two vertices are in a~non-trivial SCC,
        which is a~contradiction with the fact that~$\graph_i$ is acyclic
        (which follows from the fact that~$\graph_0$ is acyclic and that we
        only removed vertices/edges when we were creating~$\graph_i$ in
        Step~\ref{step:remove_nodes}) of the previous iteration of the loop.
        \claimqed
    \end{enumerate}
  \end{claimproofnoqed}

  \noindent
  Finally, we need to show that~$\ass_n$ is a~solution of~$\constr$.
  We leverage the following property from the proof of
  \cref{claim:T_i-non-empty}(i): for every $0 \leq i \leq n$ and every vertex
  $\graphnode \sterm \tterm$ in~$\graph$, it holds that $\lass_i(\sterm)
  \subseteq \lass_i(\tterm)$.
  Together with the fact that $\lass_n$ assigns each variable a~singleton
  language, we infer that $\lass_n(\sterm) = \lass_n(\tterm)$.
  It follows that~$\ass_n$ is a~solution of~$\constr$.
  \qed
\end{proof}

\vspace{-0.0mm}
\section{Proof of \cref{thm:inclGraphCorr}}\label{sec:proof-incl-graph-corr}
\vspace{-0.0mm}

\thmInclCorr*

\begin{proof}
  Let $\graph_\constreq = (V,E)$ be a graph obtained by $\incl(\constreq)$. First, we prove the following auxiliary claim.
  


  \begin{claim}\label{claim:unreach}
    Let $\St'$ be a set of vertices on the end of the algorithm. Then, every $v \in \St'$ is not reachable 
    from some $u \in V\setminus \St'$ in $G_\constreq$.
  \end{claim}
  \begin{claimproof}
    Let $v \in \St'$ and $u \in V\setminus \St'$ be an arbitrary vertex s.t. $v$ is reachable from $u$ in $G_\constreq$. 
    Let $\pi$ be the path from $u$ to $v$ in $G_\constreq$. Then, there is the \emph{maximal} suffix $\pi'$ of $\pi$ s.t.
    $\pi'$ does not contain repetitions of vertices or dual vertices. Since $\pi'$ is maximal, there is some vertex $u'$ of 
    $\pi'$ belonging to $V$ (otherwise, we could extend $\pi'$). Since $\pi'$ has no duals and repetitions, such a path is 
    also in $\splitgraph$. Since $v$ is reachable from $u'\in V$, according to \cref{alg:inclgraph} we have that $v\in V$,
    which is a contradiction.
  \end{claimproof}
  
  We now prove that $\graph_\constreq$ meets conditions \condref{cond:atleastone}--\condref{cond:cycle}.
  \begin{description}
    \item[\condref{cond:atleastone}] Follows trivially from \cref{alg:inclgraph}.
    
    \item[\condref{cond:multioccur}] 
      Consider some $u = \graphnode{s}{w_1xw_2} \in \graph_\constreq$ s.t. $\dualof u \notin \graph_\constreq$. Therefore, $u \in \St'$.
      Now assume that there is some other vertex $v = \graphnode{s'}{w_1'xw_2'} \in \graph_\constreq$. If $v' = \graphnode{w_1'xw_2'}{s'} \in \graph_\constreq$ 
      meaning that $v', v \in V$, then there is an edge between $v'$ and $u$, which is a contradiction according to \cref{claim:unreach}.
      If $v' = \graphnode{w_1'xw_2'}{s'} \notin \graph_\constreq$ meaning that $v \in \St'$. Further, there is the edge between 
      $\graphnode{w_1xw_2}{s} = \dualof u$ and $v$. Since $v \in \St'$, we have that $v$ had to be a source at the time of his adding to $\St'$ and 
      hence $u$ was added to $\St'$ before $v$. But similarly, there is the edge $\dualof v$ and $u$ meaning, according to the same rationale, that 
      $v$ was added to $\St'$ before $u$, which is a contradition.


    \item[\condref{cond:edges}] Given directly from line 6 of \cref{alg:inclgraph}.

    \item[\condref{cond:cycle}] Let $\pi = v_1, v_2, \dots, V_k$ where $k > 1$ be a cycle in 
      $G_\constreq$ (there is the edge between $v_k$ and $v_1$). For the sake of contradition, we assume that for 
      some $v_i = \graphnode{s}{t}$ the opposite vertex 
      $\graphnode{t}{s} \notin G_\constreq$. Hence $v_i \in \St'$. According to \cref{claim:unreach} we have 
      that $v_i \in \St'$ for each $1 \leq i \leq k$. Therefore, $v_i \neq v_j$ for each $i \neq j$ meaning that 
      the cycle $\pi$ is in $\splitgraph$. And hence all vertices from $\pi$ belongs to the same (nontrivial) SCC, which
      is a contradiction to $v_i$ belonging to $\St'$.

  \end{description}

  Now, we prove the second part of the theorem, i.e., if there is an acyclic inclusion graph for $\constreq$, then $\incl(\constreq)$ is acyclic. 
  From the definition of chain-freeness we have that $\constreq$ is chain-free iff $\splitgraph$ is acyclic. From \cref{thm:acyclicIGChainFree}
  we further have that there is an acyclic inclusion graph iff $\splitgraph$ is acyclic. Then, according to \cref{alg:inclgraph} we get that for each $v$
  it does not contain $\dualof v$. Since $\incl(\constreq)$ is an inclusion graph (proven above), it is acyclic (otherwise it must contain at least 
  one vertex together with its dual). 

  The graph $\incl(\constreq)$ has a minimum number of vertices because vertices in nontrivial SCCs in $\splitgraph$ (and between nontrivial SCCs) 
  forces the inclusion graph to have all also the dual ones.
  \qed
\end{proof}

\vspace{-0.0mm}
\section{Proof of \cref{thm:acyclicIGChainFree}}\label{sec:proof_acyclic_splitting}
\vspace{-0.0mm}

\thmAcyclicIGChainFree*

\begin{proof}
($\Rightarrow$)
Given the splitting graph~$\splgraph_\constr = (\positions_\constr,
  \edges_\constr, \varmap)$ of a~system $\constr\colon
\bigwedge_{i = 1}^{m} \sterm_i = \tterm_i \land \bigwedge_{x \in \vars} x \in
\constrlass(x)$, we show how to construct an acyclic inclusion graph~$\graph$
for~$\constr$.
First, we use~$\splgraph_\constr$ to construct the
\emph{\"{u}bergraph}~$\ubergraph_\constr$ of~$\constr$, which can be seen as a~graph
obtained by merging all nodes of~$\splgraph_\constr$ corresponding to all
positions on one side of an~equation into one (i.e., for a~system $\bigwedge_{i
= 1}^{m} \sterm_i = \tterm_i$, there will be one node for each~$\sterm_i$ and
one node for each~$\tterm_i$).
Then, we will use the~$\ubergraph$ to choose, for each equation $\sterm_i =
\tterm_i$, whether to add into~$\graph$ either
$\graphnode{\sterm_i}{\tterm_i}$ or~$\graphnode{\tterm_i}{\sterm_i}$ (adding
both is not allowed since that would imply that there is a~cycle
in~$\graph$).

Formally, the \"{u}bergraph of~$\constr$ is a~graph~$\ubergraph_\constr = (\{1,
\ldots, 2m\}, H)$ with nodes $\{1, \ldots, 2m\}$ and edges $H = \{ (j_1, j_2)
  \mid ((j_1, k_1), (j_2, k_2)) \in \edges_\constr \text{ for some } k_1, k_2
  \}$.

\begin{claim}\label{claim:ubergraph_acyclicity}
$\splgraph_\constr$ is chain-free iff~$\ubergraph_\constr$ is acyclic.
\end{claim}

\begin{claimproof}
($\Rightarrow$) 
By contradiction.
Assume that $\splgraph_\constr$ is chain-free and~$\ubergraph_\constr$ contains
a~cycle $j_0 \to j_1 \to \ldots \to j_\ell \to j_0$.
From the definition of edges in a~splitting graph, it follows that there is an
edge from $(j_i, k_i)$ to $(j_{i'}, k_{i'})$ in~$\splgraph_\constr$ for $i' = i
+ 1 \modulo \ell$ iff there is an edge from every $(j_i, k_\alpha)$ for $1 \leq
k_\alpha \leq n_{j_i}$ to $(j_{i'}, k_{i'})$.
We can therefore construct a~chain $(j_0, k_0), (j_1, k_1), \ldots, (j_\ell,
k_\ell), (j_0, k_0)$ in~$\splgraph_\constr$ where each $k_{i + 1 \modulo \ell}$ is
a~position with an incoming edge from $(j_i, k_i)$, which is a~contradiction.
\medskip

\noindent
($\Leftarrow$) 
Because $\ubergraph_\constr$ can be seen as an existential abstraction
of~$\splgraph_\constr$, a~chain in~$\splgraph_\constr$ implies a~cycle
in~$\ubergraph_\constr$.
\end{claimproof}

\begin{claim}\label{claim:ubergraph_sinks_sources}
Let $\sterm_i = \tterm_i$ be an equation of~$\constr$ and~$(2i-1)$ and~$(2i)$ be
the nodes of~$\ubergraph_\constr$ that correspond to~$\sterm_i$ and~$\tterm_i$
respectively.
Then the node~$(2i - 1)$ is a~source (resp.\ a~sink) iff the node~$(2i)$ is a~sink
(resp.\ a~source).
%
\end{claim}

\begin{claimproof}
  Let us show that if node $(2i-1)$ is a~source in $\ubergraph_\constr$, then 
  node~$(2i)$ is a~sink (the other part can be proved similarly).
  We proceed by contradiction: assume that $(2i-1)$ is a~source but~$(2i)$ is
  not a~sink, i.e., w.l.o.g., there is an~edge $(2i) \to (2h)$.
  From the definition of \"{u}bergraph, it needs to hold that there is an edge
  $((2i, k_\alpha), (2h, k_\beta))$ in $\splgraph_\constr$ for some $1 \leq
  k_\alpha \leq n_{2i}$ and $1 \leq k_\beta \leq n_{2h}$.
  Let $\sterm_i = \tterm_i$ and $\sterm_h = \tterm_h$ be the two equations
  whose sides correspond to nodes $(2i-1)$, $(2i)$, $(2h-1)$, and $(2h)$ respectively.
  W.l.o.g, let $\varmapof{(2i, k_\alpha)} = x$  and $\varmapof{(2h, k_\beta)} = y$.
  For the edge $((2i, k_\alpha), (2h, k_\beta))$ to be in~$\splgraph_\constr$,
  it needs to hold that $\sterm_i$ contains an occurrence of variable~$y$, say,
  at position $(2i - 1, k_\gamma)$.
  Since~$y$ has an occurrence in~$\tterm_h$ and also in~$\sterm_i$, there will
  also be an edge from $(2h-1)$ to $(2i-1)$, which is a~contradiction with the
  fact that~$(2i-1)$ is a~source.
\end{claimproof}

Our procedure for constructing the inclusion graph~$\graph$ for~$\constr$ now
works as follows:
\begin{enumerate}
  \item  First, let $\sinks$ be the set of \emph{sink} nodes
    of~$\ubergraph_\constr$ and $\sources$ be the set of its \emph{source} nodes.
    We also start with an empty inclusion graph~$\graph = (\emptyset,
    \emptyset)$.
  \item  Then we proceed in a~loop until $\ubergraph_\constr$ is empty:
      \begin{enumerate}
        \item  We pick a~pair of nodes $(2i-1)$ and $(2i)$ such that $(2i-1)$ is
          a~sink and $(2i)$ is a~source (if it is the other way around, we just
          switch the sides of the $i$-th equation).
          \cref{claim:ubergraph_sinks_sources} and acyclicity of
          $\ubergraph_\constr$ (\cref{claim:ubergraph_acyclicity}) guarantee
          that this is possible.
          We remove $(2i-1)$ and $(2i)$ from~$\ubergraph_\constr$.

        \item  We insert vertex $\graphnode{\sterm_i}{\tterm_i}$ into~$\graph$
          (the two sides of the inclusion correspond to nodes $(2i-1)$ and $(2i)$
          of the \"{u}bergraph).

        \item  For every node of the modified~$\ubergraph_\constr$ that is not
          in $\sinks \cup \sources$, we check if it is a~sink (resp.\
          a~source) and add it to $\sinks$ (resp.\ $\sources$).

      \end{enumerate}

  \item  Finally, we set the set of edges of~$\graph$ to satisfy
    \igcondref{cond:edges} (this is deterministic).
\end{enumerate}

\begin{claim}
$\graph$~is an acyclic inclusion graph.
\end{claim}

\begin{claimproofnoqed}
  Let us show that the conditions of acyclic inclusion graphs hold for~$\graph$.
  \begin{enumerate}
    \item[\condref{cond:atleastone}:]
      This obviously holds because the construction of~$\graph$
      from~$\ubergraph_\constr$ proceeds until $\ubergraph_\constr$ is empty, in
      each iteration removing one equation $\sterm_i = \tterm_i$
      from~$\ubergraph_\constr$ and adding into~$\graph$ either
      $\graphnode{\sterm_i}{\tterm_i}$ or $\graphnode{\tterm_i}{\sterm_i}$.

    \item[(acyclicity):]
      By contradiction: assume that $\splgraph_\constr$ is chain-free but
      $\graph$ contains a~cycle $\graphnode{\sterm_1}{\tterm_1} \to 
      \graphnode{\sterm_2}{\tterm_2} \to
      \graphnode{\sterm_3}{\tterm_3} \to
      \ldots \to
      \graphnode{\sterm_\ell}{\tterm_\ell} \to
      \graphnode{\sterm_1}{\tterm_1}$.
      From the construction of edges in~$\graph$, it needs to hold that~$\sterm_1$
      and~$\tterm_2$ share a~variable (say~$x_1$), $\sterm_2$~and~$\tterm_3$
      share a~variable (say~$x_2$), \ldots, and~$\sterm_\ell$ and~$\tterm_1$
      share a~variable (say~$x_\ell$).
      Looking at the \"{u}bergraph~$\ubergraph_\constr$, there will then be
      an edge from~$\tterm_1$ to~$\tterm_2$ (since~$\sterm_1$ and~$\tterm_2$
      both contain~$x_1$), $\tterm_2$~to~$\tterm_3$, \ldots,
      $\tterm_\ell$~to~$\tterm_1$, creating a~cycle in~$\ubergraph_\constr$.
      From \cref{claim:ubergraph_acyclicity} it follows that~$\splgraph_\constr$
      is not chain-free, which is a~contradiction.

    \item[\condref{cond:multioccur}:]
      We show the antecedent of the condition never holds, i.e., that there is
      no vertex $\graphnode{\sterm}{\tterm}$ in~$\graph$ such that~$\tterm$
      contains a~variable with multiple occurrences on right-hand sides of
      vertices of~$\graph$.
      We proceed by contradiction and assume that there are multiple occurrences
      of some variable~$x$ on the right-hand sides of vertices in~$\graph$.
      Let us analyze the following two possible cases:
      \begin{enumerate}
        \item  \emph{$\graph$~contains a~vertex $\graphnode \sterm \tterm$ with (at
          least) two occurrences of~$x$ in~$\tterm$.}
          Then, in~$\ubergraph_\constr$, there will be an edge from~$\sterm$
          to~$\tterm$.
          Therefore, in the algorithm, at some point, $\sterm$~will become
          a~source and~$\tterm$ will become a~sink.
          But then the node inserted into~$\graph$ will be~$\graphnode \tterm
          \sterm$, which is a~contradiction with the fact that~$\graph$
          contains~$\graphnode \sterm \tterm$ and (acyclicity).

        \item  \emph{$\graph$~contains two different vertices $\graphnode{\sterm_1}{\tterm_1}$
          and $\graphnode{\sterm_2}{\tterm_2}$ with~$x$ in both~$\tterm_1$
          and~$\tterm_2$.}
          The two vertices correspond to two equations $\sterm_1 = \tterm_1$ and
          $\sterm_2 = \tterm_2$.
          Since~$x$ is in~$\tterm_1$ and in~$\tterm_2$, there will be edges
          $\sterm_1 \to \tterm_2$ and $\sterm_2 \to \tterm_1$
          in~$\ubergraph_\constr$.
          W.l.o.g., assume that the nodes~$\sterm_1$ and~$\tterm_1$ were removed
          from~$\ubergraph_\constr$ before~$\sterm_2$ and~$\tterm_2$,
          adding~$\graphnode{\sterm_1}{\tterm_1}$ into~$\graph$.
          Since~$\sterm_1$ has an outgoing edge and~$\tterm_1$ has an incoming
          edge, at best, $\sterm_1$~can be a~source and~$\tterm_1$ can be
          a~sink, which would mean that we would
          add~$\graphnode{\tterm_1}{\sterm_1}$ to~$\graph$, which is a
          contradiction with the fact that~$\graph$
          contains~$\graphnode{\sterm_1}{\tterm_1}$ and (acyclicity).

      \end{enumerate}

    \item[\condref{cond:edges}:]
      Trivial from the construction.

    \item[\condref{cond:cycle}:]
      Follows from (acyclicity).
      \claimqed
  \end{enumerate}
\end{claimproofnoqed}

\medskip

\noindent
($\Leftarrow$)
By contradiction:
let~$\graph$ be an acyclic inclusion graph of~$\constr$ and assume that there is
a~cycle $\sterm_1 \to \sterm_2 \to \ldots \to \sterm_\ell \to \sterm_1$
in~$\ubergraph_\constr$ (by \cref{claim:ubergraph_acyclicity}, considering
a~cycle in the \"{u}bergraph is sufficient).
Further, assume that the equations whose sides participate in the cycle are 
$\sterm_1 = \tterm_1$, \ldots, $\sterm_\ell = \tterm_\ell$.
From the definition of the splitting graph, we know that~$\tterm_1$
and~$\sterm_2$ share a~variable (say~$x_1$),
$\tterm_2$~and~$\sterm_3$ share a~variable (say~$x_2$),
\ldots,
and~$\tterm_\ell$ and~$\sterm_1$ share a~variable (say~$x_\ell$).
From the definition of an inclusion graph, if it is acyclic, then from each
equation $\sterm = \tterm$, there can only be one inclusion in~$\graph$
(either~$\graphnode \sterm \tterm$ or~$\graphnode \tterm \sterm$).
There are then two general ways how the inclusion graph for~$\constr$ can look
like:
\begin{enumerate}
  \item  It contains the nodes
    $\graphnode{\tterm_1}{\sterm_1}$, 
    $\graphnode{\tterm_2}{\sterm_2}$, 
    \ldots,
    $\graphnode{\tterm_\ell}{\sterm_\ell}$.
    But then, from \igcondref{cond:edges}, since~$\tterm_1$ and~$\sterm_2$
    share~$x_1$, $\tterm_2$~and~$\sterm_3$ share~$x_2$, \ldots,
    $\tterm_\ell$~and~$\sterm_1$ share~$x_\ell$, there will be a~cycle
    $\graphnode{\tterm_1}{\sterm_1} \to \graphnode{\tterm_2}{\sterm_2} \to
    \ldots
    \to \graphnode{\tterm_\ell}{\sterm_\ell}
    \to \graphnode{\tterm_1}{\sterm_1}$ in~$\graph$, which is a~contradiction
    with the fact that~$\graph$ is acyclic.
    (similar reasoning holds for nodes
    $\graphnode{\sterm_\ell}{\tterm_\ell}$, 
    \ldots,
    $\graphnode{\sterm_2}{\tterm_2}$, 
    $\graphnode{\sterm_1}{\tterm_1}$).

  \item  Or it contains the nodes
    $\graphnode{\tterm_1}{\sterm_1}$, 
    $\graphnode{\tterm_2}{\sterm_2}$, 
    \ldots,
    $\graphnode{\tterm_\ell}{\sterm_\ell}$ with at least one inclusion
    $\graphnode{\sterm_i}{\tterm_i}$ swapped such that the next inclusion is in the
    ``normal'' direction: $\graphnode{\tterm_{i+1}}{\sterm_{i+1}}$.
    But then, since~$\tterm_i$ shares~$x_i$ with~$\sterm_{i+1}$, the
    variable~$x_i$ occurs multiple times on the right-hand sides of vertices
    of~$\graph$, which means (due to \igcondref{cond:multioccur}) that~$\graph$
    also contains the vertex~$\graphnode{\tterm_i}{\sterm_i}$ and from
    \igcondref{cond:edges}, it follows that there is a~cycle in~$\graph$, which
    is a~contradiction with the fact that~$\graph$ is acyclic.
    \qed

\end{enumerate}

\end{proof}

\vspace{-0.0mm}
\section{Proof of \cref{thm:soundness}}\label{sec:proof-soundness}
\vspace{-0.0mm}

\thmSoundness*

\begin{proof}
  We inductively prove the following invariant of the algorithm that need to hold in every iteration of the main 
  loop: for each $(\autass, W) \in \worklist$ and $v = \graphnode \sterm \tterm \in V\setminus W$ it holds that $\langof{\autass(\sterm)} \subseteq 
  \langof{\autass(\tterm)}$. 
  \begin{enumerate}
    \item Base case: At the beginning $W = V$ and therefore the invariant holds trivially.
    \item Inductive case: Assume that the invariant holds for all $(\autass, W)$ in $\worklist$. We prove that it 
    is still valid after an iteration of the main loop. Consider $(\autass, W)$ on line 3. If the condition on line 
    \ref{line:inclusiontest} holds then for $(\autass, W\setminus\{v\})$ the invariant clearly holds. We proceed with 
    the case the condition is not fulfilled. Now assume the case that $\sterm$ and $\tterm$ do not share a variable. Then,
    from the property of $\refine$, we have that for each $\autass'\in \tightset$: $\langof{\autass'(\sterm)} \subseteq 
    \langof{\autass'(\tterm)}$, therefore, $v$ need not be included in $W'$. From the condition \condref{cond:edges} of the 
    inclusion graph, we have that only successors of $v$ in $\graph_\constreq$ might be affected by the refinement of $\sterm$ 
    (other variables remain the same). Hence, for each $(\autass', W')$ where $\autass'\in \tightset$ the invariant holds. 
    For the case that $\sterm$ and $\tterm$ share a variable, the reasoning is the same as in the previous case except that 
    the inclusion $\langof{\autass'(\sterm)} \subseteq \langof{\autass'(\tterm)}$ might not be true in general. However, 
    again from \condref{cond:edges} we get that $v$ has a self-loop in $\graph_\constreq$ and hence it is included to 
    $W'$ as well. 
  \end{enumerate}

  Based on the invariant, if the algorithm returns $\sat$, we have that the $\graph_\constreq$ is stable wrt. $\autass$ (line 4). 
  From \cref{thm:graphwstable} we then obtain that $\constr$ is satisfiable. Further, we know 
  that $\refine$ preserves solutions. Therefore, for each solution $\nu$ of $\constr$, there is some $(\autass'', W'') \in \worklist$ 
  s.t. $\nu(x) \in \langof{\autass''(x)}$ for each variable $x$. Hence, if the algorithm returns $\unsat$, $\worklist = \emptyset$, which 
  means that $\constr$ is unsatisfiable. \qed
\end{proof}

\vspace{-0.0mm}
\section{Proof of \cref{thm:chain-free-compl}}\label{sec:proof-completeness}
\vspace{-0.0mm}

\thmChainFreeComplete*

\begin{proof}
  In the following, successors of a $(\autass, W)$ we mean all $(\autass', W')$ that were transitively obtained 
  from the refinement of $(\autass, W)$. Therefore, the computation of the algorithm can be seen as an infinite 
  tree whose vertices are labelled by items from $\worklist$.
  Since the inclusion graph $\graph_\constreq$ is acyclic, it is possible 
  do a topological sort of vetices as $\{ v_1, \dots, v_n \} = V$. From the FIFO used for storing the vertices, 
  we hence have if $W \cap \{ v_1, \dots, v_k \} = \emptyset$
  for some $(\autass, W)$ and $k\in\N$, then for all future successors $(\autass', W')$ of $(\autass, W)$ 
  we have $W' \cap \{ v_1, \dots, v_k \} = \emptyset$. Moreover, since we use FIFO structure in the algorithm and the inclusion graph is acyclic, 
  each branch leading from $(\constrautass,\vertices)$ contains either a vertex having no successors (returning 
  $\sat$ or $\tightset = \emptyset$) or a vertex $(\autass', W')$ where $v_1 \notin W'$. The same reasoning can be 
  repeated for such vertices for $v_2$ etc. Therefore, we have if $W \cap \{ v_1, \dots, v_k \} = \emptyset$ for 
  some vertex $(\autass, W)$, then all branches leading from this vertex contains either a vertex having no 
  successors or a vertex $(\autass', W')$ where $\{ v_1, \dots, v_k, v_{k+1} \} \cap W' = \emptyset$ where $k < |V|$.
  Using this argument, we immediately get that all branches ends with a vertex having no successors, or with a 
  vertex having $W = \emptyset$ and therefore, the algorithm always terminates. The rest is a direct consequence 
  of \cref{thm:soundness}.
  \qed
\end{proof}

\vspace{-0.0mm}
\section{Proof of \cref{thm:sattermination}}\label{sec:proof-termination}
\vspace{-0.0mm}

\thmSatTermination* 

\begin{proof}
  \emph{(Idea)} Assume that $\constr$ is satisfiable and therefore it has a fixed 
  solution $\nu$. Then, in each iteration of the algorithm, 
  we have that the length of a current shortest word of $\autass(x)$ is strictly bigger 
  than the length of a shortest word of $\autass'(x)$ (on line~13) or it holds 
  $\Min{\autass'(x)} \subset \Min{\autass(x)}$. Therefore, after a finite number of 
  steps we reach,in the worst case, an automata assignment having $\autass$ s.t. 
  $\Min{\autass(x)} = \{ \nu(x) \}$ for each variable $x$, which is stable on 
  line~6 and hence, the algorithm eventually returns $\sat$.
  \qed
\end{proof}

\vspace{-0.0mm}
\section{Noodles for the example in \cref{sec:overview}}\label{sec:noodles_example}
\vspace{-0.0mm}

The noodles for the example in \cref{sec:overview} are in \cref{fig:ex-noodles}.

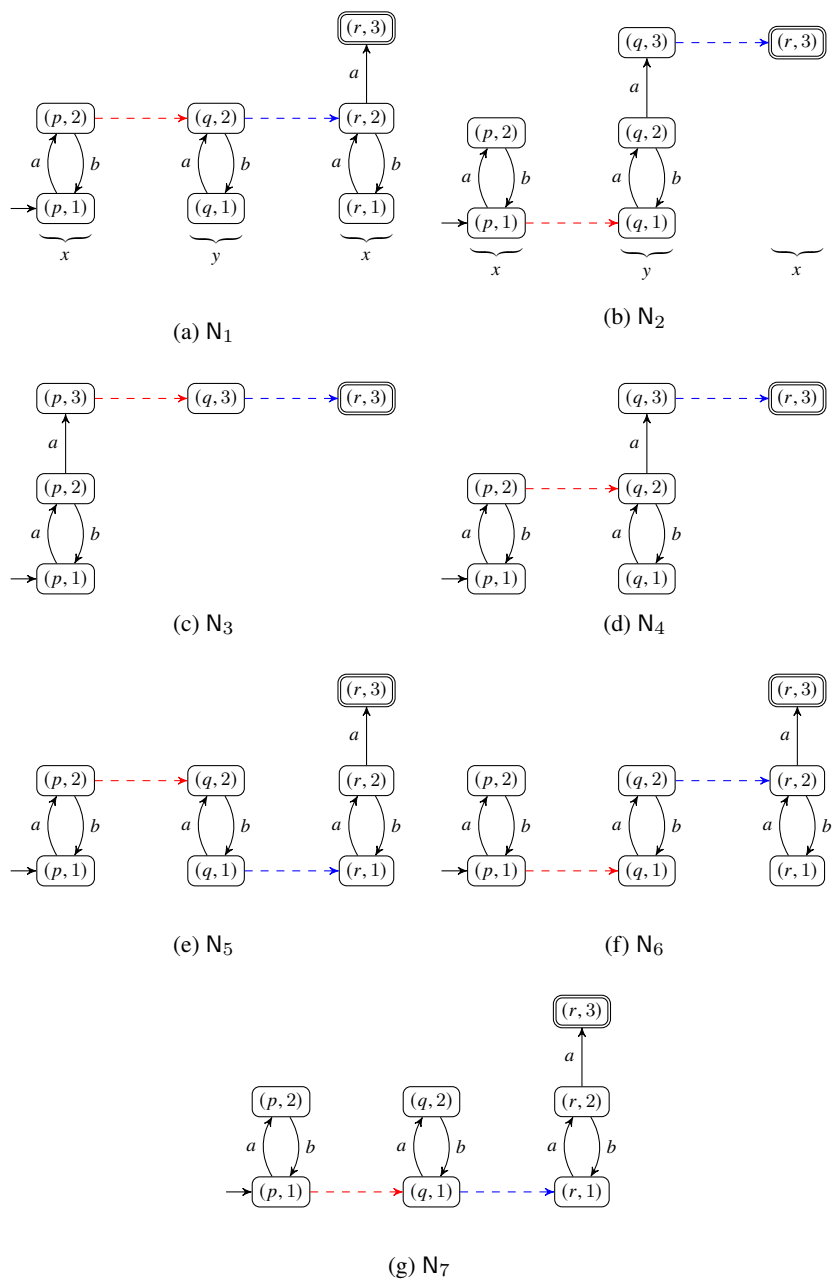
\begin{figure}
  \centering
  \begin{subfigure}{0.45\textwidth}
    \begin{tikzpicture}[->,>=stealth',shorten >=0pt,auto,node distance=15mm,transform shape,scale=0.8]
    \tikzset{prod/.style={draw,rectangle,rounded corners=1mm,inner sep=3pt}}
    \tikzset{hidden/.style={draw=none,fill=white,opacity=0}}

    \node[prod, initial, initial text={}] (p1) {$(p,1)$};
    \node[prod,above of=p1] (p2) {$(p,2)$};
    \node[hidden,above of=p2] (p3) {$(p,3)$};

    \node[prod,right of=p1, node distance=25mm] (q1) {$(q,1)$};
    \node[prod,above of=q1] (q2) {$(q,2)$};
    \node[hidden,above of=q2] (q3) {$(q,3)$};

    \node[prod,right of=q1, node distance=25mm] (r1) {$(r,1)$};
    \node[prod,above of=r1] (r2) {$(r,2)$};
    \node[prod,above of=r2,accepting] (r3) {$(r,3)$};

    \node[below of=p1,yshift=9mm] (p_below) {$\underset{\displaystyle x}{\underbrace{~~~~~~~~~~~}}$};
    \node[below of=q1,yshift=8.5mm] (q_below) {$\underset{\displaystyle y}{\underbrace{~~~~~~~~~~~}}$};
    \node[below of=r1,yshift=9mm] (r_below) {$\underset{\displaystyle x}{\underbrace{~~~~~~~~~~~}}$};

    \draw (p1) edge[bend left] node{$a$} (p2)
          (p2) edge[bend left] node{$b$} (p1)
          (q1) edge[bend left] node{$a$} (q2)
          (q2) edge[bend left] node{$b$} (q1)
          (r1) edge[bend left] node{$a$} (r2)
          (r2) edge[bend left] node{$b$} (r1)
          (r2) edge node{$a$} (r3);
    \draw[dashed,->,color=red] (p2) -- (q2);
    \draw[dashed,->,color=blue]      (q2) -- (r2);

  \end{tikzpicture}
    \caption{$\noodle_1$}
  \end{subfigure}
  ~
  \begin{subfigure}{0.45\textwidth}
    \begin{tikzpicture}[->,>=stealth',shorten >=0pt,auto,node distance=15mm,transform shape,scale=0.8]
    \tikzset{prod/.style={draw,rectangle,rounded corners=1mm,inner sep=3pt}}
    \tikzset{hidden/.style={draw=none,fill=white,opacity=0}}

    \node[prod, initial, initial text={}] (p1) {$(p,1)$};
    \node[prod,above of=p1] (p2) {$(p,2)$};
    \node[hidden,above of=p2] (p3) {$(p,3)$};

    \node[prod,right of=p1, node distance=25mm] (q1) {$(q,1)$};
    \node[prod,above of=q1] (q2) {$(q,2)$};
    \node[prod,above of=q2] (q3) {$(q,3)$};

    \node[hidden,right of=q1, node distance=25mm] (r1) {$(r,1)$};
    \node[hidden,above of=r1] (r2) {$(r,2)$};
    \node[prod,above of=r2,accepting] (r3) {$(r,3)$};

    \node[below of=p1,yshift=9mm] (p_below) {$\underset{\displaystyle x}{\underbrace{~~~~~~~~~~~}}$};
    \node[below of=q1,yshift=8.5mm] (q_below) {$\underset{\displaystyle y}{\underbrace{~~~~~~~~~~~}}$};
    \node[below of=r1,yshift=9mm] (r_below) {$\underset{\displaystyle x}{\underbrace{~~~~~~~~~~~}}$};

    \draw (p1) edge[bend left] node{$a$} (p2)
          (p2) edge[bend left] node{$b$} (p1)
          (q1) edge[bend left] node{$a$} (q2)
          (q2) edge[bend left] node{$b$} (q1)
          (q2) edge node{$a$} (q3);
    \draw[dashed,->,color=red]      (p1) -- (q1);
    \draw[dashed,->,color=blue]      (q3) -- (r3);

  \end{tikzpicture}
    \caption{$\noodle_2$}
  \end{subfigure}
  
  \bigskip
  \begin{subfigure}{0.45\textwidth}
    \begin{tikzpicture}[->,>=stealth',shorten >=0pt,auto,node distance=15mm,transform shape,scale=0.8]
    \tikzset{prod/.style={draw,rectangle,rounded corners=1mm,inner sep=3pt}}
    \tikzset{hidden/.style={draw=none,fill=white,opacity=0}}

    \node[prod, initial, initial text={}] (p1) {$(p,1)$};
    \node[prod,above of=p1] (p2) {$(p,2)$};
    \node[prod,above of=p2] (p3) {$(p,3)$};

    \node[hidden,right of=p1, node distance=25mm] (q1) {$(q,1)$};
    \node[hidden,above of=q1] (q2) {$(q,2)$};
    \node[prod,above of=q2] (q3) {$(q,3)$};

    \node[hidden,right of=q1, node distance=25mm] (r1) {$(r,1)$};
    \node[hidden,above of=r1] (r2) {$(r,2)$};
    \node[prod,above of=r2,accepting] (r3) {$(r,3)$};

    \draw (p1) edge[bend left] node{$a$} (p2)
          (p2) edge[bend left] node{$b$} (p1)
          (p2) edge node{$a$} (p3);
    \draw[dashed,->,color=red]      (p3) -- (q3);
    \draw[dashed,->,color=blue]      (q3) -- (r3);

  \end{tikzpicture}
    \caption{$\noodle_3$}
  \end{subfigure}
  ~
  \begin{subfigure}{0.45\textwidth}
    \begin{tikzpicture}[->,>=stealth',shorten >=0pt,auto,node distance=15mm,transform shape,scale=0.8]
    \tikzset{prod/.style={draw,rectangle,rounded corners=1mm,inner sep=3pt}}
    \tikzset{hidden/.style={draw=none,fill=white,opacity=0}}

    \node[prod, initial, initial text={}] (p1) {$(p,1)$};
    \node[prod,above of=p1] (p2) {$(p,2)$};
    \node[hidden,above of=p2] (p3) {$(p,3)$};

    \node[prod,right of=p1, node distance=25mm] (q1) {$(q,1)$};
    \node[prod,above of=q1] (q2) {$(q,2)$};
    \node[prod,above of=q2] (q3) {$(q,3)$};

    \node[hidden,right of=q1, node distance=25mm] (r1) {$(r,1)$};
    \node[hidden,above of=r1] (r2) {$(r,2)$};
    \node[prod,above of=r2,accepting] (r3) {$(r,3)$};

    \draw (p1) edge[bend left] node{$a$} (p2)
          (p2) edge[bend left] node{$b$} (p1)
          (q1) edge[bend left] node{$a$} (q2)
          (q2) edge[bend left] node{$b$} (q1)
          (q2) edge node{$a$} (q3);
    \draw[dashed,->,color=red] (p2) -- (q2);
    \draw[dashed,->,color=blue]      (q3) -- (r3);

  \end{tikzpicture}
    \caption{$\noodle_4$}
  \end{subfigure}

  \bigskip
  \begin{subfigure}{0.45\textwidth}
    \begin{tikzpicture}[->,>=stealth',shorten >=0pt,auto,node distance=15mm,transform shape,scale=0.8]
    \tikzset{prod/.style={draw,rectangle,rounded corners=1mm,inner sep=3pt}}
    \tikzset{hidden/.style={draw=none,fill=white,opacity=0}}

    \node[prod, initial, initial text={}] (p1) {$(p,1)$};
    \node[prod,above of=p1] (p2) {$(p,2)$};
    \node[hidden,above of=p2] (p3) {$(p,3)$};

    \node[prod,right of=p1, node distance=25mm] (q1) {$(q,1)$};
    \node[prod,above of=q1] (q2) {$(q,2)$};
    \node[hidden,above of=q2] (q3) {$(q,3)$};

    \node[prod,right of=q1, node distance=25mm] (r1) {$(r,1)$};
    \node[prod,above of=r1] (r2) {$(r,2)$};
    \node[prod,above of=r2,accepting] (r3) {$(r,3)$};

    \draw (p1) edge[bend left] node{$a$} (p2)
          (p2) edge[bend left] node{$b$} (p1)
          (q1) edge[bend left] node{$a$} (q2)
          (q2) edge[bend left] node{$b$} (q1)
          (r1) edge[bend left] node{$a$} (r2)
          (r2) edge[bend left] node{$b$} (r1)
          (r2) edge node{$a$} (r3);
    \draw[dashed,->,color=red] (p2) -- (q2);
    \draw[dashed,->,color=blue]      (q1) -- (r1);

  \end{tikzpicture}
    \caption{$\noodle_5$}
  \end{subfigure}
  ~
  \begin{subfigure}{0.45\textwidth}
    \begin{tikzpicture}[->,>=stealth',shorten >=0pt,auto,node distance=15mm,transform shape,scale=0.8]
    \tikzset{prod/.style={draw,rectangle,rounded corners=1mm,inner sep=3pt}}
    \tikzset{hidden/.style={draw=none,fill=white,opacity=0}}

    \node[prod, initial, initial text={}] (p1) {$(p,1)$};
    \node[prod,above of=p1] (p2) {$(p,2)$};
    \node[hidden,above of=p2] (p3) {$(p,3)$};

    \node[prod,right of=p1, node distance=25mm] (q1) {$(q,1)$};
    \node[prod,above of=q1] (q2) {$(q,2)$};
    \node[hidden,above of=q2] (q3) {$(q,3)$};

    \node[prod,right of=q1, node distance=25mm] (r1) {$(r,1)$};
    \node[prod,above of=r1] (r2) {$(r,2)$};
    \node[prod,above of=r2,accepting] (r3) {$(r,3)$};

    \draw (p1) edge[bend left] node{$a$} (p2)
          (p2) edge[bend left] node{$b$} (p1)
          (q1) edge[bend left] node{$a$} (q2)
          (q2) edge[bend left] node{$b$} (q1)
          (r1) edge[bend left] node{$a$} (r2)
          (r2) edge[bend left] node{$b$} (r1)
          (r2) edge node{$a$} (r3);
    \draw[dashed,->,color=red]      (p1) -- (q1);
    \draw[dashed,->,color=blue]      (q2) -- (r2);

  \end{tikzpicture}
    \caption{$\noodle_6$}
  \end{subfigure}

  \bigskip
  \begin{subfigure}{0.45\textwidth}
    \begin{tikzpicture}[->,>=stealth',shorten >=0pt,auto,node distance=15mm,transform shape,scale=0.8]
    \tikzset{prod/.style={draw,rectangle,rounded corners=1mm,inner sep=3pt}}
    \tikzset{hidden/.style={draw=none,fill=white,opacity=0}}

    \node[prod, initial, initial text={}] (p1) {$(p,1)$};
    \node[prod,above of=p1] (p2) {$(p,2)$};
    \node[hidden,above of=p2] (p3) {$(p,3)$};

    \node[prod,right of=p1, node distance=25mm] (q1) {$(q,1)$};
    \node[prod,above of=q1] (q2) {$(q,2)$};
    \node[hidden,above of=q2] (q3) {$(q,3)$};

    \node[prod,right of=q1, node distance=25mm] (r1) {$(r,1)$};
    \node[prod,above of=r1] (r2) {$(r,2)$};
    \node[prod,above of=r2,accepting] (r3) {$(r,3)$};

    \draw (p1) edge[bend left] node{$a$} (p2)
          (p2) edge[bend left] node{$b$} (p1)
          (q1) edge[bend left] node{$a$} (q2)
          (q2) edge[bend left] node{$b$} (q1)
          (r1) edge[bend left] node{$a$} (r2)
          (r2) edge[bend left] node{$b$} (r1)
          (r2) edge node{$a$} (r3);
    \draw[dashed,->,color=red]      (p1) -- (q1);
    \draw[dashed,->,color=blue]      (q1) -- (r1);

  \end{tikzpicture}
    \caption{$\noodle_7$}
  \end{subfigure}

  \caption{Noodles obtained from $\noodlify(\prod)$ based on the product from
  \cref{fig:overview}.}
  \label{fig:ex-noodles}
\end{figure}

\vspace{-0.0mm}
\section{Additional Examples}\label{sec:add-examples}
\vspace{-0.0mm}

\newcommand{\figNontermExample}[0]{
\begin{wrapfigure}[5]{r}{1.5cm}
\vspace*{-7mm}
\hspace*{-2mm}
\begin{minipage}{1.5cm}
  \centering
\scalebox{1.0}{
\begin{tikzpicture}[->,>=stealth',shorten >=0pt,auto,node
  distance=15mm,transform shape,scale=0.8]
  \tikzset{every node/.style={rounded corners=1mm}}

  \node[draw] (x<xy) {$\graphnode{x}{xy}$};
  \node[draw,below of=x<xy,yshift=2mm] (xy<x) {$\graphnode{xy}{x}$};

  \draw (x<xy) edge[bend left] (xy<x);
  \draw (xy<x) edge[bend left] (x<xy);
  \draw (x<xy) edge[loop right] (x<xy);
  \draw (xy<x) edge[loop right] (xy<x);

\end{tikzpicture}
}
\end{minipage}
\end{wrapfigure}
}

\begin{minipage}{\textwidth}

\begin{wrapfigure}[5]{r}{1.5cm}
\vspace*{-7mm}
\hspace*{-2mm}
\begin{minipage}{1.5cm}
  \centering
\scalebox{1.0}{
\begin{tikzpicture}[->,>=stealth',shorten >=0pt,auto,node
  distance=15mm,transform shape,scale=0.8]
  \tikzset{every node/.style={rounded corners=1mm}}

  \node[draw] (x<xy) {$\graphnode{x}{xy}$};
  \node[draw,below of=x<xy,yshift=2mm] (xy<x) {$\graphnode{xy}{x}$};

  \draw (x<xy) edge[bend left] (xy<x);
  \draw (xy<x) edge[bend left] (x<xy);
  \draw (x<xy) edge[loop right] (x<xy);
  \draw (xy<x) edge[loop right] (xy<x);

\end{tikzpicture}
}
\end{minipage}
\end{wrapfigure}

\beginexample\label{ex:nonterminating}
  Consider for instance the system $xy=x \wedge x\in a^+ \wedge y\in a$.
  The inclusion graph (actually the only possible) is shown on the right.
  The initial automata assignment $\constrautass$ is then given as 
  $\langof{\constrautass(x)} = a^+$ and $\langof{\constrautass(x)} = \{ a\}$.
  The queue $\worklist$ on \lnref{ln:prop:init} of \cref{alg:propagate} is hence 
  initialised as $\worklist = \lifo{(\constrautass, \lifo{\graphnode{x}{xy}, \graphnode{xy}{x}}) }$.
  The computation then looks as follows:

  \begin{description}
    \item[1st iteration.] The condition on \lnref{line:inclusiontest} is not
    satisfied and, therefore, the algorithm calls
    $\refine(\graphnode{x}{xy}, \constrautass)$.
    The refinement yields a new automata assignment $\autass_1$ refining $\constrautass$ with $\autass_1(x) = a^+a$.
    The queue $\worklist$ is then given as $\lifo{(\autass_1, \lifo{\graphnode{xy}{x}, \graphnode{x}{xy}}) }$.

    \item[2nd iteration.] The condition on \lnref{line:inclusiontest} is satisfied and, therefore, 
      the queue $\worklist$ is updated to $\lifo{(\autass_1, \lifo{\graphnode{x}{xy}}) }$.

    \item[3rd iteration.] The condition on \lnref{line:inclusiontest} is not
    satisfied and, therefore, the algorithm calls
    $\refine(\graphnode{x}{xy}, \autass_1)$.
    The refinement yields a new automata assignment $\autass_2$ refining $\autass_1$ with $\autass_2(x) = a^+a^2$.
    The queue $\worklist$ is then given as $\lifo{(\autass_2, \lifo{\graphnode{xy}{x}, \graphnode{x}{xy}}) }$.

    \item[4th iteration.] The condition on \lnref{line:inclusiontest} is satisfied and, therefore, 
      the queue $\worklist$ is updated to $\lifo{(\autass_2, \lifo{\graphnode{x}{xy}}) }$.

    \item[] \dots
  \end{description}
  It is evident that \cref{alg:propagate} does not terminate on this case (which is clearly unsatisfiable), 
  since the refined automata assignments for $x$ reach $a^+a^n$ for all $n\in\N$. \qed
\end{minipage}
  %
  
  
  

  \section{Detailed Results}
  \label{sec:det-results}

  The plots comparing times of \noodler with other tools on all benchmarks are in 
  \cref{fig:pyex-all}, \cref{fig:kaluza-all}, \cref{fig:str-all}, and \cref{fig:slog-all}.

  \begin{figure}
    \centering
    \includegraphics[width=0.32\textwidth]{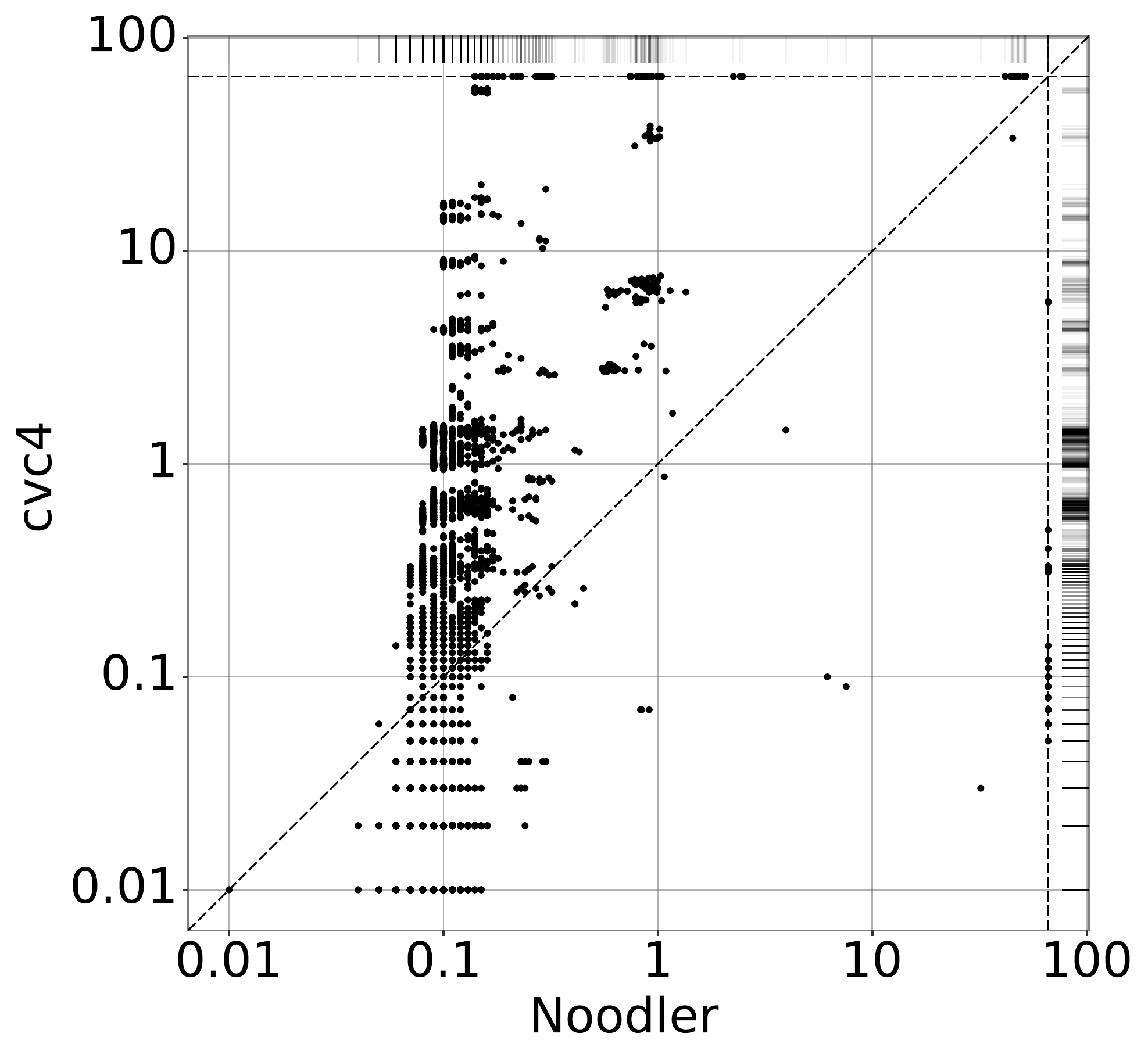}
    \includegraphics[width=0.32\textwidth]{figs/noodler/fig_noodler_vs_cvc5-pyex.pdf}
    \includegraphics[width=0.32\textwidth]{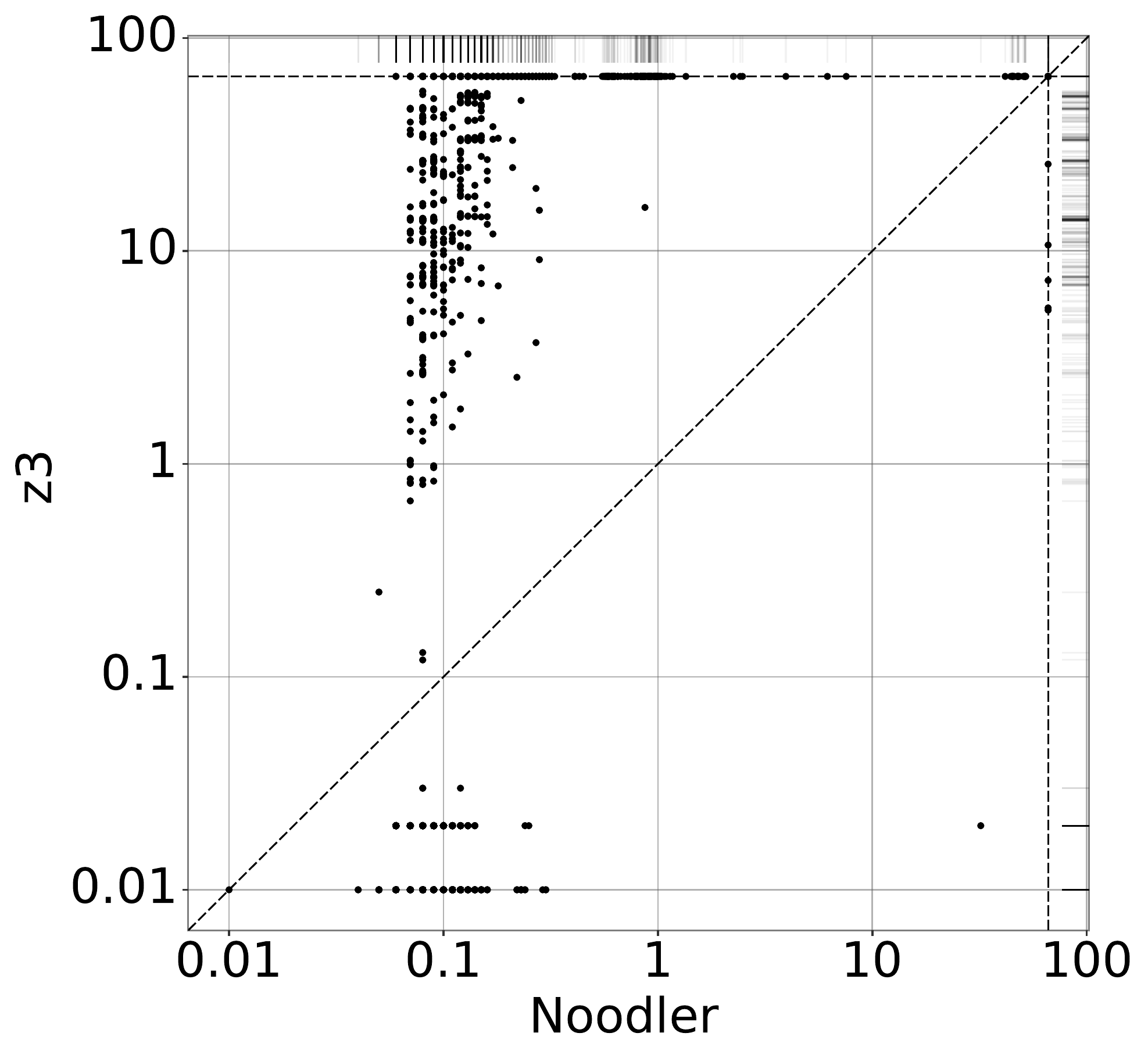}

    \includegraphics[width=0.32\textwidth]{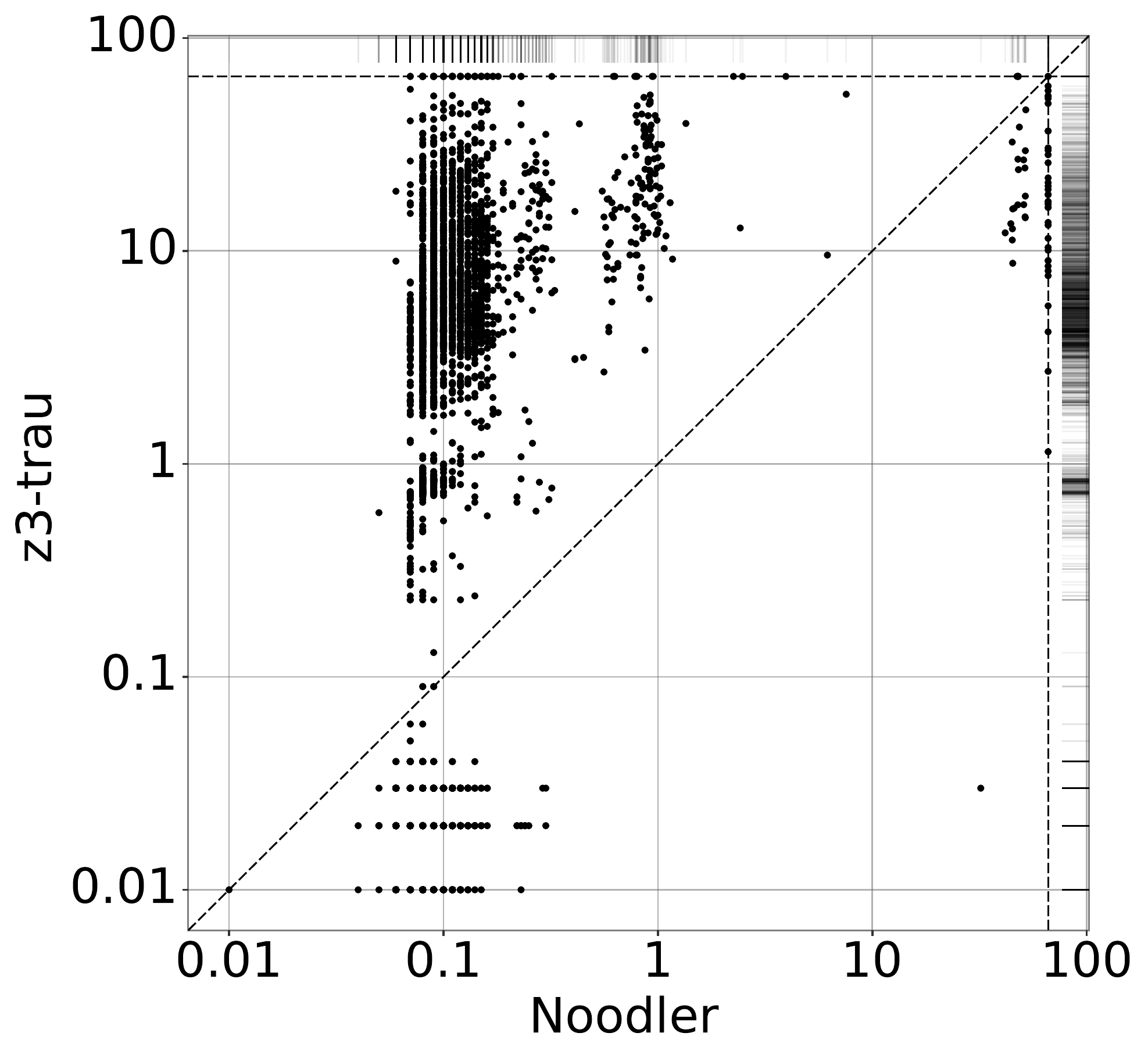}
    \includegraphics[width=0.32\textwidth]{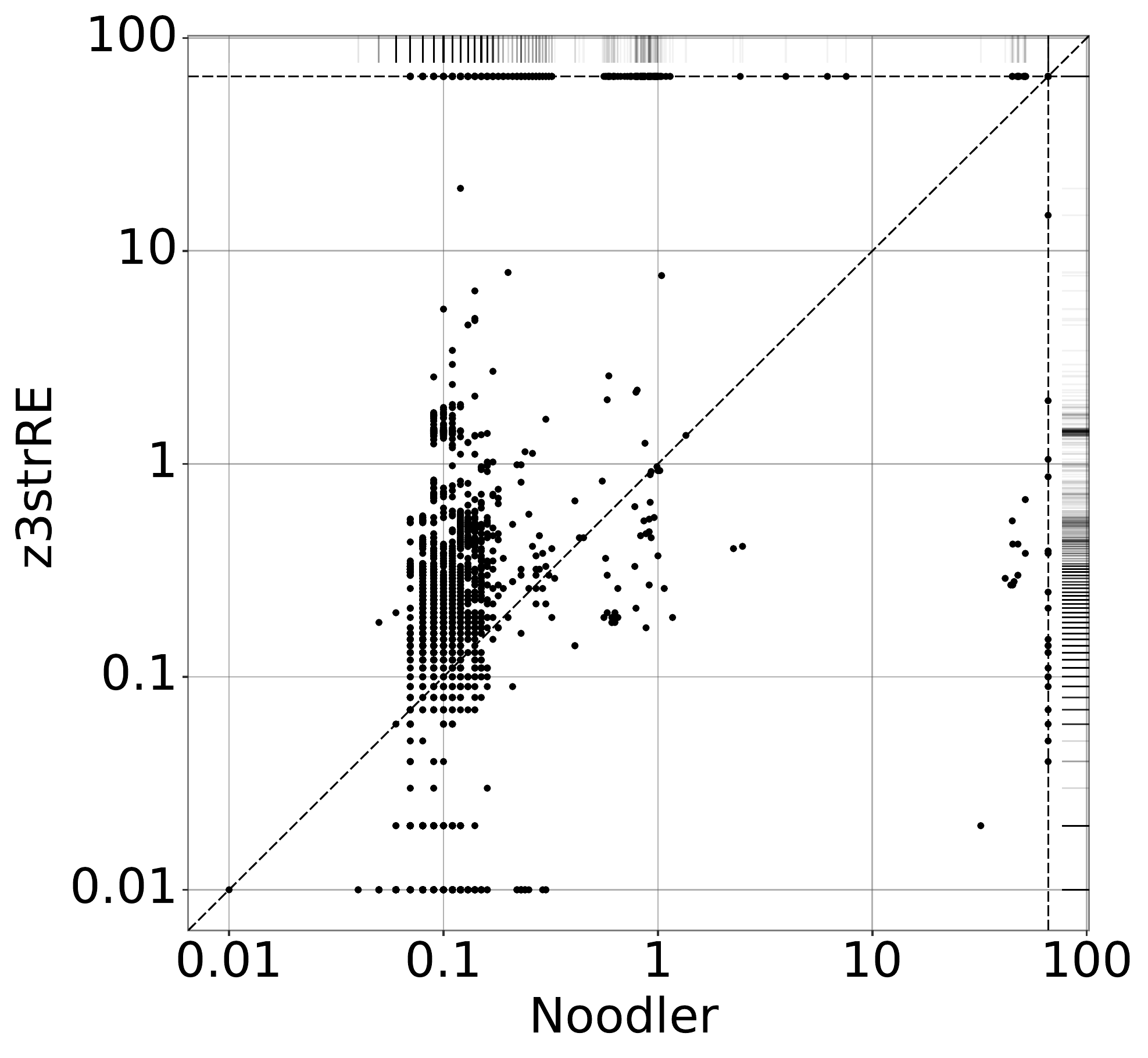}
    \includegraphics[width=0.32\textwidth]{figs/noodler/fig_noodler_vs_z3str4-pyex.pdf}

    \includegraphics[width=0.32\textwidth]{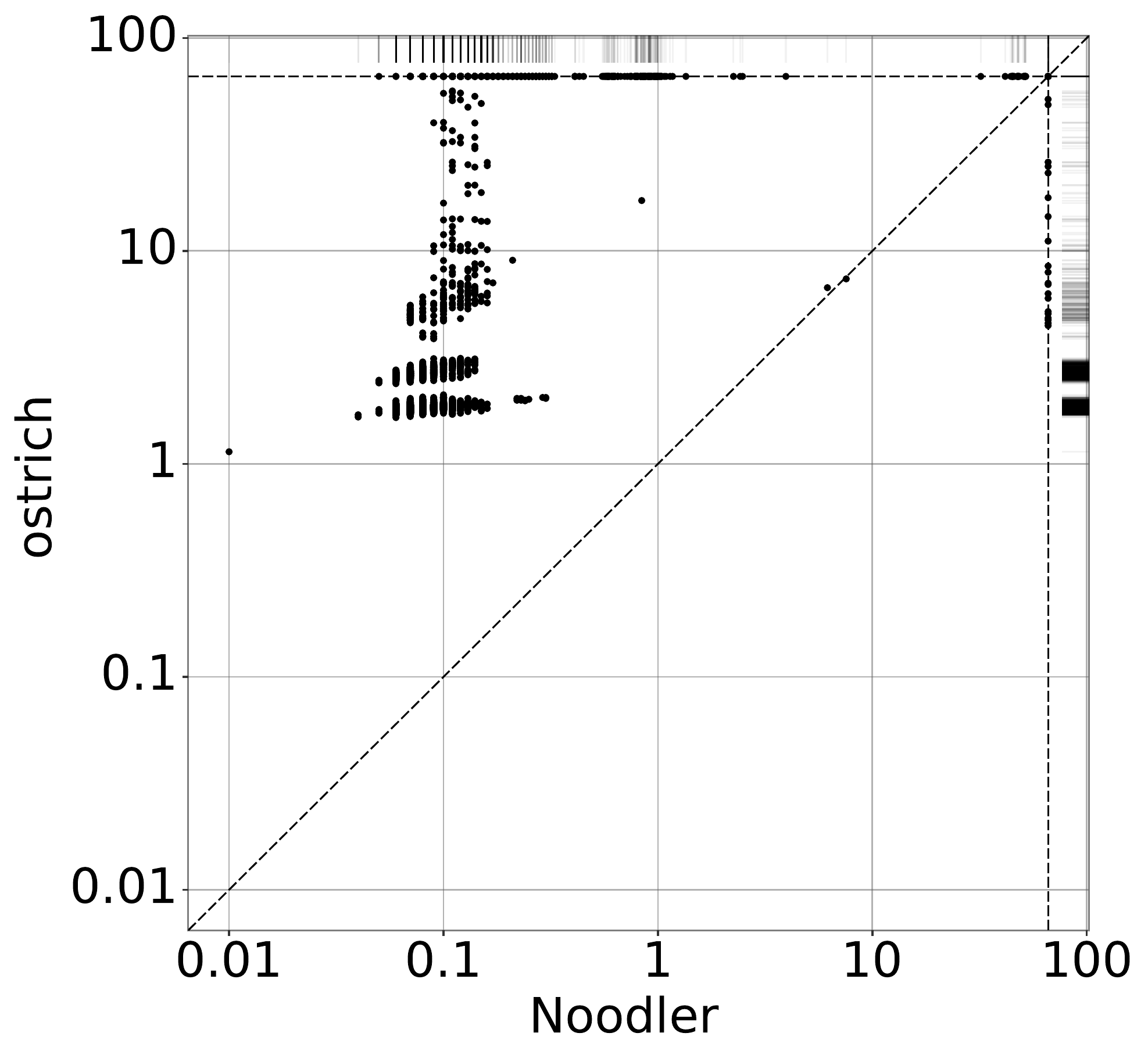}
    \includegraphics[width=0.32\textwidth]{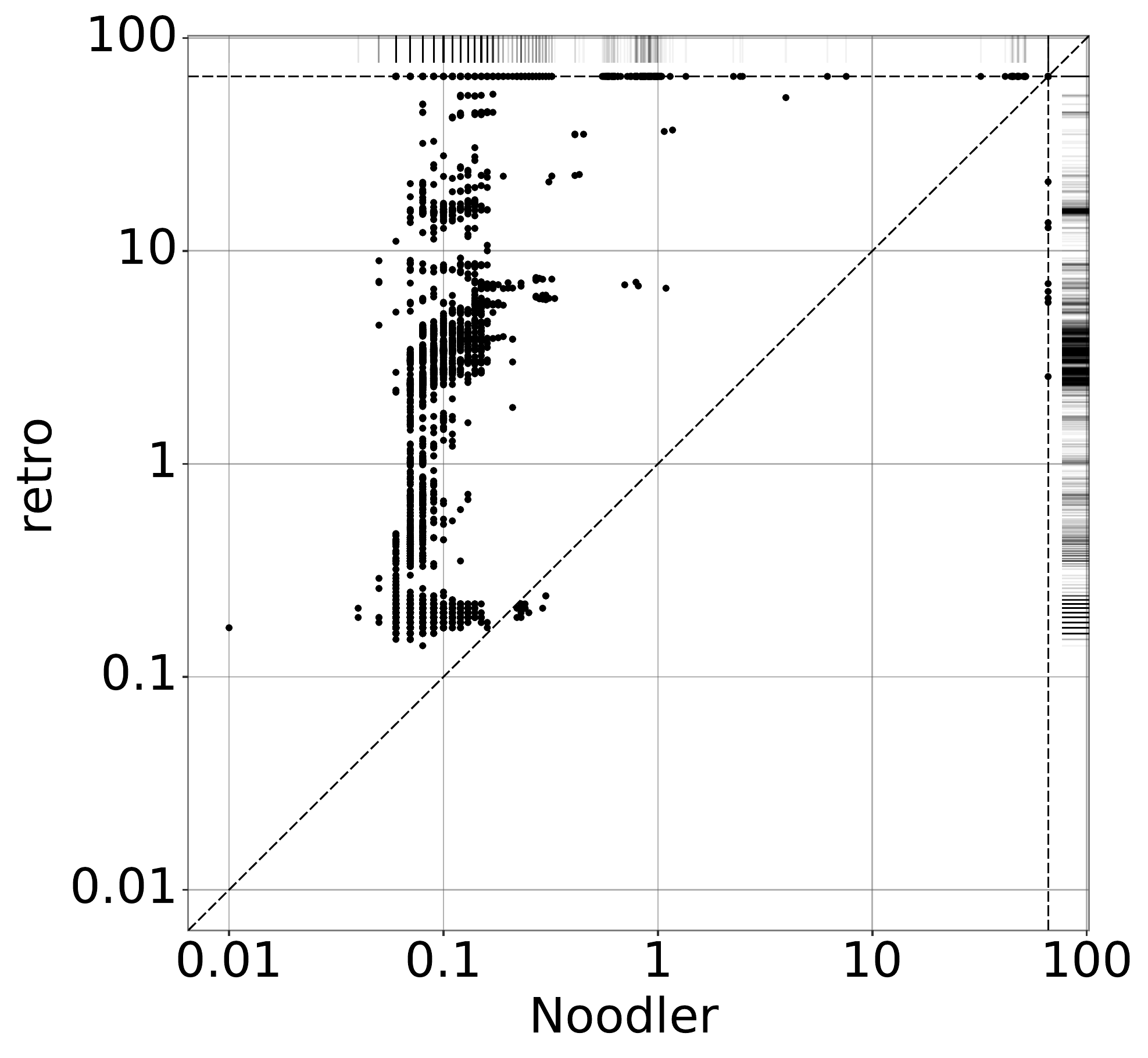}
    \includegraphics[width=0.32\textwidth]{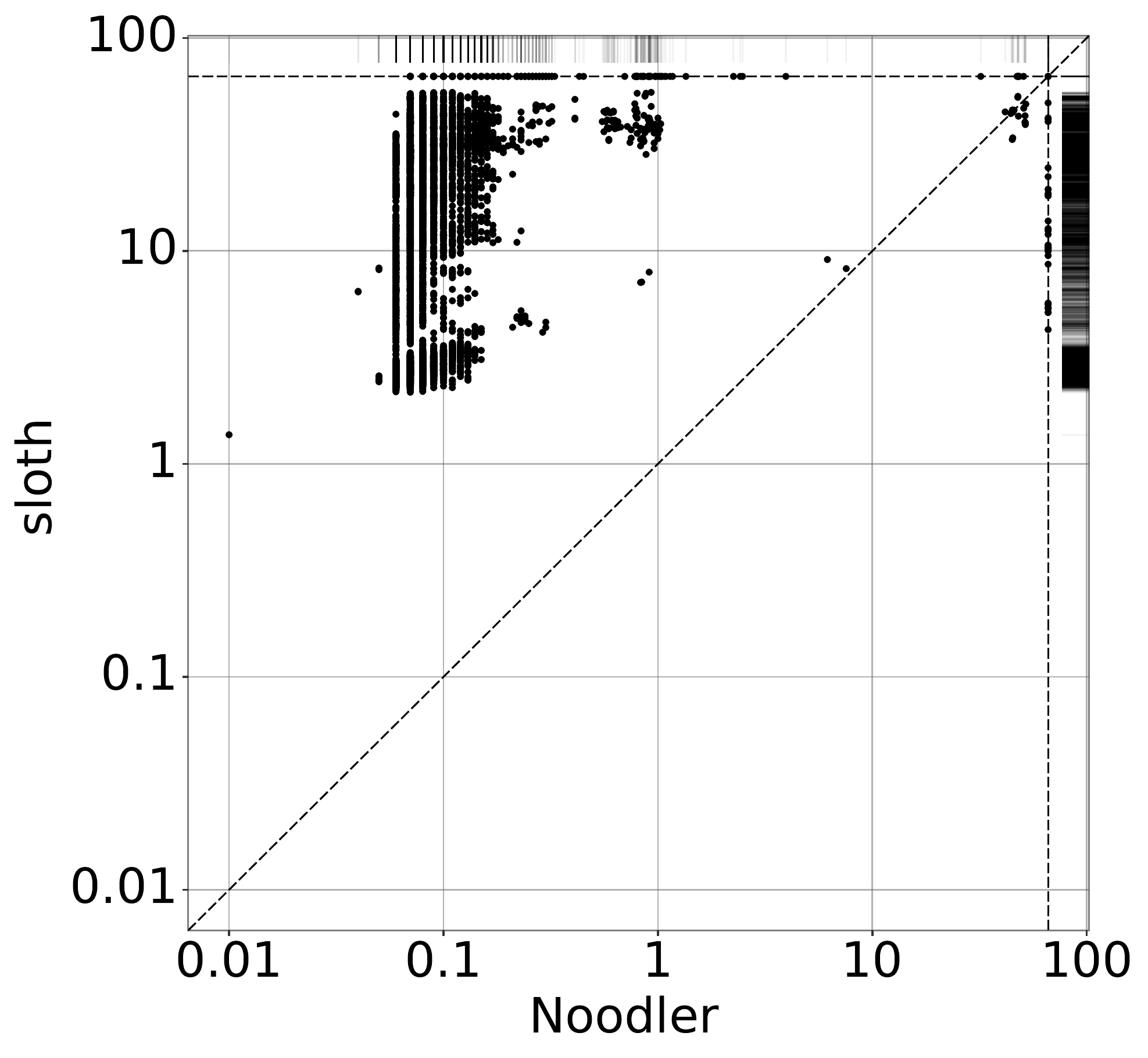}

    \caption{Comparison of run times of \noodler with all other tools (in seconds) on \pyexhard. 
    Dashed lines represent timeouts.}
    \label{fig:pyex-all}
  \end{figure}

  \begin{figure}
    \centering
    \includegraphics[width=0.32\textwidth]{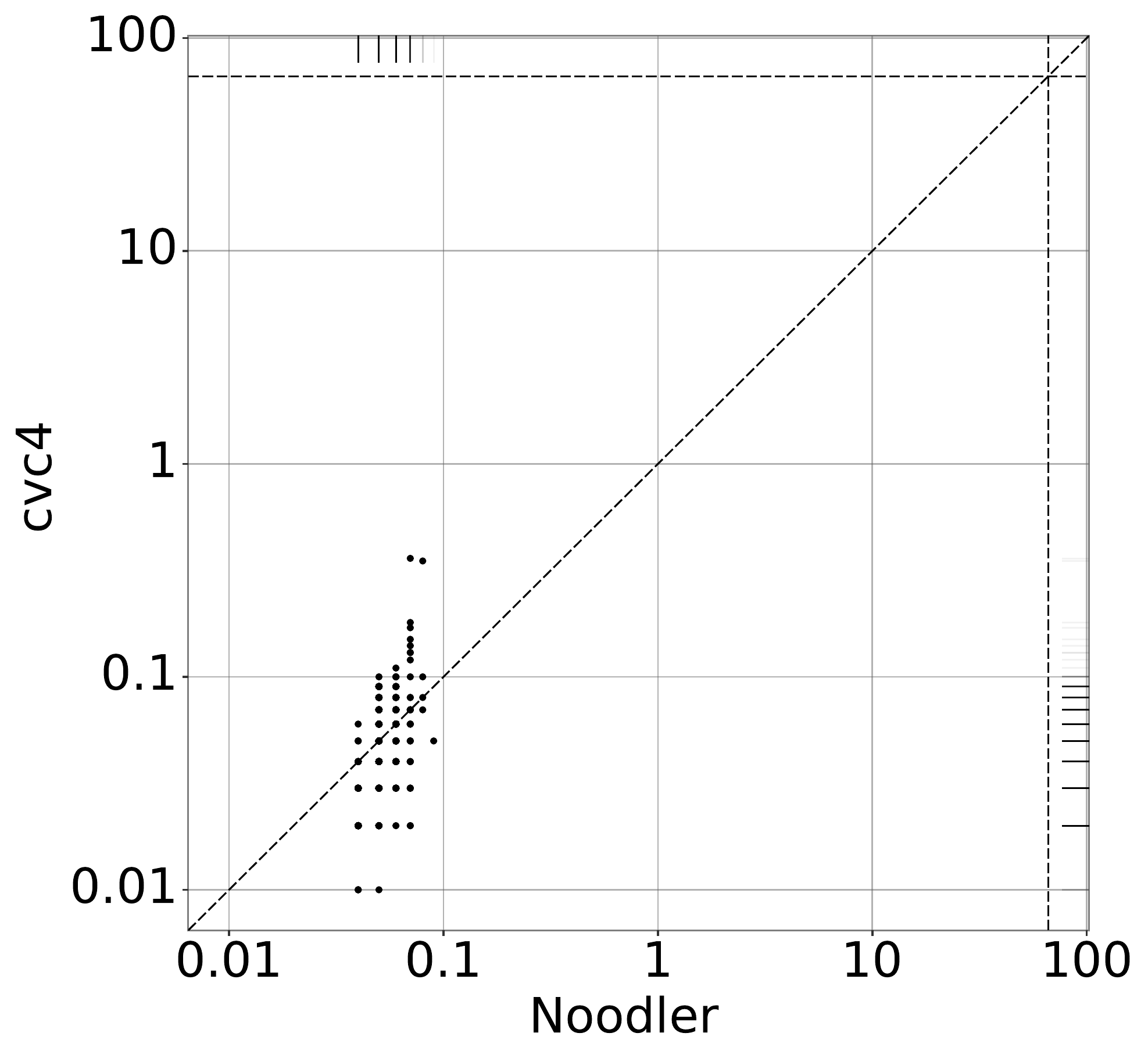}
    \includegraphics[width=0.32\textwidth]{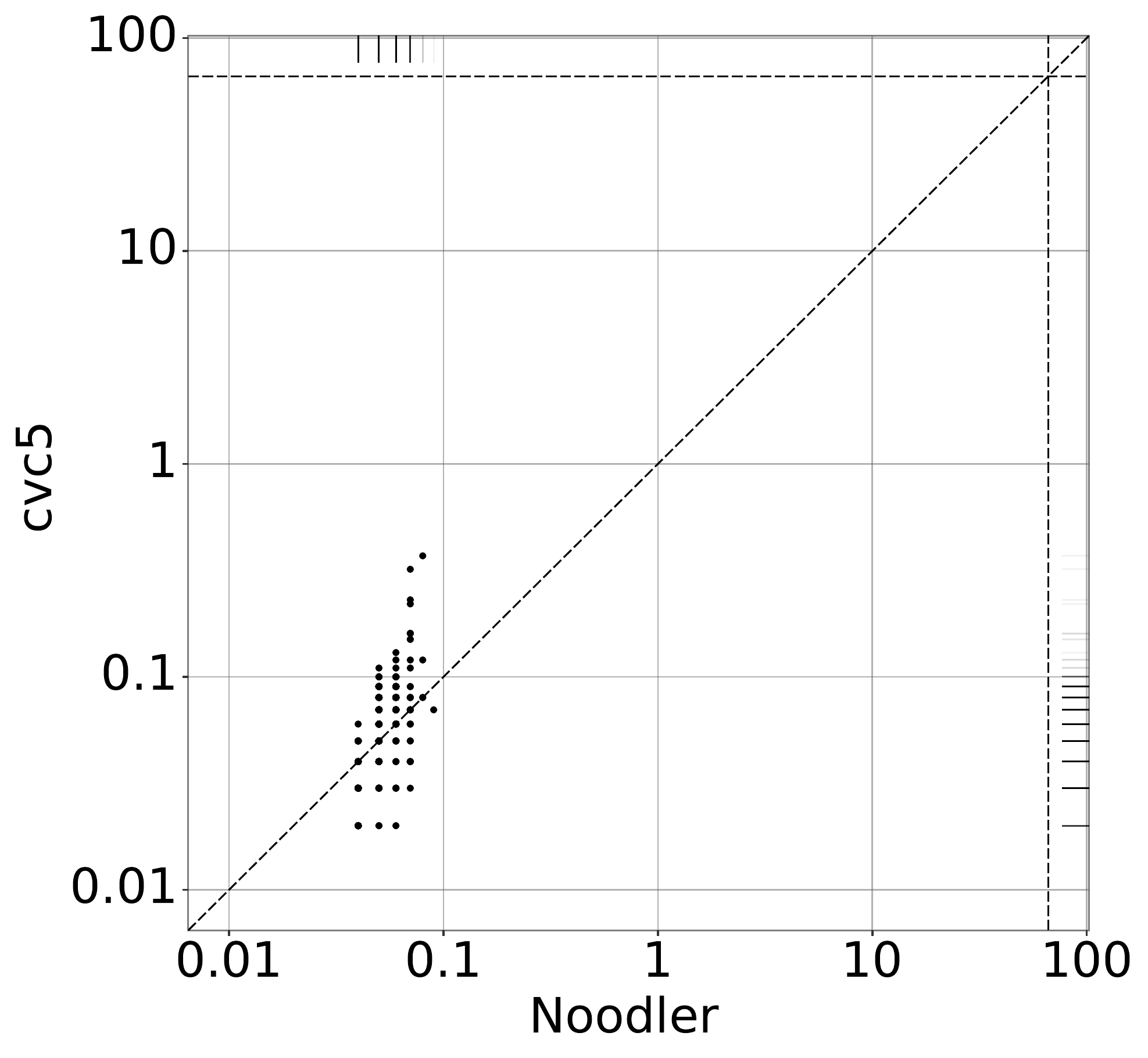}
    \includegraphics[width=0.32\textwidth]{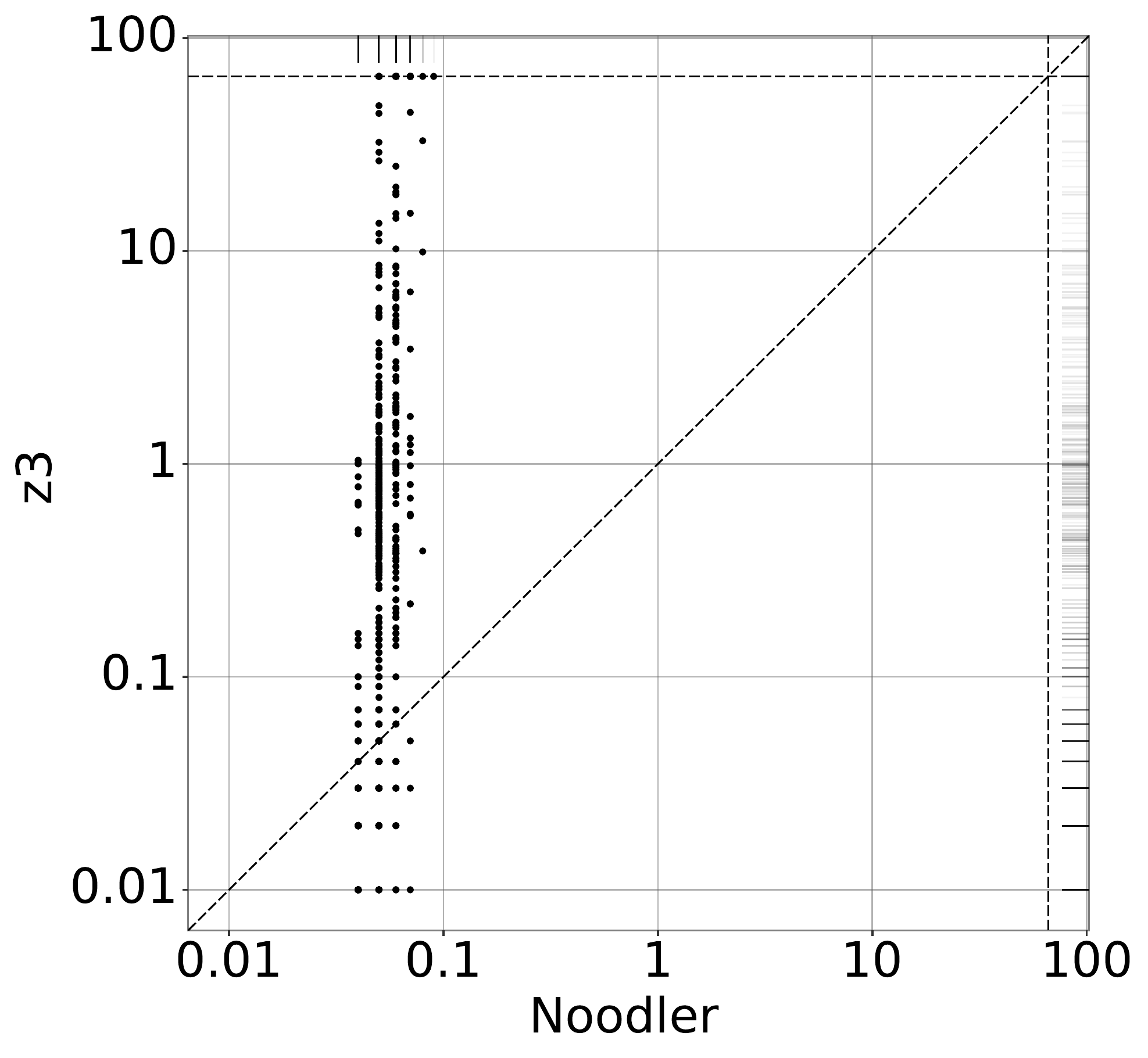}

    \includegraphics[width=0.32\textwidth]{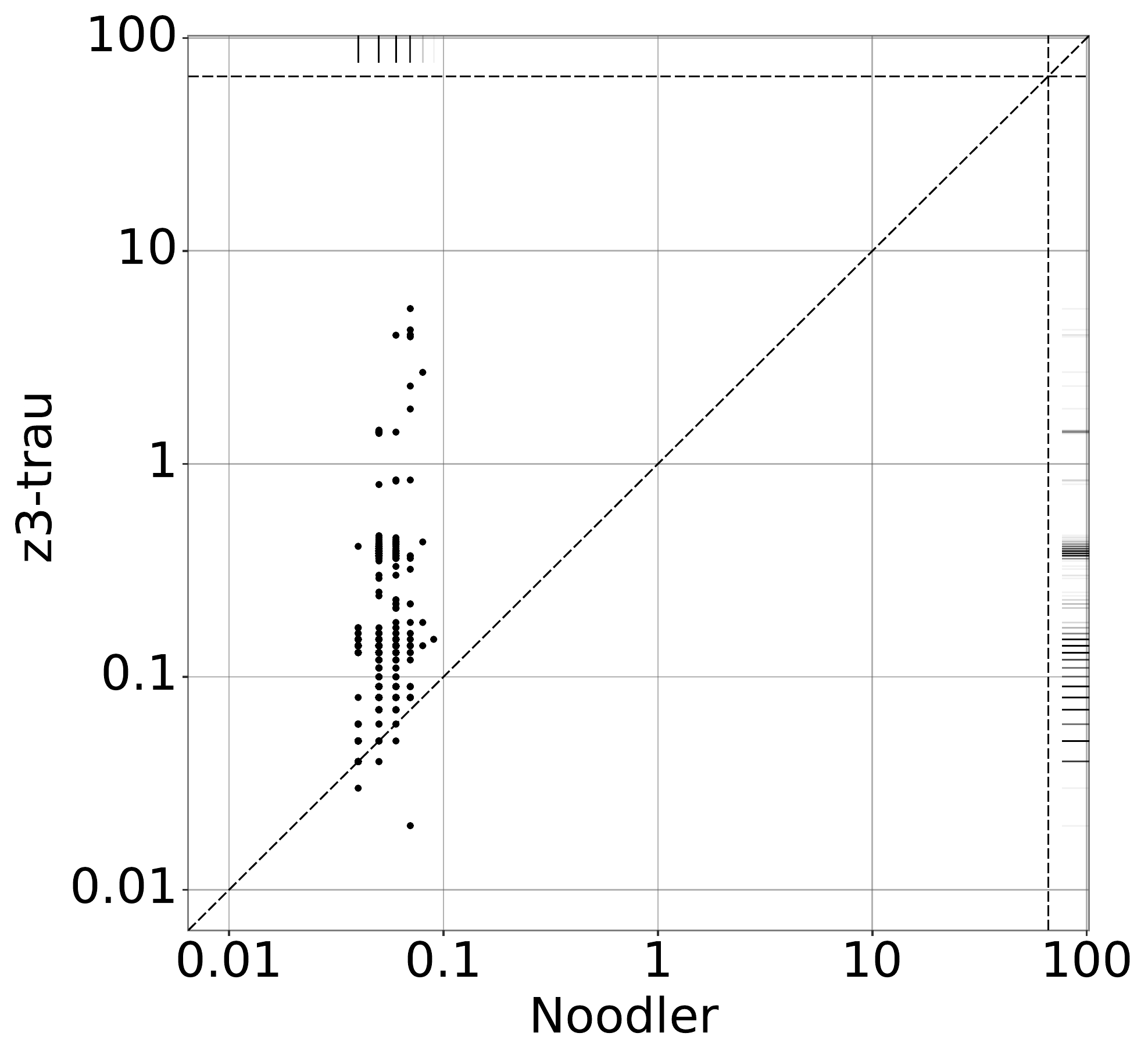}
    \includegraphics[width=0.32\textwidth]{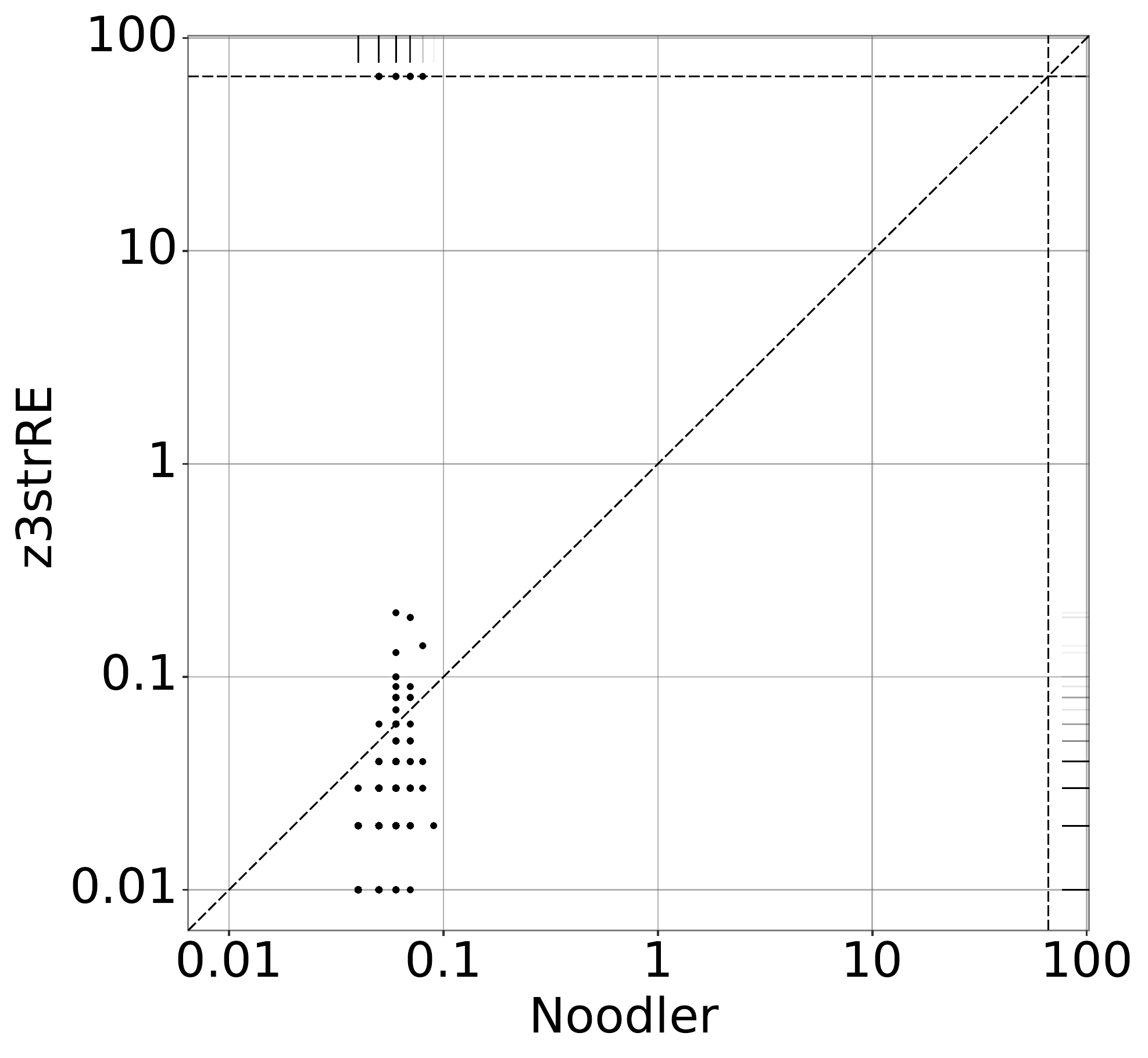}
    \includegraphics[width=0.32\textwidth]{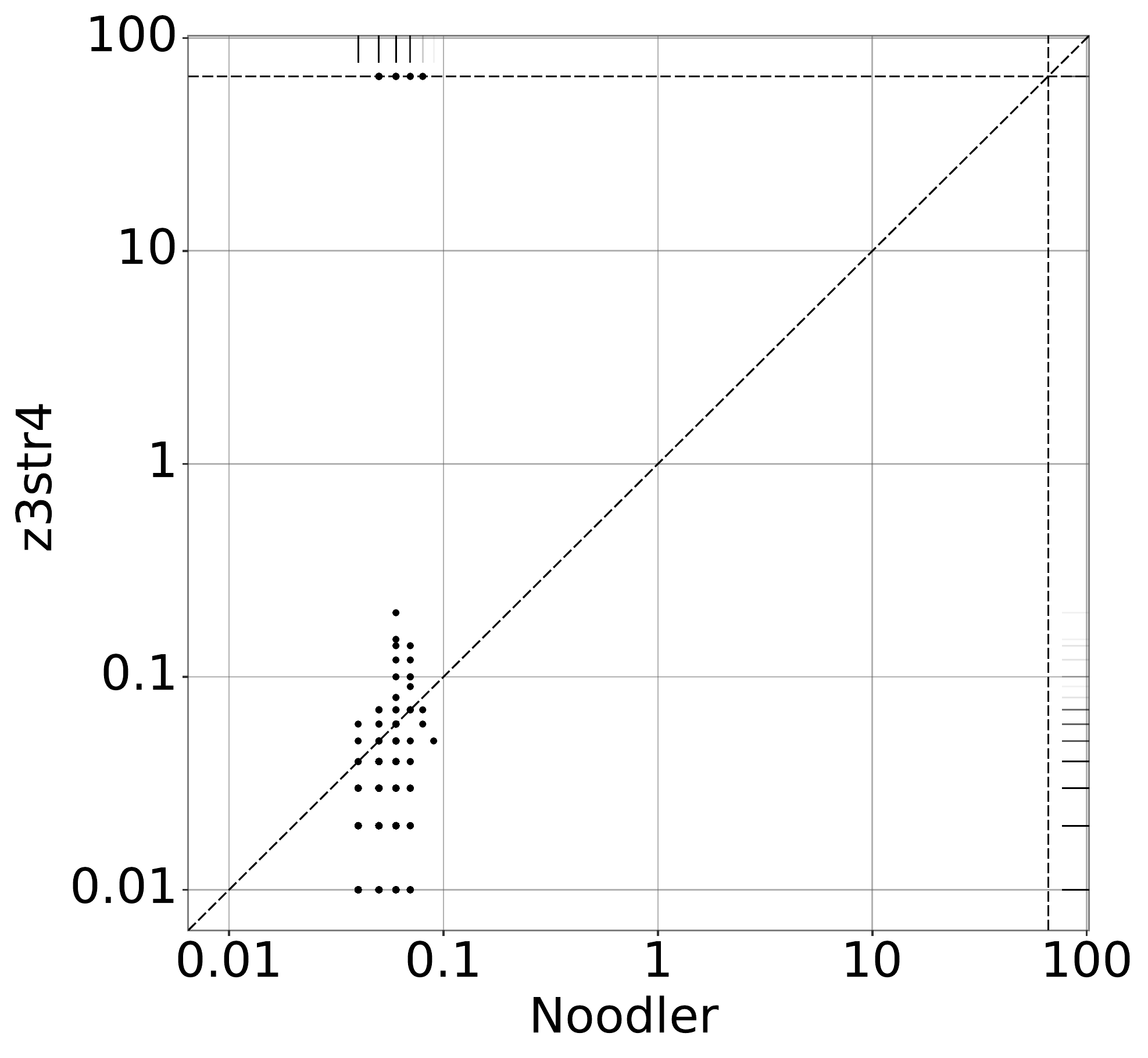}

    \includegraphics[width=0.32\textwidth]{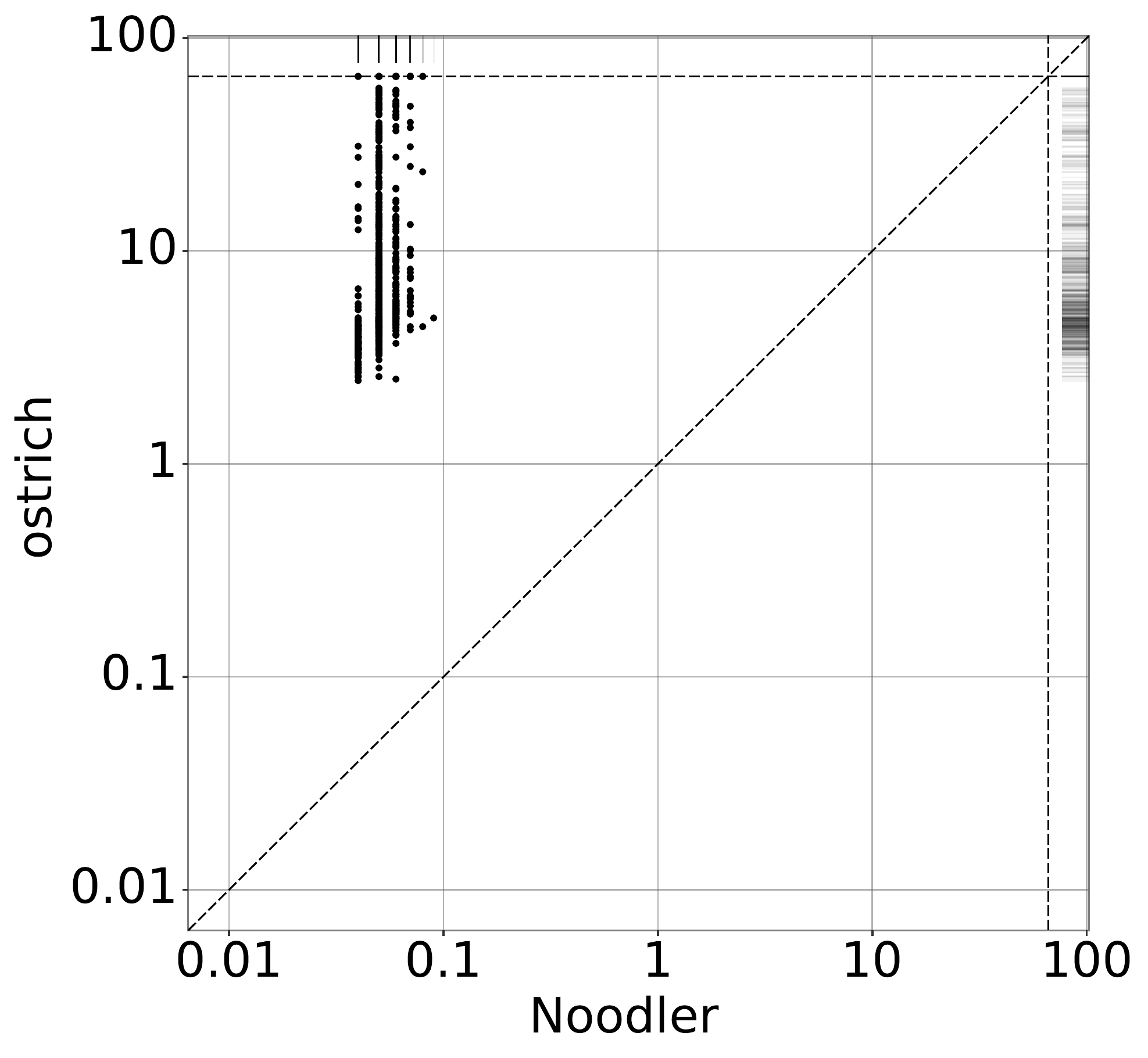}
    \includegraphics[width=0.32\textwidth]{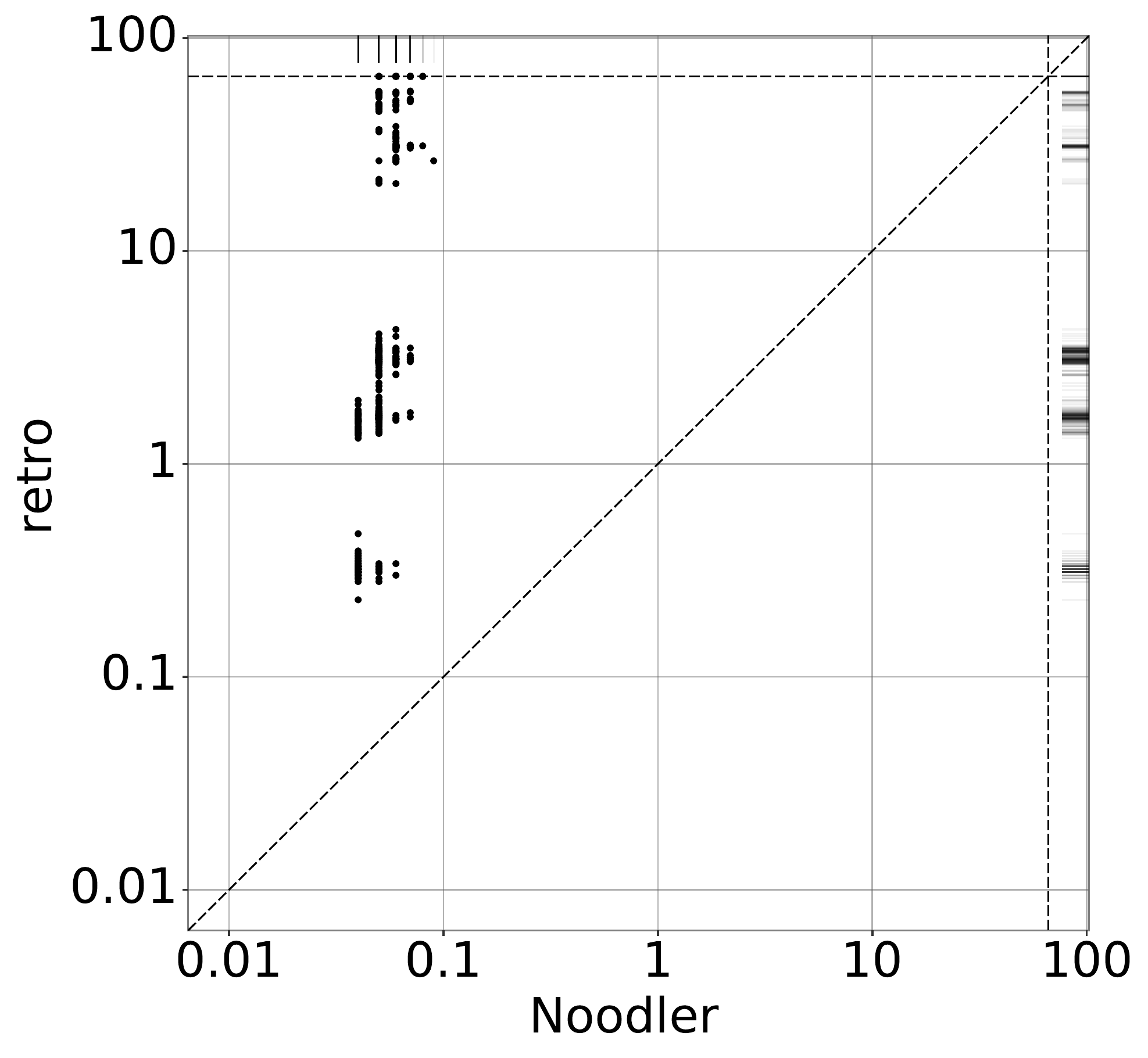}
    \includegraphics[width=0.32\textwidth]{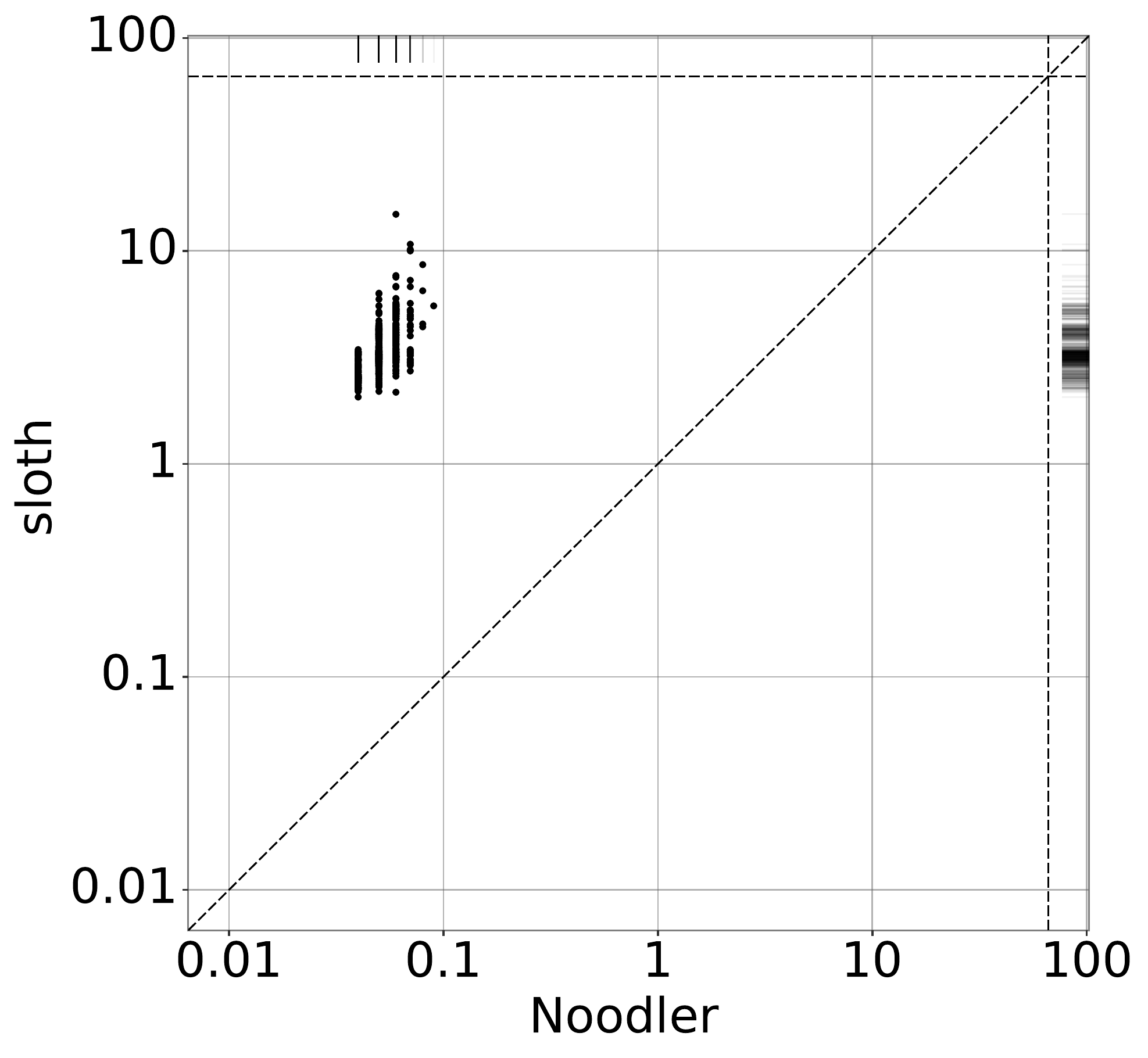}

    \caption{Comparison of run times of \noodler with all other tools (in seconds) on \kaluzahard. 
    Dashed lines represent timeouts.}
    \label{fig:kaluza-all}
  \end{figure}

  \begin{figure}
    \centering
    \includegraphics[width=0.32\textwidth]{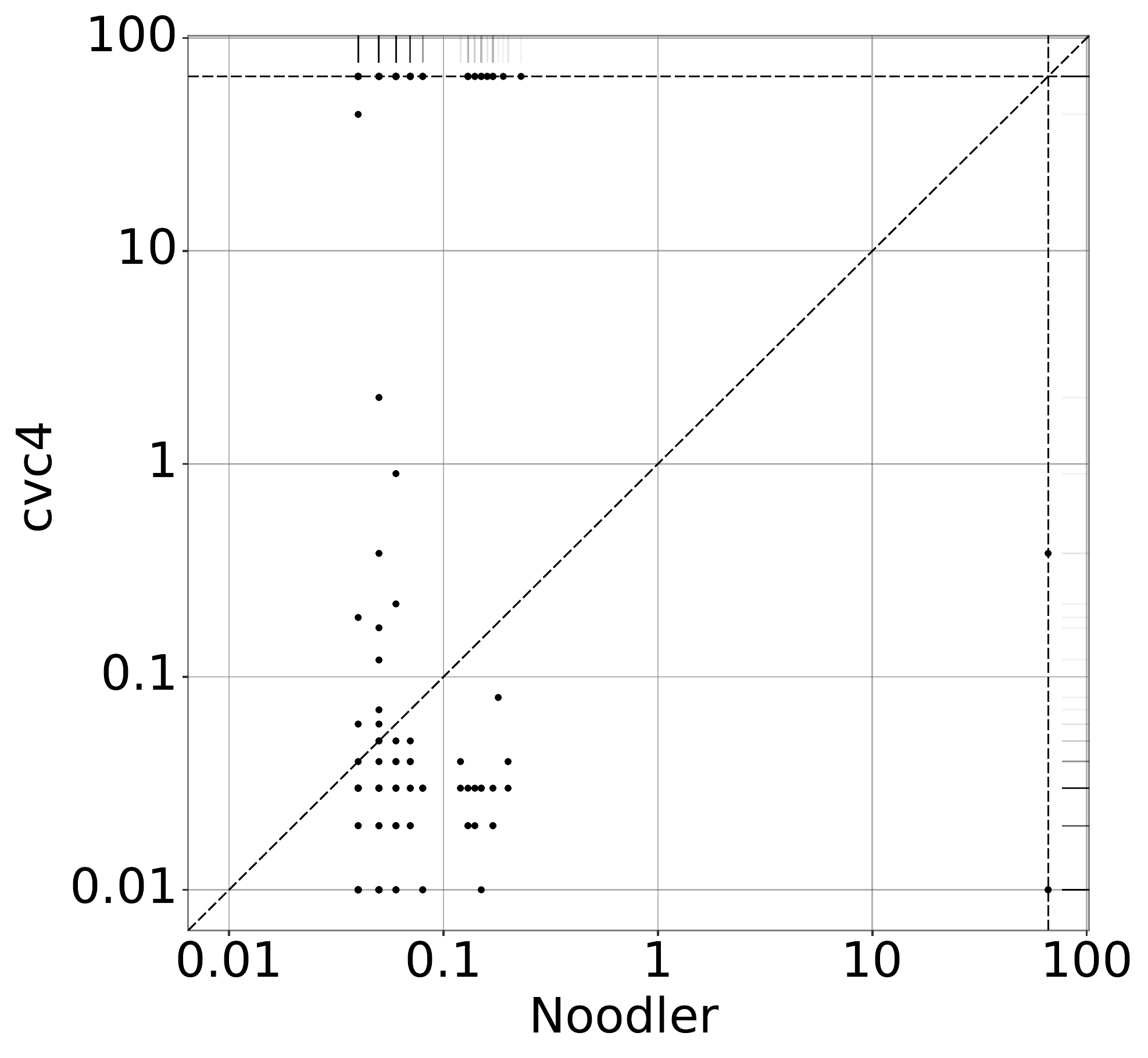}
    \includegraphics[width=0.32\textwidth]{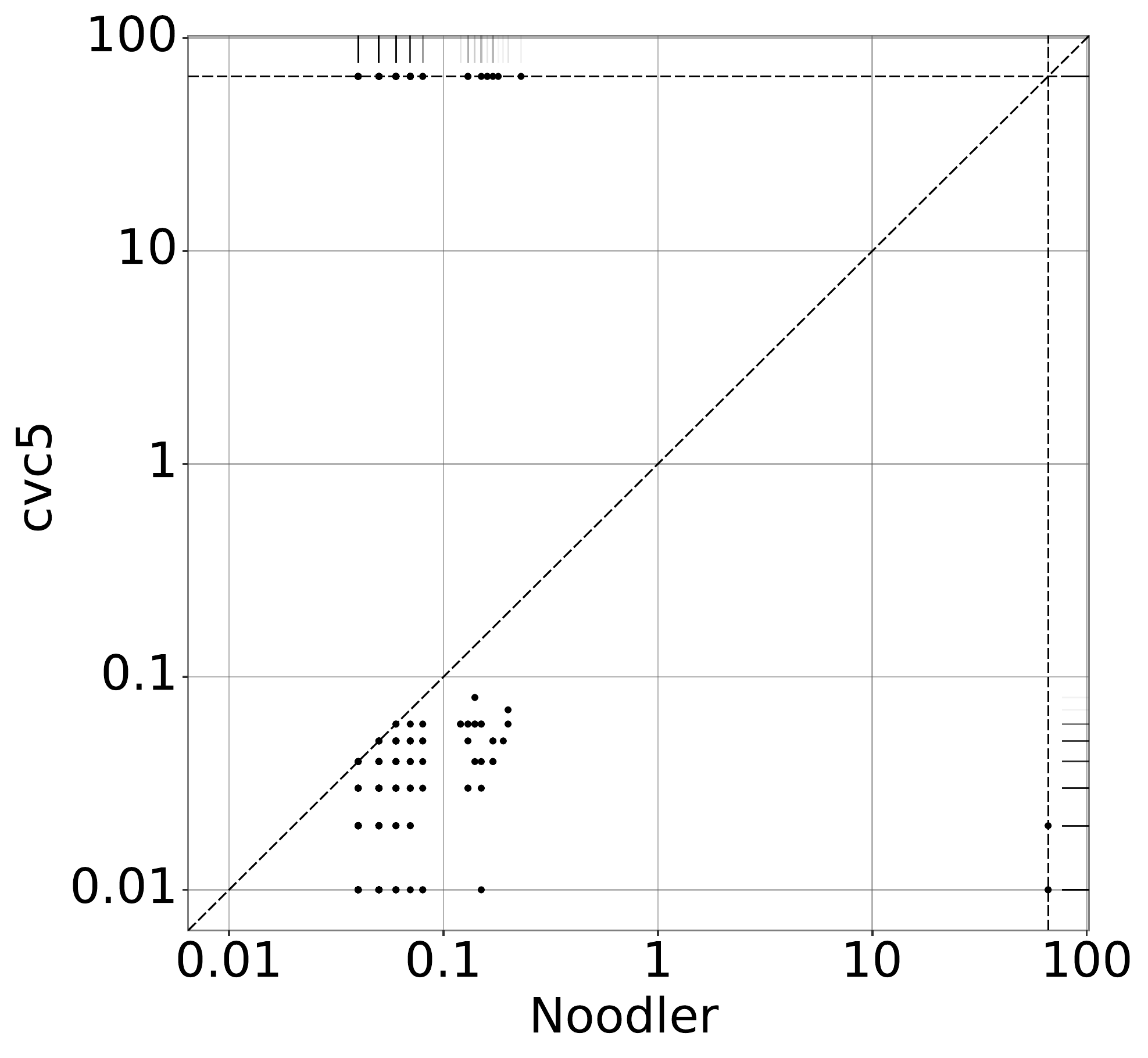}
    \includegraphics[width=0.32\textwidth]{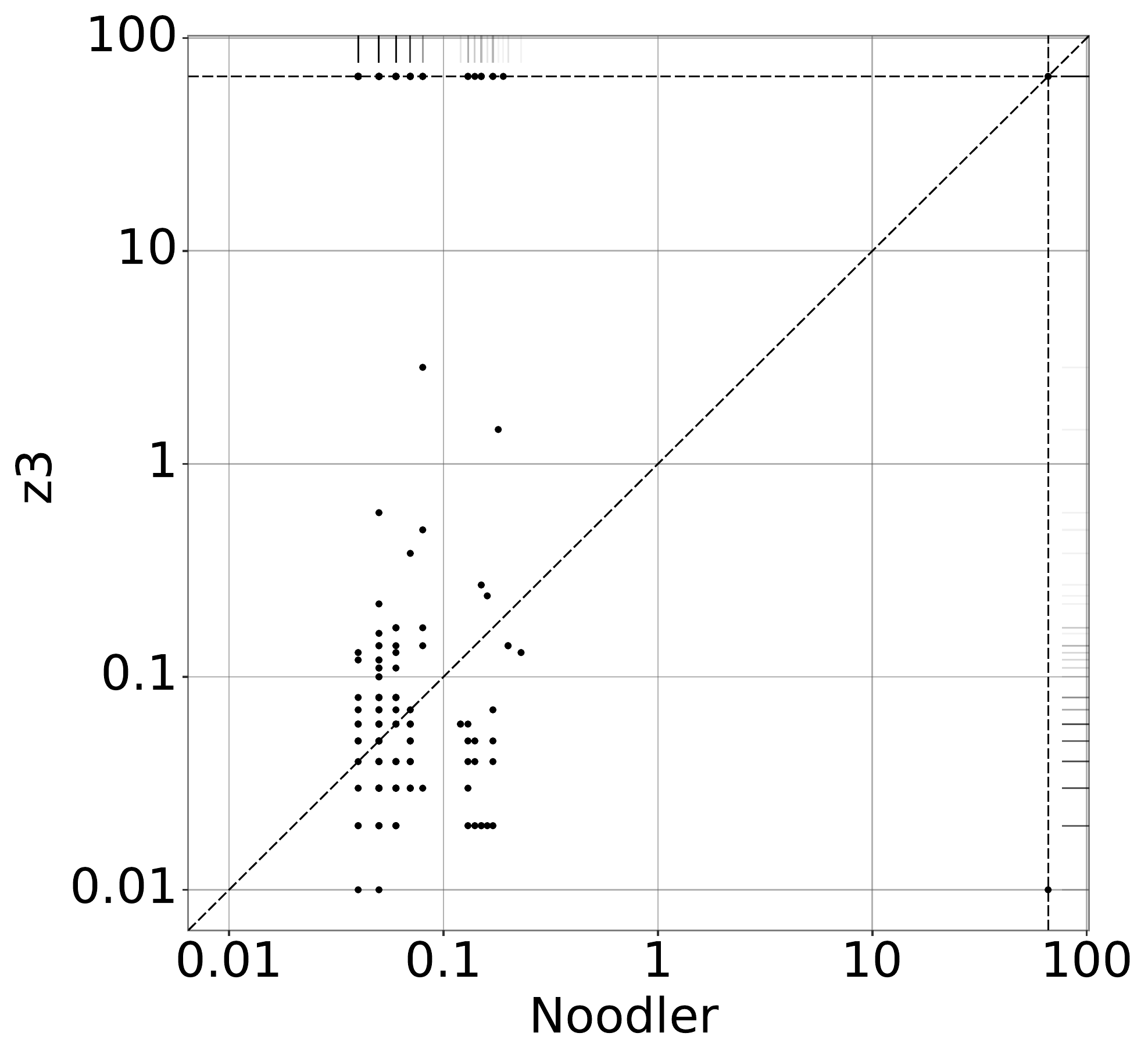}

    \includegraphics[width=0.32\textwidth]{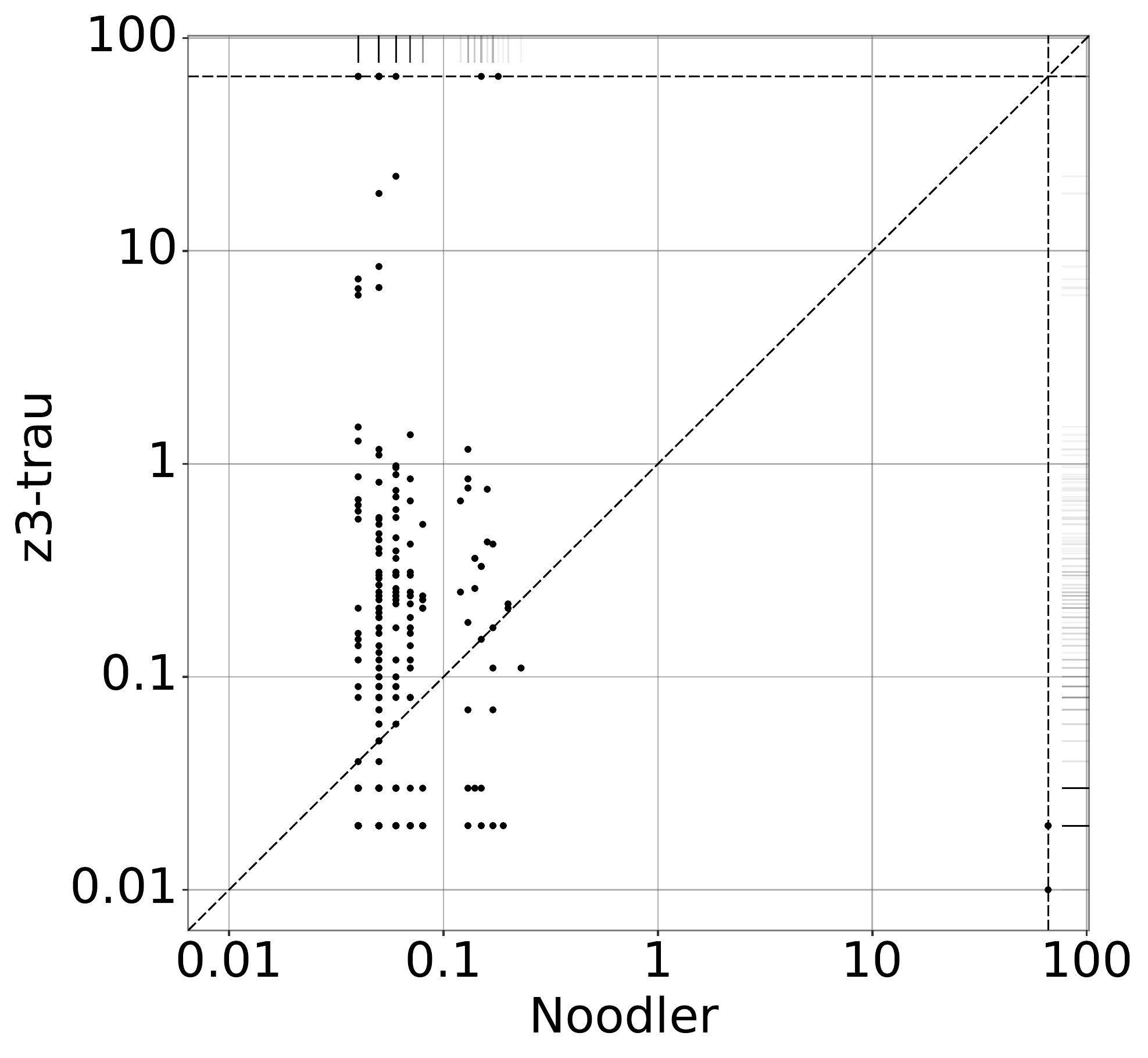}
    \includegraphics[width=0.32\textwidth]{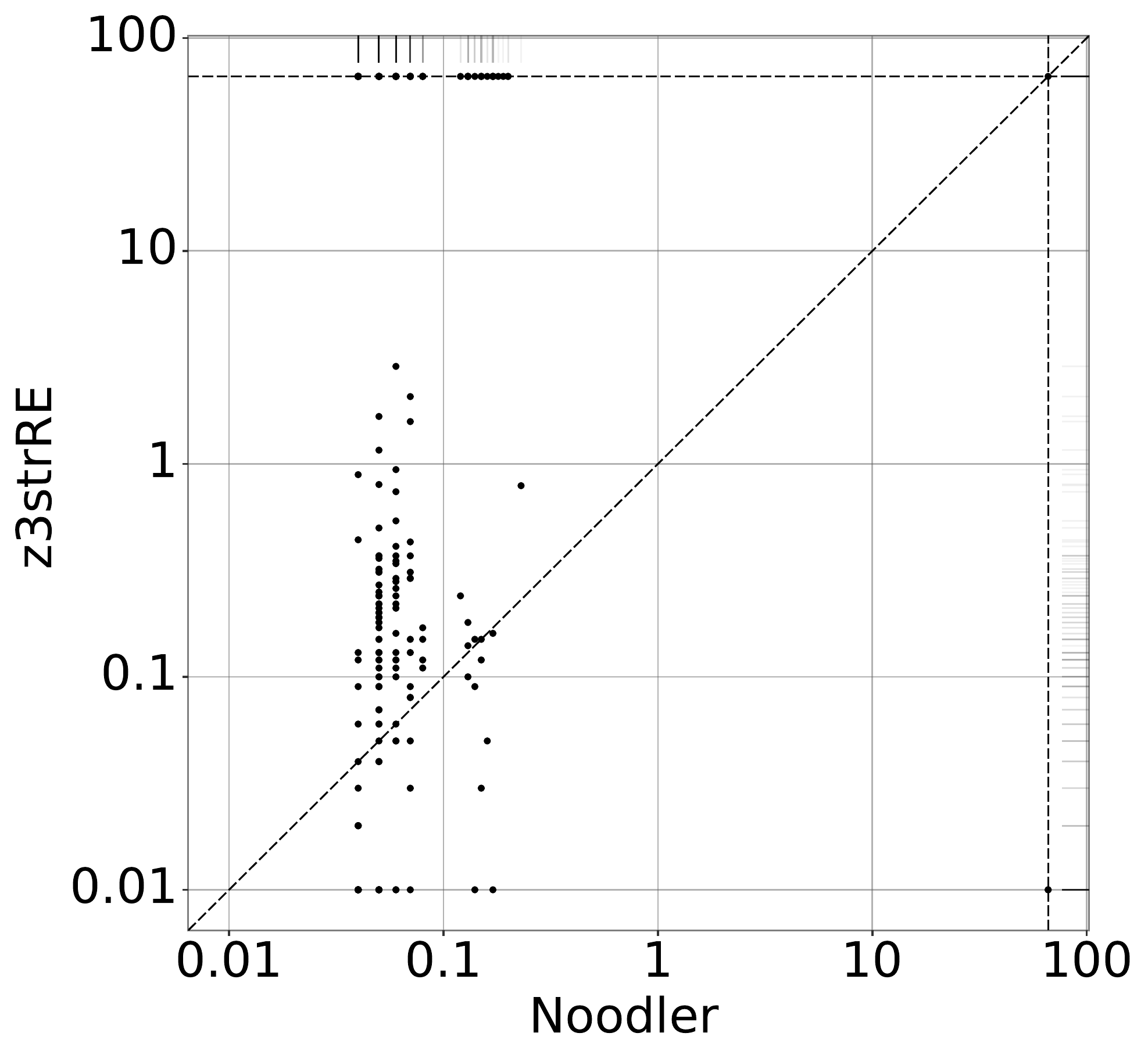}
    \includegraphics[width=0.32\textwidth]{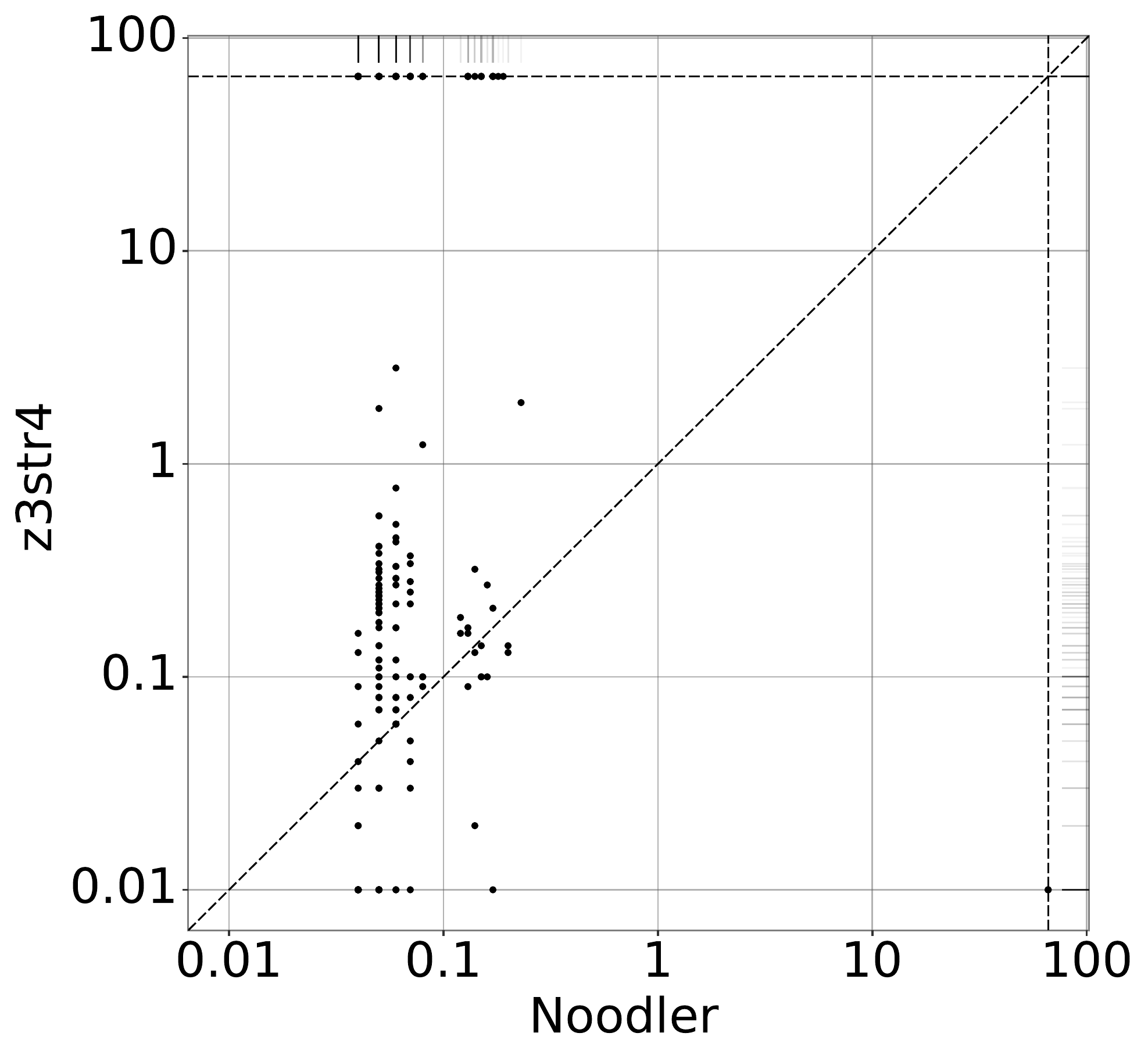}

    \includegraphics[width=0.32\textwidth]{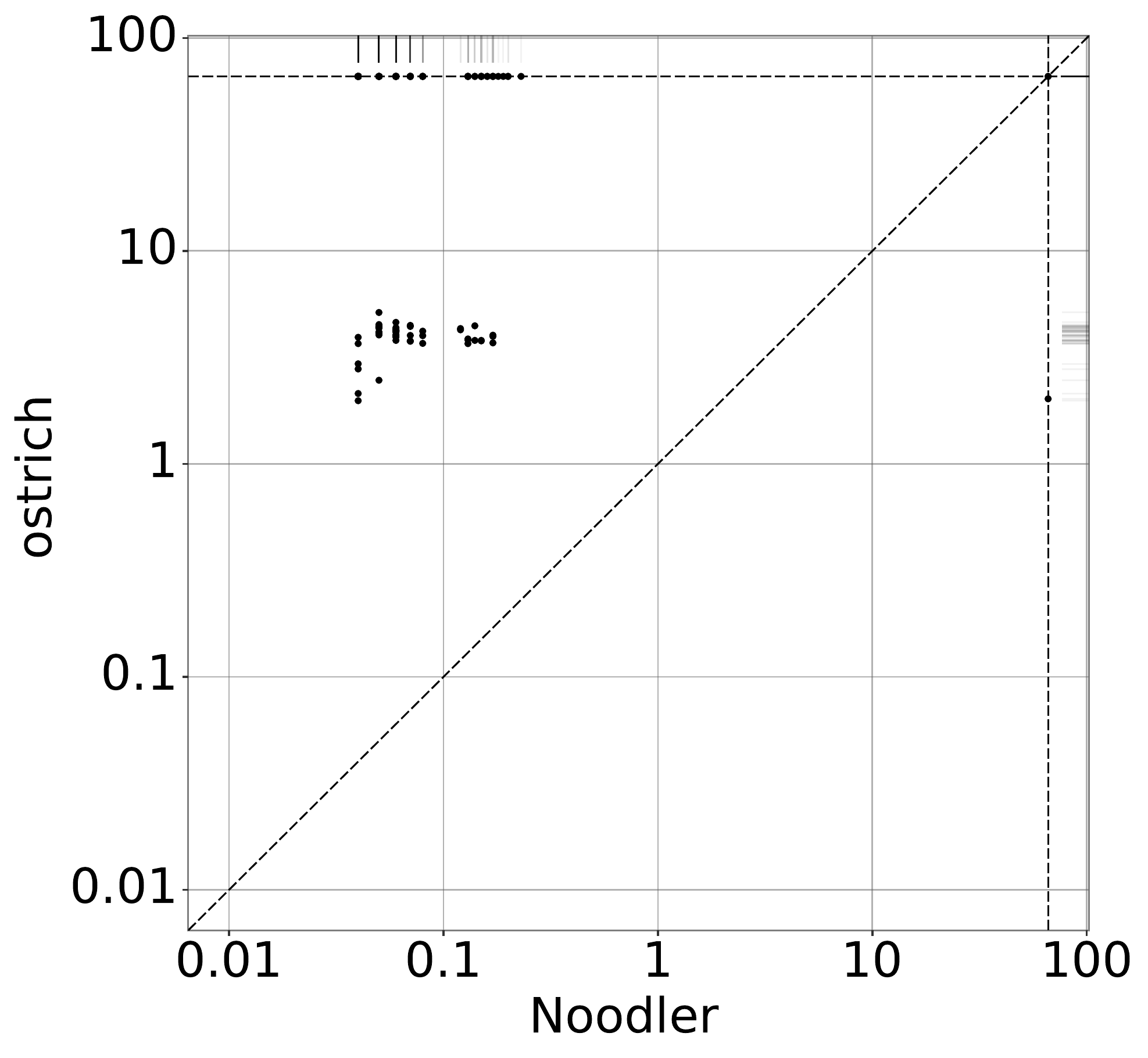}
    \includegraphics[width=0.32\textwidth]{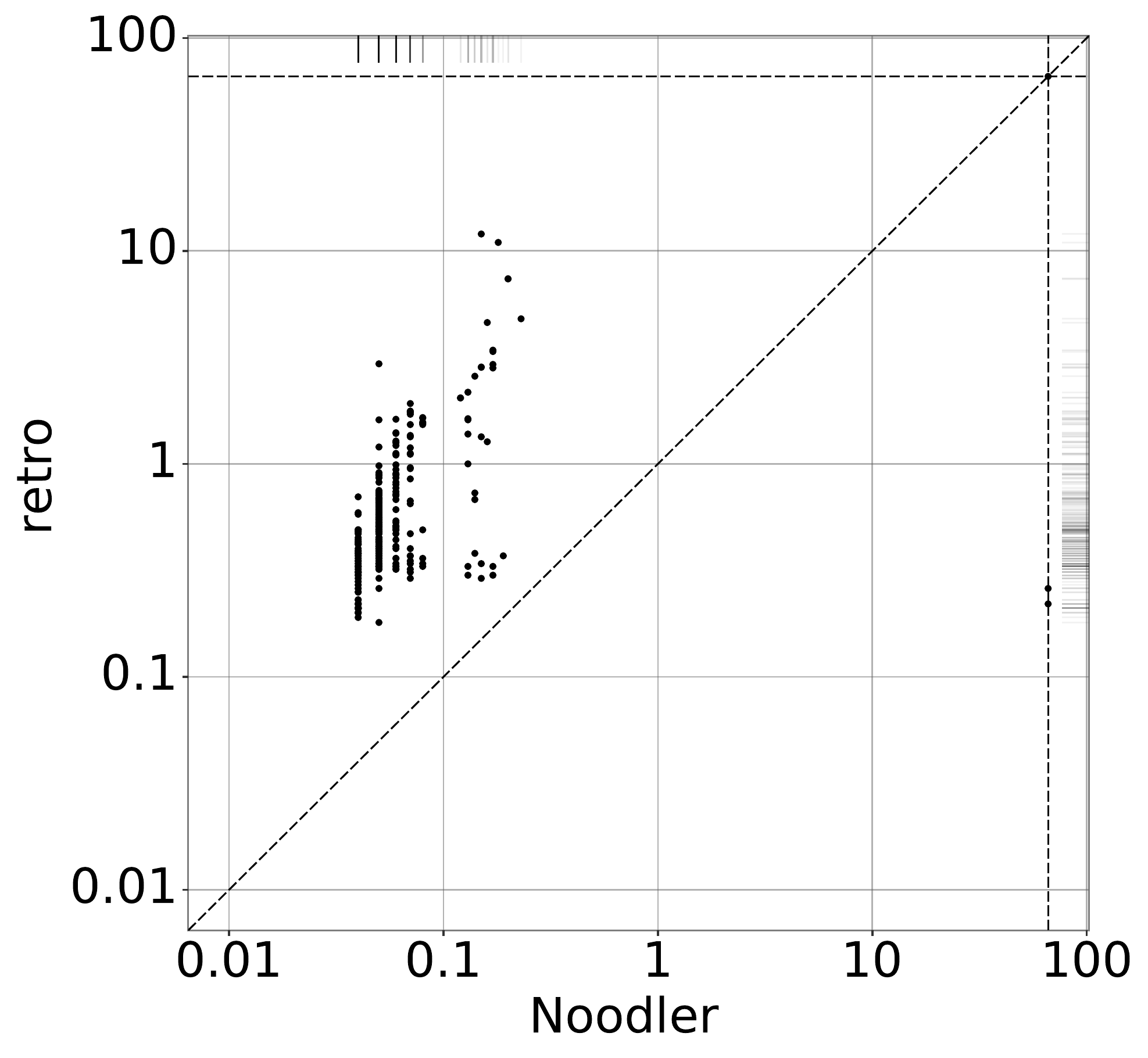}

    \caption{Comparison of run times of \noodler with all other tools (in seconds) on \strii. 
    Dashed lines represent timeouts.}
    \label{fig:str-all}
  \end{figure}

  \begin{figure}
    \centering
    \includegraphics[width=0.32\textwidth]{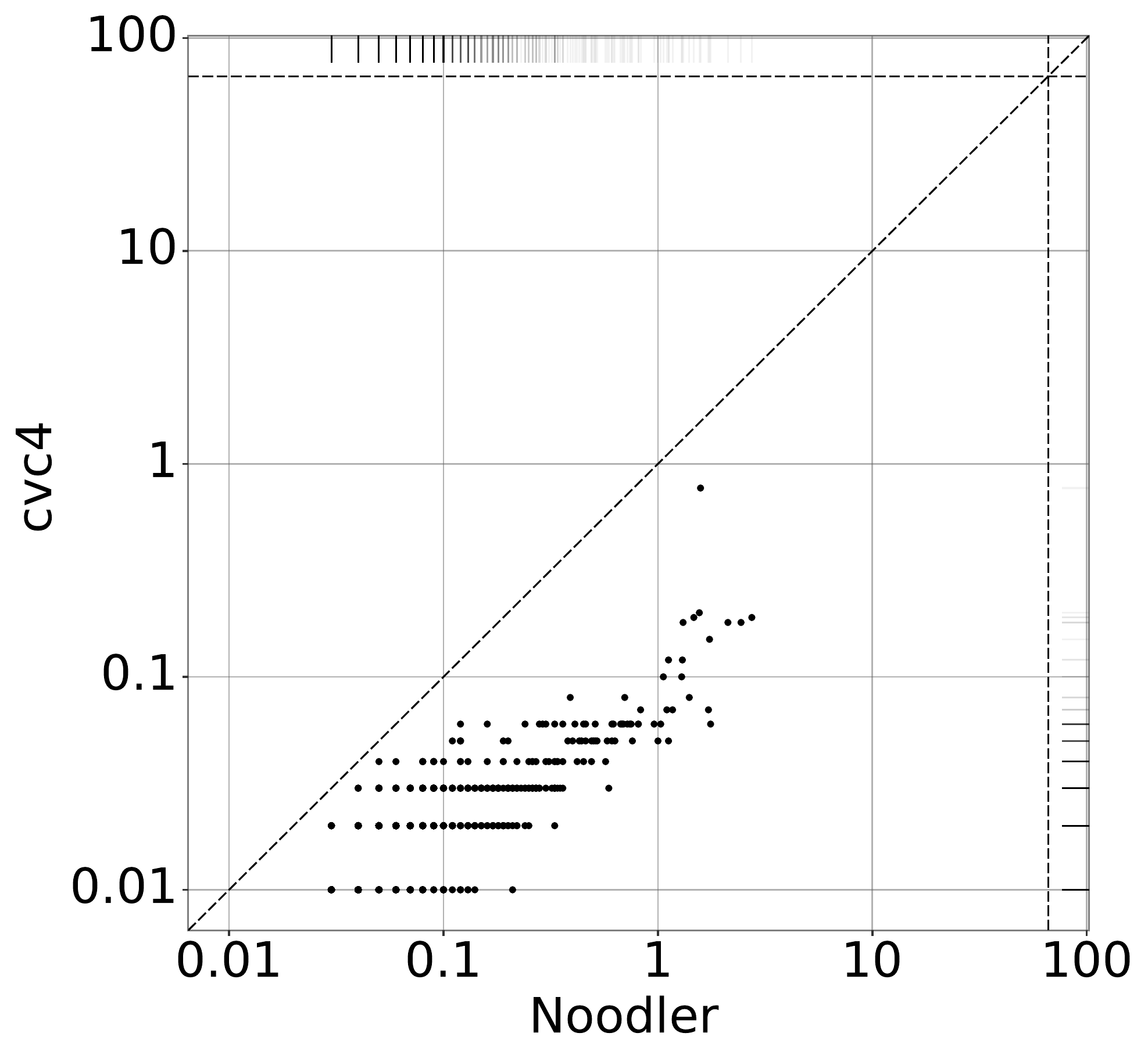}
    \includegraphics[width=0.32\textwidth]{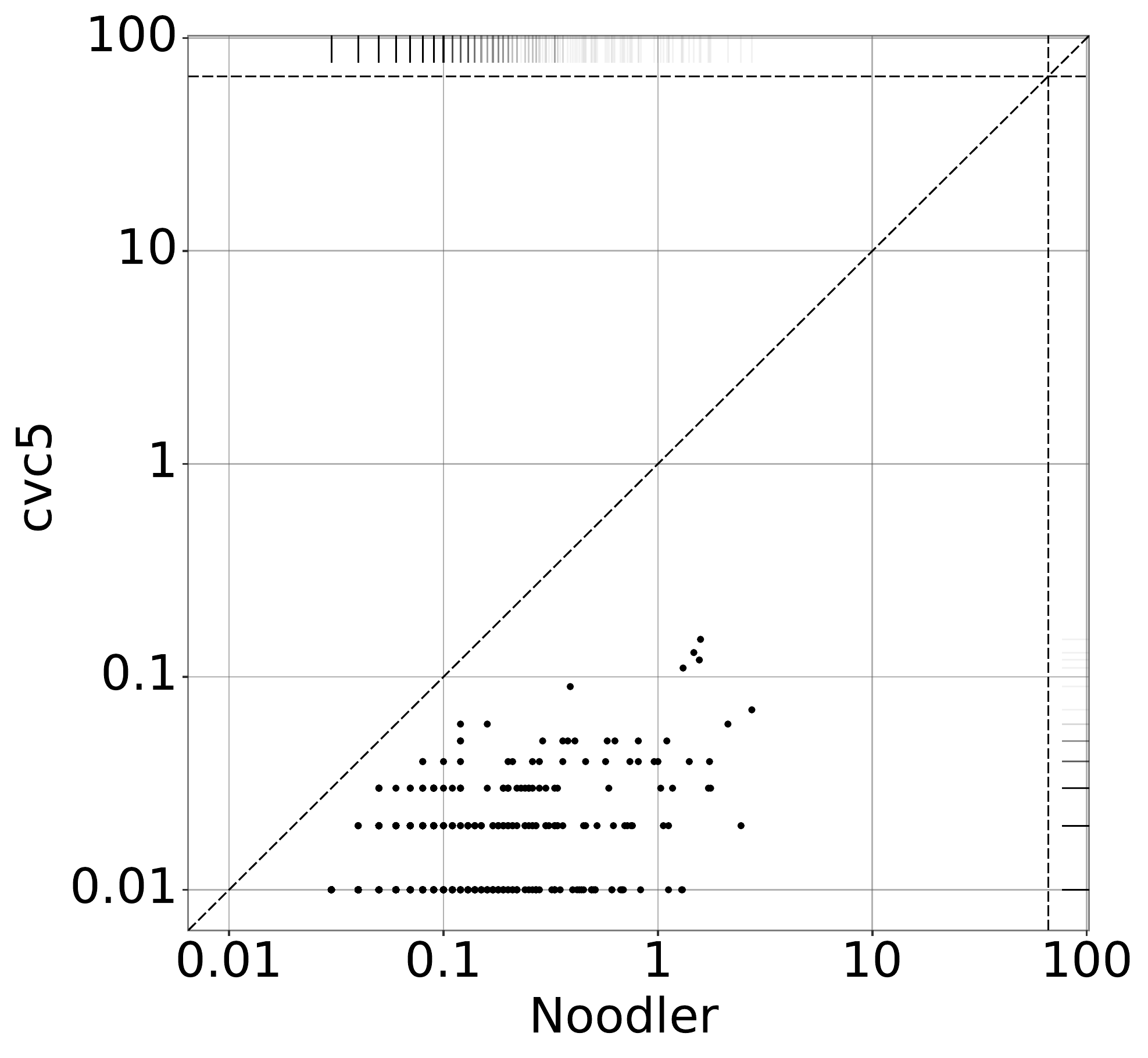}
    \includegraphics[width=0.32\textwidth]{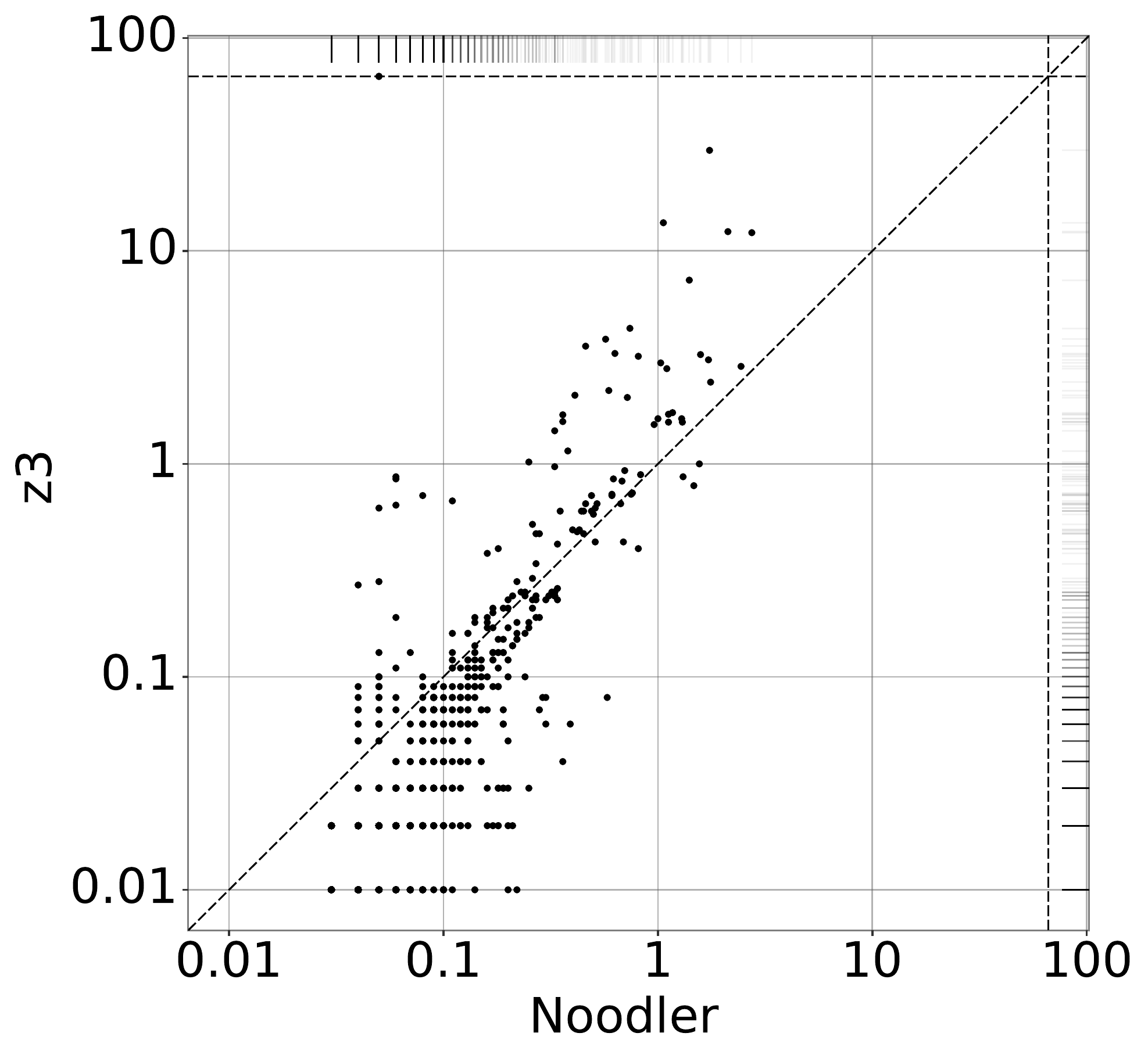}

    \includegraphics[width=0.32\textwidth]{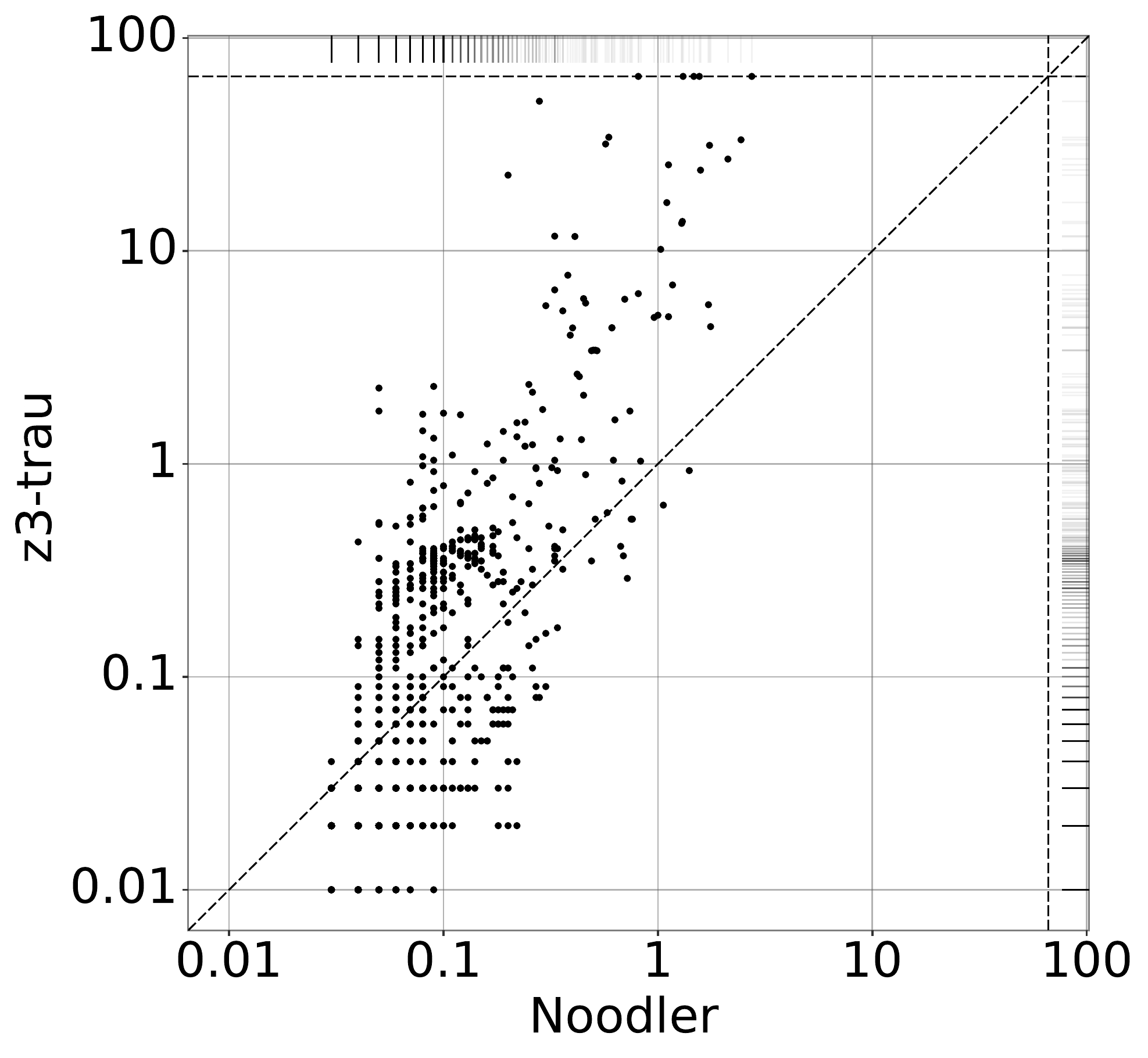}
    \includegraphics[width=0.32\textwidth]{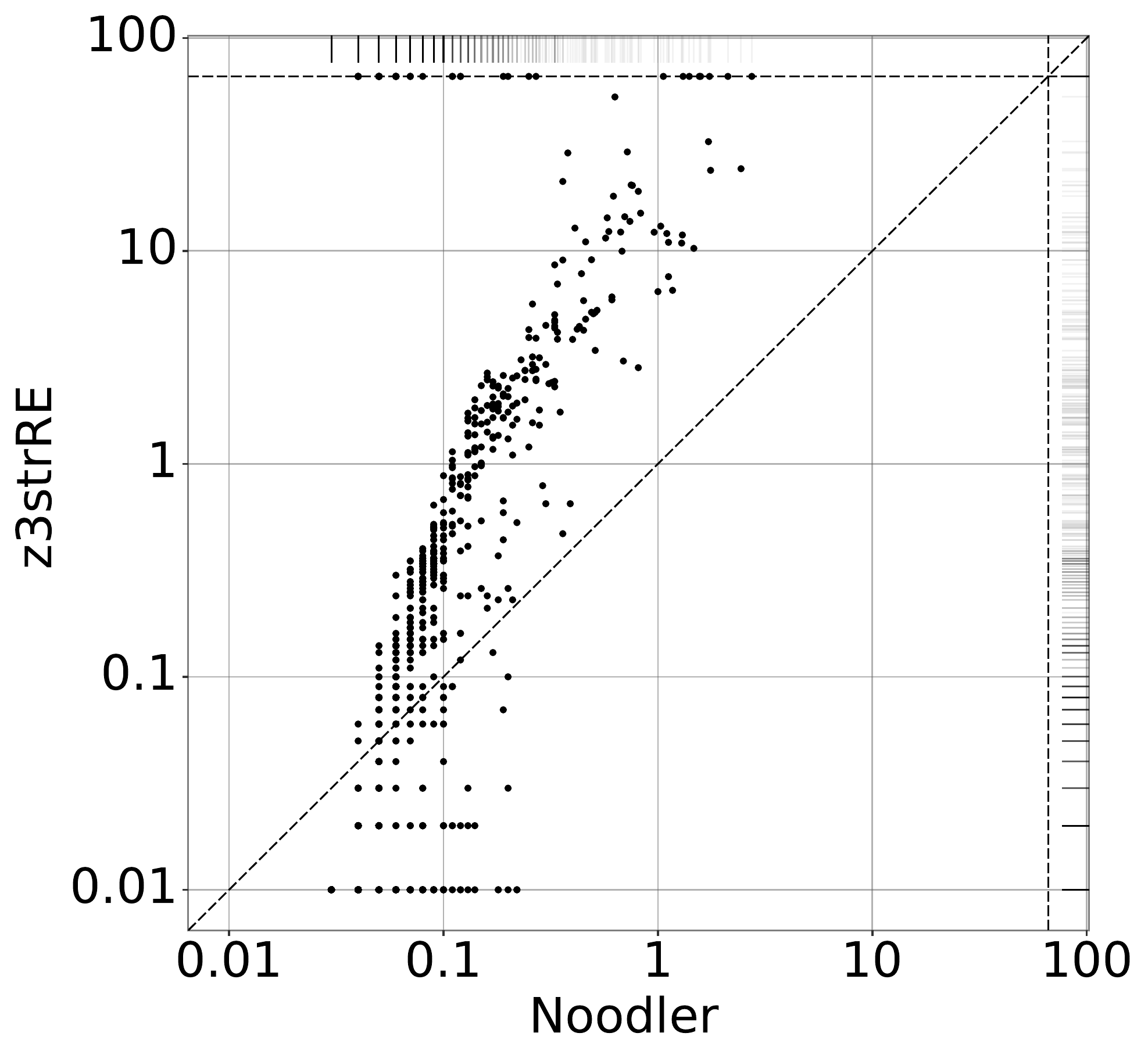}
    \includegraphics[width=0.32\textwidth]{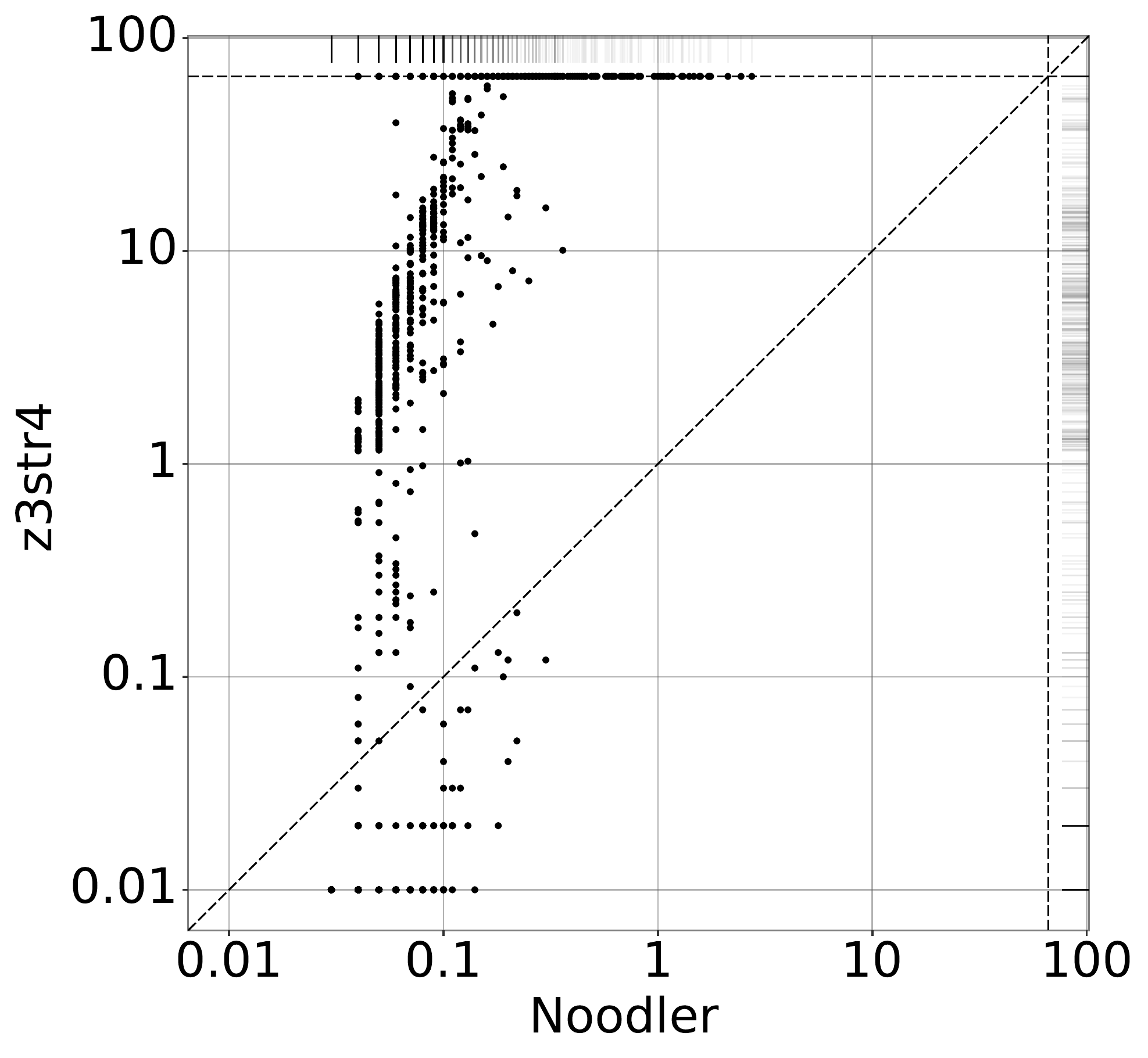}

    \includegraphics[width=0.32\textwidth]{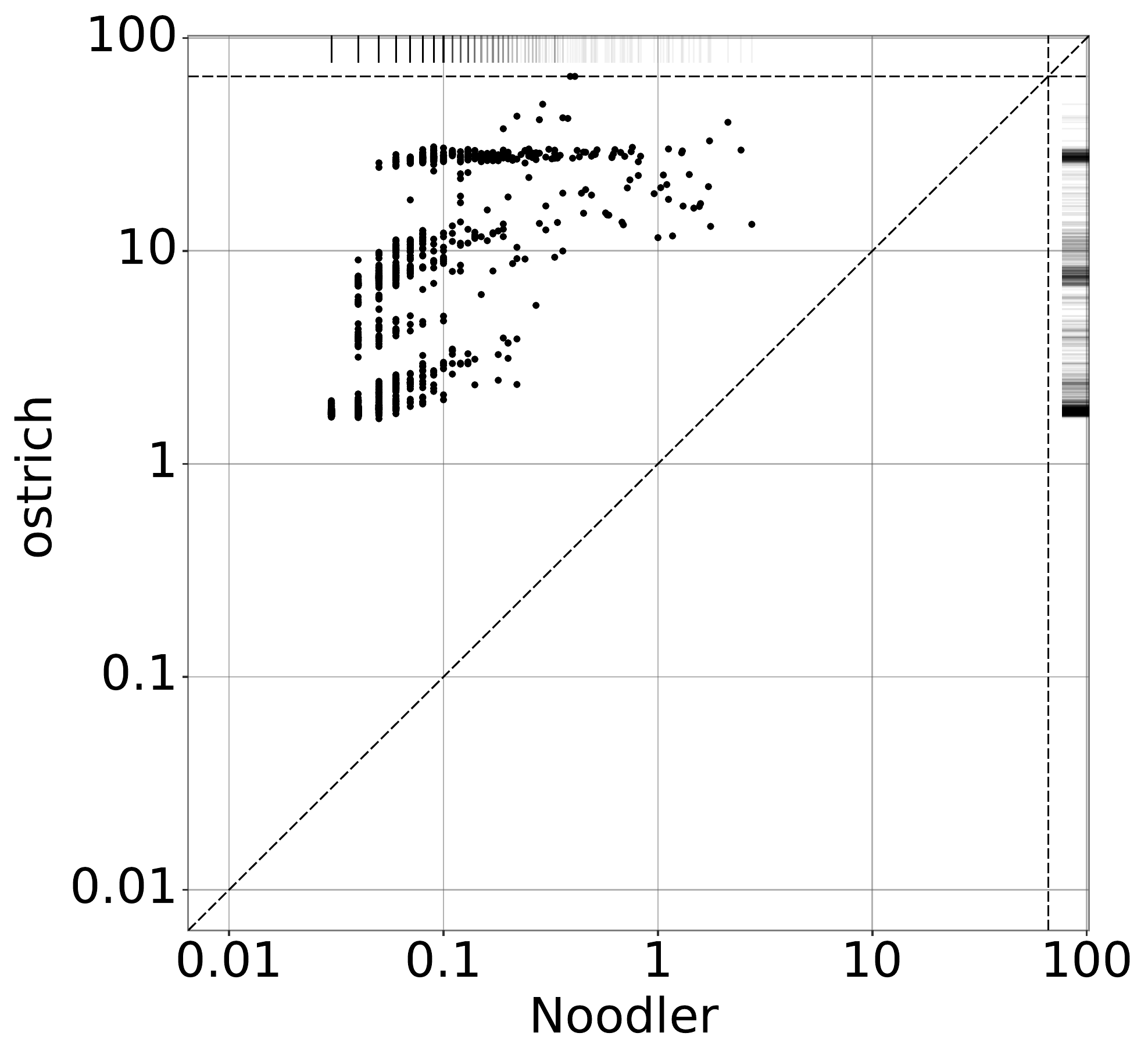}
    \includegraphics[width=0.32\textwidth]{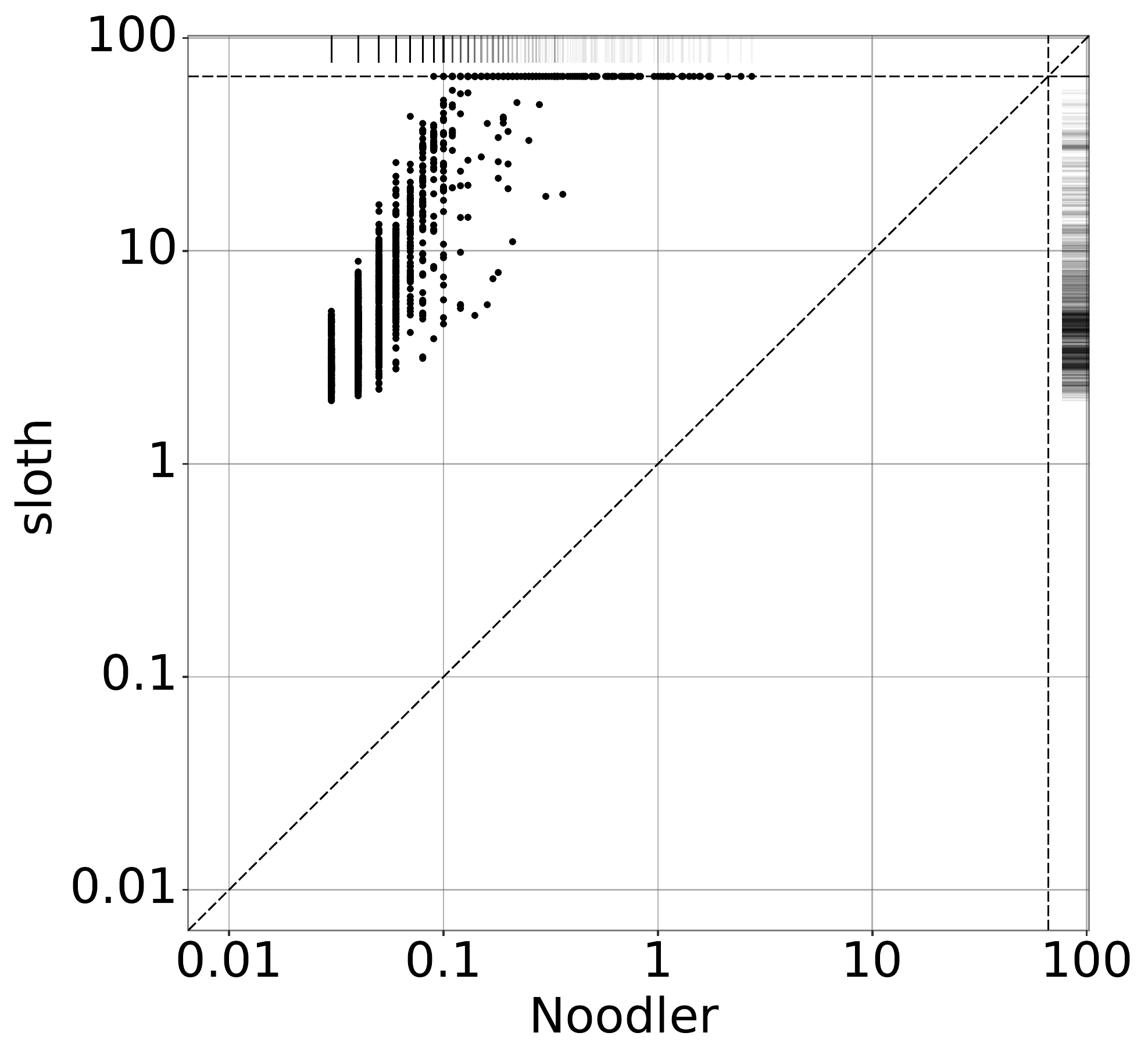}

    \caption{Comparison of run times of \noodler with all other tools (in seconds) on \slog. 
    Dashed lines represent timeouts.}
    \label{fig:slog-all}
  \end{figure}

\end{document}